%% file: main.tex
\documentclass[sigplan,screen,nonacm]{acmart}
\usepackage[T1]{fontenc}
\usepackage{amsfonts}
\usepackage{amsmath}

\usepackage{amssymb}

\usepackage{style}

\input{macro.tex}

\newif\ifsubmission
\newcommand{\appref}[1]{%
  \ifsubmission
    the full paper%
  \else
    #1%
  \fi
}
\begin{document}

\title{UFO Trees: Practical and Provably-Efficient Parallel Batch-Dynamic Trees}

\author{Quinten De Man}
\orcid{0009-0009-1321-9590}
\affiliation{%
  \institution{University of Maryland}
  \city{College Park}
  \country{USA}
}
\email{deman@umd.edu}

\author{Atharva Sharma}
\orcid{0009-0008-4469-1997}
\affiliation{%
  \institution{University of Maryland}
  \city{College Park}
  \country{USA}
}
\email{asharm07@umd.edu}

\author{Kishen N Gowda}
\orcid{0000-0001-6573-9445}
\affiliation{%
  \institution{University of Maryland}
  \city{College Park}
  \country{USA}
}
\email{kishen19@umd.edu}

\author{Laxman Dhulipala}
\orcid{0000-0003-0685-064X}
\affiliation{%
  \institution{University of Maryland}
  \city{College Park}
  \country{USA}
}
\email{laxman@umd.edu}

%% Remove these two lines for the camera-ready version
\settopmatter{printacmref=false} % Removes ACM Reference format
\fancyhead{}  % clears the header fields that trigger HotCRP format checker
% \renewcommand\footnotetextcopyrightpermission[1]{} % removes the ACM copyright footnote

%%
%% The abstract is a short summary of the work to be presented in the
%% article.
\begin{abstract}
\input{00_abstract}

\end{abstract}

\maketitle

\input{01_intro}
\input{02_prelims}

%\clearpage
\input{03_topology}
%\clearpage
\input{04_ufo}
%\clearpage
\input{05_parallel}

\input{06_experiments}
\input{07_conclusion}

\clearpage

\begin{acks}
This work is supported by NSF grants CCF-2403235 and CNS-2317194. We thank the anonymous reviewers for their useful comments.
\end{acks}

%%
%% The next two lines define the bibliography style to be used, and
%% the bibliography file.
\bibliographystyle{ACM-Reference-Format}
\bibliography{references}

\clearpage

\begin{appendix}
\input{AA_prelims}
\input{AB_analysis}
\input{AC_queries}
\input{AD_experiments}
\end{appendix}

\end{document}

%% file: macro.tex
  \newcommand{\fname}[1]{\textsc{#1}}
  \newcommand{\vname}[1]{\textit{#1}}

  \newcommand{\revision}[1]{{#1}}

\definecolor{forestgreen}{rgb}{0.13, 0.55, 0.13}

\newcommand{\etal}{\text{et al}.\xspace}

\newcommand{\ifconference}{{{\ifx\fullversion\undefined}}}

\makeatletter
\def\dfnt@space@setup{%
\dfnt@preskip=\parskip
  \dfnt@postskip=0pt}
\makeatother

\newtheoremstyle{exampstyle}
{.05in} % Space above
{.05in} % Space below
{} % Body font
{.5em} % Indent amount
{\sc \bfseries} % Theorem head font
{.} % Punctuation after theorem head
{.5em} % Space after theorem head
{} % Theorem head spec (can be left empty, meaning `normal')
\theoremstyle{exampstyle} 
\theoremstyle{exampstyle} 

\makeatletter
\renewenvironment{proof}[1][\proofname]{\par
\vspace{-\topsep}% remove the space after the theorem
\pushQED{\qed}%
\normalfont
\topsep0pt \partopsep0pt % no space before
\trivlist
\item[\hskip\labelsep
      \itshape
  #1\@addpunct{.}]\ignorespaces
}{%
\popQED\endtrivlist\@endpefalse
\addvspace{3pt plus 3pt} % some space after
}

\newcommand{\diam}{\ensuremath{D}\xspace}
\newcommand{\domain}{\ensuremath{\mathcal{D}}\xspace}

 \crefname{section}{Sec.}{Sec.}
 \crefname{theorem}{Thm.}{Thm.}
 \crefname{lemma}{Lem.}{Lem.}
 \crefname{corollary}{Col.}{Col.}
 \crefname{table}{Tab.}{Tabs.}
 \crefname{algorithm}{Alg.}{Algs.}
 \Crefname{table}{Tab.}{Tabs.}

\newtheorem{definition}{Definition}

%% file: 00_abstract.tex
The dynamic trees problem is to maintain a tree under edge updates while supporting queries like connectivity queries or path queries.
Despite the first data structure for this fundamental problem---the link-cut tree---being invented 40 years ago, our experiments reveal that they are still the fastest sequential data structure for the problem.
However, link-cut trees cannot support parallel batch-dynamic updates and have limitations on the kinds of queries they support.

In this paper, we design a new parallel batch-dynamic trees data structure called UFO trees that simultaneously supports a wide range of query functionality, supports work-efficient parallel batch-dynamic updates, and is competitive with link-cut trees when run sequentially.
We prove that a key reason for the strong practical performance of both link-cut trees and UFO trees is that they can perform updates and queries in sub-logarithmic time for low-diameter trees.
We perform an experimental study of our optimized C++ implementations of UFO trees with ten other dynamic tree implementations, several of which are new, in a broad benchmark of both synthetic and real-world trees of varying diameter and size.
Our results show that, in both sequential and parallel settings, UFO trees are the fastest dynamic tree data structure that supports a wide range of queries.
Our new implementation of UFO trees has low space usage and easily scales to billion-size inputs, making it a promising building block for implementing more complex dynamic graph algorithms in practice.

%% file: 01_intro.tex
\section{Introduction}

The \emph{dynamic trees problem} is one of the most basic and fundamental dynamic graph problems, and has been extensively studied since its introduction by Sleator and Tarjan~\cite{sleator1983data}.
The problem is to maintain a data structure representing a forest (i.e., a collection of trees) while supporting edge insertions, edge deletions, and vertex connectivity queries.
\revision{The dynamic trees and batch-dynamic trees problems are some of the most fundamental problems in parallel batch-dynamic graph algorithms, and are the basis for many more complex dynamic and static graph algorithms such as connectivity~\cite{holm2001poly,acar2019batchconnect}, minimum spanning forest~\cite{holm2001poly,tseng2022parallel}, hierarchical agglomerative clustering~\cite{dhulipala2024optimal, deman2025fully}, and more.}
Sleator and Tarjan first solved the dynamic trees problem with the celebrated \emph{link-cut tree }~\cite{sleator1983data} data structure.
In the intervening years, a wide range of dynamic trees have been designed to enable different applications, queries, and tradeoffs including \emph{Euler tour trees}~\cite{henzinger1995randomized}, \emph{rake-compress trees}~\cite{acar2004dynamizing,acar2005experimental}, \emph{topology trees}~\cite{frederickson1985data,frederickson1997ambivalent,frederickson1985data}, and \emph{top trees}~\cite{backstrom2006group}.
However, surprisingly, our experimental results show that none of these data structures can match the sequential performance of a carefully optimized implementation of link-cut trees.
%\footnote{In other settings, such as the insertion-only setting or the batch-parallel setting, other data structures have been shown to outperform link-cut trees.}.
%This has been observed in prior experimental studies~\cite{}, and is further demonstrated in this paper.
%\laxman{Should we talk about applications of dynamic trees somewhere early on? Why should the PPoPP reviewer care about this problem?}

Unfortunately, link-cut trees have some significant limitations: they can only support connectivity queries and path queries, and they do not support parallel operations.
In contrast, more recent dynamic tree data structures support a wider range of queries and can process batches of operations in parallel.
In Table~\ref{tab:dynamic_trees_summary}, we summarize the operations supported by existing dynamic tree data structures and the costs of operations.
Among these, rake-compress trees~\cite{acar2004dynamizing, acar2005experimental} support the broadest suite of queries, to the best of our knowledge. 
These include connectivity queries, path queries, subtree queries, least common ancestor queries, and many others. 
Such queries are essential in many applications of dynamic tree data structures~\cite{sleator1983data, auletta1996dynamic, radzik1998implementation, cheng1996isomorphism, frederickson1997ambivalent, holm2001poly, kapron2013dynamic}.

For batch operations, recent work on Euler tour trees~\cite{henzinger1995randomized, tseng2019batch} and rake-compress trees~\cite{anderson2023parallel, anderson2024deterministic} has designed work-efficient and poly-logarithmic depth parallel \emph{batch-dynamic} versions of these structures supporting both batch updates and queries. 
Parallel batch-dynamic data structures can support rapid streams of updates by
leveraging multiple cores and have seen a flurry of recent work due to 
their use in designing scalable static and dynamic algorithms~\cite{GT24, acar2019batchconnect, yesantharao2021parallel,acar2020parallel, dhulipala2019parallel, ghaffari2025, BB25a, wang2020closest}.

Experimentally, parallel Euler tour trees (ETTs)~\cite{tseng2019batch} are the fastest parallel batch-dynamic tree data structure available today. 
However, ETTs can only support connectivity queries and subtree queries, making them unusable for applications requiring more complex query support.
%
%Implementing parallel batch-dynamic trees that also support path queries was remained a challenging open problem---until recently.
%
A very recent implementation of parallel rake-compress (RC) trees~\cite{ikram2025parallel} was the first to support path queries and many other crucial queries in the parallel batch-dynamic setting. 
However, our benchmarking results show that experimentally its performance is significantly worse than that of parallel ETT's.
A major reason for this is that RC trees perform \emph{ternarization}: since RC trees only support inputs with constant-bounded degree, the input tree must be mapped to a tree with degree at most three.
Ternarization introduces large overheads and involves additional costs to maintain dynamically.
\revision{Specifically, the ternarized tree may have up to $2$ times as many vertices as the original tree, and a single edge update to the original tree may require up to $7$ edge updates to the ternarized tree (see \appref{Appendix~\ref{app:ternarization} for more details}).}
An important question given this state of affairs is: \emph{can we design a parallel batch-dynamic tree data structure that supports a wide range of queries, and is very fast in practice?}

\begin{table*}[ht]
    \small
    \centering
    \caption{\small A summary of the existing dynamic tree data structures and the new theoretical results from this paper. Entries marked with an \defn{asterisk*} are contributions of this paper. Here $n$ is the number of vertices in the input tree and $\diam$ is the diameter of the input tree.}
    \begin{tabular}{|l|c c|c|c c|c c c|}
        \hline
         & \multicolumn{2}{c|}{\revision{Sequential} Cost} & Not & \multicolumn{2}{c|}{Parallel Operations} & \multicolumn{3}{c|}{Queries Supported} \\
         & Update & Query & Ternarized & Update & Query & Subtree & Path & Non-Local \\
         \hline
        Link-cut tree & $O(\min\{\log n,\diam^2\})^*$ & $O(\min\{\log n,\diam^2\})^*$ & $\checkmark$ & & & & $\checkmark$ & \\
        Euler tour tree & $O(\log n)$ & $O(\log n)$ & $\checkmark$ & $\checkmark$ & $\checkmark$ & $\checkmark$ & & \\
        Top tree & $O(\log n)$ & $O(\log n)$ & $\checkmark$ & & $\checkmark$ & $\checkmark$ & $\checkmark$ & $\checkmark$ \\
        Rake-compress tree & $O(\log n)$ & $O(\log n)$ & & $\checkmark$ & $\checkmark$ & $\checkmark$ & $\checkmark$ & $\checkmark$ \\
        Topology tree & $O(\log n)$ & $O(\log n)$ & & $\checkmark^*$ & $\checkmark^*$ & $\checkmark^*$ & $\checkmark$ & $\checkmark^*$ \\
        UFO tree$^*$ & $O(\min\{\log n,\diam\})^*$ & $O(\min\{\log n,\diam\})^*$ & $\checkmark^*$ & $\checkmark^*$ & $\checkmark^*$ & $\checkmark^*$ & $\checkmark^*$ & $\checkmark^*$ \\
        \hline
    \end{tabular}
    \label{tab:dynamic_trees_summary}
\end{table*}

%\subsection{Summary of Results}

In this paper, we give a positive answer to this question by designing a new dynamic tree data structure called \defn{unbounded fan-out trees (UFO trees)}.
UFO trees can support all of the queries that rake-compress trees can, which to our knowledge covers all queries described in the literature for dynamic tree data structures.
Additionally, UFO trees are highly parallelizable and admit batch update algorithms that are work-efficient and have poly-logarithmic depth.

\myparagraph{UFO Trees}
UFO trees are based on parallel tree contraction~\cite{miller1989parallel} and map the input forest to a balanced tree structure by iteratively {\em clustering} the input until each tree in the forest contains a single cluster.
%
%UFO trees maintain an instance of parallel tree contraction over the input. 
%
Initially, the vertices of the input tree are the input clusters.
In each round, the remaining clusters are {\em merged} with one or more neighboring cluster, forming a cluster in the next round.
In particular, UFO trees allow any degree 3 or higher node to merge with {\em all} of its degree 1 neighbors. 
The remaining degree 1 and degree 2 nodes are merged along a maximal matching.

% Because UFO trees allow high degree nodes to cluster with multiple degree 1 neighbors, they do not require ternarization, unlike other dynamic tree structures such as topology trees~\cite{frederickson1985data,frederickson1997ambivalent,frederickson1985data} or RC trees~\cite{acar2004dynamizing, acar2020parallel} also based on tree contraction.
\revision{In other contraction-based dynamic tree data structures such as RC trees and topology trees, update efficiency is reliant on clusters being merged in pairs (at most 2 clusters per merge). This restriction necessitates ternarization: without it, inputs with high degree vertices would require each neighbor to be merged with the central vertex one by one, resulting in an unbalanced contraction tree.
UFO trees bypass this limitation by supporting the merging of high-degree clusters with multiple degree 1 neighbors in a single round.}

%\footnote{Ternarization transforms an input tree  containing vertices with degree greater than 3 to an tree containing vertices of degree at most 3, while preserving properties such as connectivity and enabling queries to still be performed on the transformed tree. We describe more details about ternarization in \appref{Appendix~\ref{todo}}.
%}

UFO trees represent the contraction process as a tree where the children of a node are the nodes in the previous round that clustered to form this node.
We show that the height of this tree is $O(\log n)$.
To perform updates efficiently, our update algorithm essentially deletes all nodes whose contents have changed, and then reclusters the tree bottom-up.
However, if we delete nodes with many children or neighbors, then the reclustering step of the update can become extremely costly.
We avoid this issue by showing that the tree can be correctly reclustered without needing to delete any nodes with more than 3 children or more than 3 neighbors.
We also present work-efficient and poly-logarithmic depth algorithms for batch-insertions and batch-deletions in UFO trees. 
Our update bounds are work-efficient (i.e., have work matching that of the best sequential algorithms) in the worst case, but can be asymptotically better in the case when the tree diameter is small, as we discuss next.
%Specifically, the work and depth of batch-insertions and batch-deletions with batch size $k$ are $O(\min\{ k \log (1+n/k), kd \})$ and $O(\log n \log k)$ respectively in the binary-fork join model~\cite{blelloch2020optimal}. These algorithms are theoretically as efficient as the best known algorithms, and better for low diameter inputs.

\myparagraph{Parametrizing by Tree Diameter}
We prove that both UFO trees and link-cut trees have update and query complexity that can be expressed in terms of the tree diameter.
In particular, UFO tree updates can be sequentially performed in $O(\min\{\log n, \diam\})$ time, and link-cut tree updates can be sequentially performed in $O(\min\{\log n, \diam^2\})$ amortized time.
Our experimental results show that this asymptotic result actually has an impact in practice, and helps explain the strong performance of both trees compared to existing dynamic tree data structures.
%Link-cut trees perform updates and queries significantly faster on input forests with low diameter in our experiments.
%One contribution of this paper is explaining their strong performance by proving that the costs of operations in link-cut trees can always be bounded by $O(\diam^2)$ where $\diam$ is the diameter of the input forest. Since link-cut trees also have a worst-case bound of $O(\log n)$ per-operation, the cost of operations can be expressed as $O(\min\{\log n, \diam^2\})$.
We note that many applications of dynamic trees in practice may have low diameter forests as an input (e.g., consider breadth-first search trees on web graphs or social network graphs).
%For example, state-of-the-art fully-dynamic graph connectivity systems~\cite{} heuristically maintain a low diameter spanning forest of the graph. Additionally, we expect breadth-first search trees or minimum spanning trees of certain types of real world graphs to have low diameter.

%\myparagraph{Excellent Sequential Performance}
%The update and query speed of UFO trees is always within a small constant factor of that of link-cut trees on the same input.
%Additionally, UFO trees are faster than all the other data structures that we evaluated other than link-cut trees: rake-compress trees, Euler tour trees, topology trees, and top trees.
%UFO trees also use less memory than the other data structures, and almost as little memory as link-cut trees.

Other existing dynamic tree data structures do not obtain asymptotically faster performance on low diameter inputs, and our experimental study shows that their performance in practice is relatively stable regardless of the input.
Because of this, the performance gap between the diameter-bounded trees (link-cut trees and UFO trees) and the other data structures can become very large on low diameter inputs.
%On the other hand, our new data structure, UFO tree, achieves the same cost bound of $O(\min\{\log n, \diam\})$ for all of its operations, and thus exhibits similar performance benefits on low diameter inputs, staying within a small constant factor of the speed of link-cut trees.

\myparagraph{Full Query Support}
Finally, we show how to modify UFO trees to support (to the best of our knowledge) all of the known queries offered over dynamic trees, matching the query capabilities of RC trees~\cite{acar2004dynamizing,acar2005experimental}.
These include subtree queries, path queries, and several non-local queries.
Subtree queries and path queries return the result of some associative function applied over all the weights in some subtree or path of the input tree respectively.
The remaining non-local queries include least common ancestor queries, diameter queries, center queries, median queries, and distance to nearest marked vertex queries.
Importantly, for most of these queries and especially those commonly used in practice (e.g., path queries, subtree queries with invertible functions, etc.), the UFO tree algorithm can be easily modified to support the query while keeping the cost bound of $O(\min\{\log n, \diam\})$ per operation. UFO trees require $O(\log n)$ time per operation only for a small set of queries (e.g., subtree maximum queries, where there is a lower bound of $\Omega(\log n)$ time per operation even if the diameter is constant).

\myparagraph{Our Contributions}
The contributions of this paper are:
\begin{enumerate}[label=(\arabic*),topsep=0pt,itemsep=0pt,parsep=0pt,leftmargin=15pt]
\item A novel practical dynamic trees data structure, UFO trees, based on parallel tree contraction with unbounded fanout merges, and batch update algorithms for UFO trees that are work-efficient with respect to state-of-the-art dynamic tree data structures and low depth.

\item New query algorithms for UFO trees and topology trees that capture all known queries offered over dynamic trees in the literature.

\item An analysis of UFO trees and link-cut trees showing that they perform sub-logarithmic work when the input tree diameter is low.

\item An experimental study comparing our implementations of sequential and parallel UFO trees with 10 other dynamic trees, including new implementations of topology trees and RC trees. Our experimental results show significant improvements in terms of speed and space over existing state-of-the-art dynamic tree data structures that also support complex queries and batch updates.

\end{enumerate}

%% file: 02_prelims.tex
\section{Preliminaries}

\myparagraph{Parallel Model}
Our analysis of our parallel algorithms is based on
nested fork-join parallelism~\cite{CLRS, frigo1998implementation}.
In particular, we analyze our algorithms using a work-span model based on binary-forking~\cite{blelloch2020optimal}.
The \defn{work} of a parallel algorithm is the total number of operations, while  the \defn{depth} is the critical path length of its computational DAG.
Computations with work $W$ and depth $D$ can be executed using a randomized work-stealing
scheduler in practice in $W/P+O(D)$ time with high probability on $P$ processors~\cite{BL98,ABP01}.
We use $O(f(n))$ \emph{with high probability} (\whp{}) to mean $O(cf(n))$ with probability $\geq 1-n^{-c}$ for $c \geq 1$.

\myparagraph{Parallel Primitives}
Given a sequence $A$ of $k$ entries, each with a \emph{key}, \emph{semisort} reorders $A$ such that all entries with the same key are consecutive. Semisort runs in $O(k)$ expected work and $O(\log k)$ depth \whp{}~\cite{gu2015top}.
We also use \emph{parallel hash tables}~\cite{gil1991towards}, which support batches of $k$ updates to a table with $n$ elements in $O(k)$ work and $O(\log n)$ depth \whp{}. 
%This also supports a ``parallel for each element'' operation in $O(n)$ work and $O(\log n)$ depth.
%\quinten{Check that this is correct and in binary forking. Add appropriate cite.}

%The binary fork-join model extends the standard RAM model with a {\em fork} instruction, which forks a child thread, and a {\em join} instruction to join threads after forking them. All threads share a common memory. A computation starts with a single initial thread, and ends when an {\em end all} operation is called. The computation can be viewed as a directed acyclic graph (DAG), where each node represents an instruction, nodes point to the node for the next instructions, nodes for fork instructions point to two nodes, and nodes for join instructions are pointed to by the last instruction of the child thread. 

\myparagraph{Dynamic Trees Background}
The dynamic trees problem~\cite{sleator1983data} is to maintain a data structure representing a forest while supporting efficient edge insertions, edge deletions, and vertex connectivity queries.

All dynamic tree data structures can support connectivity queries, and many also support more advanced queries.
\emph{Path queries} and \emph{subtree queries} compute a commutative and associative function applied over the weights along a particular path or within a subtree, respectively.
\emph{Lowest common ancestor (LCA) queries} return the LCA of two vertices with respect to a designated root vertex.
\emph{Diameter queries} return the length of the longest path.
\emph{Center queries} return the vertex that minimizes the maximum distance to any other vertex.
\emph{Median queries} return the vertex that minimizes the sum of the weighted distances to all other vertices.
\emph{Nearest marked vertex queries} return the nearest marked vertex (or the distance from the nearest marked vertex) to a given vertex.

%\myparagraph{Dynamic Tree Data Structures}
Table~\ref{tab:dynamic_trees_summary} summarizes the known dynamic tree data structures and the operations they support.
In this paper, we show that our new dynamic trees structure, UFO trees, supports all of the aforementioned query operations. We also show that our query algorithms work on topology trees~\cite{frederickson1985data,frederickson1997data,frederickson1997ambivalent}.
%
%Due to space constraints, we give a detailed commentary on the related work presented in Table~\ref{tab:dynamic_trees_summary} in \appref{Appendix~\ref{TODO}}.

%\emph{Link-cut trees}~\cite{sleator1983data} were the first dynamic tree data structure, and only supported connectivity queries and path queries.
%\emph{Euler tour trees}~\cite{henzinger1995randomized} support subtree queries only.
%\emph{Rake-compress trees}~\cite{acar2004dynamizing} are based on parallel tree contraction, support all the queries listed above, and support batch-parallel operations as well.
%\emph{Top trees}~\cite{alstrup2003maintaining} and \emph{Splay Top Trees}~\cite{holm2023SplayTopTrees} also support all the queries listed above. \kishen{please check this} Unlike RC trees, they support unbounded degree tree inputs, and thus, do not require ternarization. However, they are not known to be parallelizable.\kishen{Anything else to add?}

%In this paper, we discuss \emph{topology trees}~\cite{frederickson1985data,frederickson1997data,frederickson1997data} in detail (Section~\ref{sec:topology}), and present our new data structure, \emph{UFO trees} (Section~\ref{sec:ufo}).

%\myparagraph{Parallel Batch-Dynamic Trees}
%\emph{Parallel Euler tour trees}
%\emph{Parallel rake-compress trees}
%
%In this paper we present \emph{parallel topology trees} and \emph{parallel UFO trees}.

\myparagraph{Ternarization}
Since topology trees and rake-compress trees are defined for constant-degree trees, \defn{ternarization} needs to be applied to the input tree before building a dynamic tree over it.
Ternarization takes a tree of arbitrary degree and converts it into a tree with maximum degree $\leq 3$.
This mapping can also be efficiently maintained in the batch-dynamic setting~\cite{ikram2025parallel}.
Ternarization does not affect the asymptotic costs of the data structure, but it can lead to high overheads in both space and time in practice.
We describe ternarization in detail in \appref{Appendix~\ref{app:ternarization}}.

%\myparagraph{Prior Experimental Study}
%

%% file: 03_topology.tex
\section{Topology Trees Revisited}\label{sec:topology}

%\quinten{Explain that this is a historical data structure that we re-present and simplify, we correct the algorithm, and we give new proofs.}

%\atharva{This was referring to the initial description of the topology tree:}\quinten{This writing is sort of awkward and redundant up to this point. Make the description more terse and clear. Figure. Reuse one figure a lot e.g. one input tree, a ternarized version, a topology tree on it, a UFO tree.}

Topology trees~\cite{frederickson1985data, frederickson1997ambivalent, frederickson1997data} were first described four decades ago and were the first dynamic tree structures based on parallel tree contraction.
Since UFO trees, our main contribution in this paper, are based on a similar type of tree contraction as topology trees, we provide a thorough and self-contained description of topology trees in this section.
We also give a new analysis of the cost of the update algorithm for topology trees, which is missing in the original papers.
Finally, we show how to make topology trees support path queries, subtree queries, LCA queries, diameter queries, center queries, median queries, and nearest marked vertex queries.

\subsection{Topology Tree Definition}\label{sec:topology_definition}

\begin{figure}[t]
    \centering
    \includegraphics[width=\columnwidth]{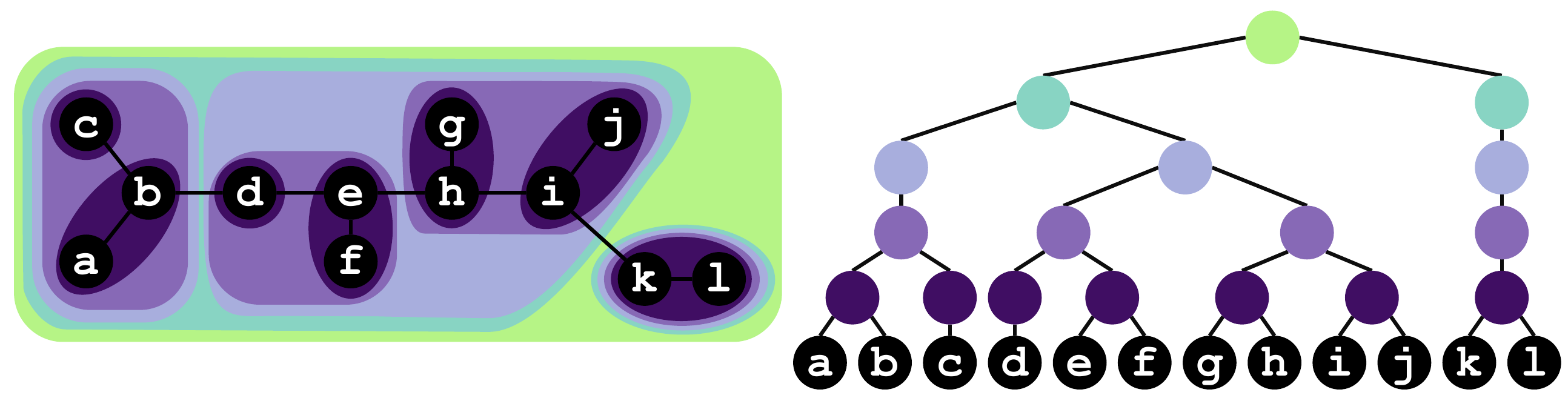}
    \caption{\small An example of an input tree (left) and a topology tree for it (right). \revision{The rightmost chain of nodes forms because the only neighbor of that cluster repeatedly merges with a different cluster.}}
    \label{fig:topology_example}
\end{figure}

\begin{figure*}[t]
    \centering
    \includegraphics[width=\textwidth]{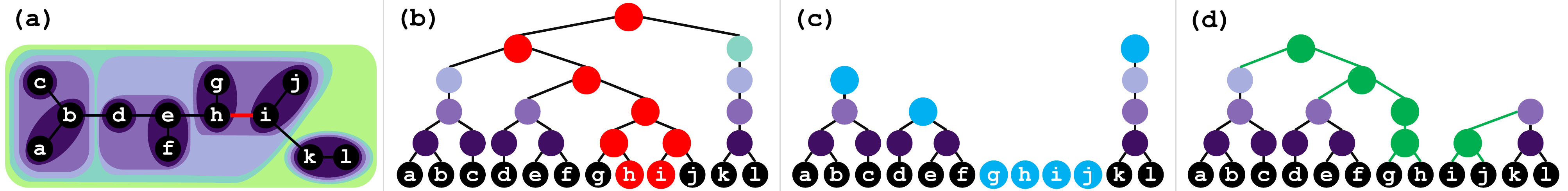}
    \caption{\small An example of an update to a topology tree. In (a), we show the input tree, with the deleted edge $(h,i)$ in red.
    In (b), we show the original topology tree with the ancestor clusters of $h$ and $i$ highlighted in red.
    In (c), we show the topology tree after the initial delete ancestors step with the root clusters highlighted in blue.
    In (d), we show the topology tree after the reclustering process with the newly created clusters highlighted in green.
    In this example, a root cluster ($i$'s parent) merges with a non-root cluster ($k$'s parent) to maintain a maximal matching. The ancestors of $k$'s grandparent are deleted.}
    \label{fig:topology_update}
\end{figure*}

Given an input tree $T = (V,E)$, a topology tree $\mathcal{T}$ is a rooted binary tree representing a bottom-up hierarchical clustering over the vertices of the input tree.
Figure~\ref{fig:topology_example} shows an example of an input tree and a possible topology tree for the input.
We refer to the nodes of $\mathcal{T}$ as \defn{clusters}.
The leaves of $\mathcal{T}$ represent each individual vertex in $V$, and we refer to them as the \defn{leaf clusters} or the \defn{level $0$ clusters}.
Each subsequent level of $\mathcal{T}$ represents a round of parallel tree contraction, and we refer to the set of clusters formed from the $i$-th round as the \defn{level $i$ clusters}.

The type of tree contraction used by topology trees works as follows: in each round, a maximal independent set of level $i$ cluster pairs are determined (i.e., a matching), and the clusters in each pair are \defn{merged} together along a common incident edge to form a new cluster at level $i+1$, with clusters not participating in a merge forming their own clusters.
Level $i$ edges along which clusters are merged do not appear in level $i+1$. 
Level $i$ edges that are not merged appear at level $i+1$ as an edge between the two level $i+1$ clusters formed from its level $i$ endpoint clusters.
The \defn{degree} of a level $i$ cluster is the number of level $i$ edges incident to that cluster.
Topology trees restrict the kinds of merges that are permitted in the tree contraction. In particular, the allowed merges are:
\begin{itemize}[itemsep=0pt,parsep=0pt,leftmargin=15pt,topsep=2pt]
    \item a degree 1 cluster with a degree 1 cluster
    \item a degree 1 cluster with a degree 2 cluster
    \item a degree 2 cluster with a degree 2 cluster
    \item a degree 1 cluster with a degree 3 cluster
\end{itemize}
These merges result in new degree 0, degree 1, degree 2, or degree 2 clusters, respectively.

For each level $i$ cluster in $\mathcal{T}$, its parent is the level $i+1$ cluster formed from it during tree contraction. 
The children of a level $i>0$ cluster are the one or two level $i-1$ clusters that merged to form it.
The \defn{fanout} of a level $i>0$ cluster is the number of level $i-1$ clusters that were merged to form it.
Note that for topology trees, the fanout of a cluster is always 1 or 2; for UFO trees we will relax this rule in the following sections.
Furthermore, any degree 3 cluster must have fanout 1, since none of the allowed merges produce clusters with degree 3.
Since each round merges a maximal independent set of cluster pairs, it can be shown that there is a geometric decrease in the number of clusters in each round, resulting in Theorem~\ref{thm:topology_tree_properties}. We provide the proofs in \appref{Appendix~\ref{app:topology_efficiency_proof}}.

\begin{theorem}\label{thm:topology_tree_properties}
    Topology trees have height $O(\log n)$, contain $O(n)$ total nodes, and use $O(n)$ space, where $n$ is the number of vertices in the input tree.
\end{theorem}

\subsection{Topology Tree Updates}\label{sec:topology_updates}

Here we provide a high-level overview of the update algorithm for topology trees. A detailed description and pseudo-code are provided in \appref{Appendix~\ref{app:topology_update_proof}}.
% The supplementary materials contain a detailed description and pseudo-code.
Figure~\ref{fig:topology_update} shows an example of an update to a topology tree.
The update algorithm first deletes all clusters in $\mathcal{T}$ that are ancestors of the leaf clusters corresponding to the endpoints of the edge insertion or deletion, and then \defn{reclusters} the resulting topology tree components in a bottom-up fashion. 
Deleting all the ancestors of some clusters in the topology tree results in a forest of disjoint topology tree components. We call the clusters at the roots of these components \defn{root clusters}.

The topology tree is reclustered bottom-up starting from the level $0$ root clusters.
If all the root clusters are degree $0$, the update is complete.
The algorithm attempts to match each level $i$ root cluster with one of its neighbors that it may merge with according to the allowed merges listed previously.
If the neighbor it matches with is also a root cluster, a new level $i+1$ cluster is created that is the parent of the two clusters.
If the neighbor it matches with is not a root cluster (we refer to these as \defn{non-root clusters}), it can only merge with the target neighbor cluster if it does not already merge with some other cluster. If this is the case, the root cluster can set its parent to its neighbors existing parent.

When a root cluster merges with a non-root cluster, it is necessary to remove the ancestors of the existing parent of the non-root cluster.
This is necessary because this new merge changes the vertices contained in the parent cluster and any of its ancestors, thus the merges they participate in at higher levels may be invalid or non-maximal (see example in Figure~\ref{fig:topology_update}).
Although paths of ancestors may be deleted many times, it can be shown that the number of root clusters at each level and the number of deleted clusters at each level are bounded by a constant, resulting in Theorem~\ref{thm:topology_tree_updates}. We provide the proofs in \appref{Appendix~\ref{app:topology_update_proof}}.

\begin{theorem}\label{thm:topology_tree_updates}
    Topology trees updates take $O(\log n)$ time, where $n$ is the number of vertices in the input tree.
\end{theorem}

\subsection{Topology Tree Queries}\label{sec:topology_queries}
Prior theoretical work~\cite{frederickson1985data} has shown how to use topology trees to answer path queries.
In this paper, we greatly expand the queries that topology trees can support. Our results are summarized in Theorem~\ref{thm:topology_tree_queries}. 
%\laxman{Let's make sure to say somewhere (here? earlier?) that prior to this result Topology trees were only known to support X, Y, Z queries.}
We give complete descriptions of how to support all of these queries in \appref{Appendix~\ref{app:topology_queries}}.
\revision{Crucially, all queries are read-only and any number of queries can be ran in parallel with no synchronization.}

\begin{theorem}\label{thm:topology_tree_queries}
    Topology trees can be modified to support path queries, subtree queries, LCA queries,  diameter queries, center queries, median queries, and nearest marked vertex queries.
\end{theorem}

%% file: 04_ufo.tex
\section{UFO Trees}\label{sec:ufo}

In this section, we describe \defn{unbounded fan-out trees (UFO trees)}, a new dynamic tree data structure that directly allows for input trees with vertices of unbounded degree while supporting path queries as well as a rich collection of other types of queries on trees.
Unlike many other dynamic tree data structures, UFO trees do not require ternarization, and thus, they benefit from better memory usage and faster update speeds.
We prove that the height of a UFO tree is $O(\min(\diam, \log n))$ which gives asymptotically better performance than other dynamic tree data structures for all operations on low diameter inputs.
Figure~\ref{fig:ufo_example} shows an example of an input tree and a possible UFO tree for the input.

UFO trees are based on a similar type of tree contraction as topology trees and represent the hierarchical clustering of the vertices in the same way: as a tree where the children of any node represent the clusters that merged to form this cluster.
In order to handle unbounded degree vertices, UFO trees differ from topology trees in the ways that clusters may be combined. The allowed merges are:
\begin{itemize}[itemsep=0pt,parsep=0pt,leftmargin=15pt,topsep=2pt]
    \item A degree 1 cluster with a degree 1 cluster
    \item A degree 1 cluster with a degree 2 cluster
    \item A degree 2 cluster with a degree 2 cluster
    \item A \emph{high degree} cluster with one or more degree 1 clusters
\end{itemize}
The main difference is this fourth merge rule.
We call a cluster with degree $\geq 3$ a \defn{high degree} cluster.
Allowing multiple degree 1 clusters to combine with high degree clusters prevents the issue of a high degree vertex having to combine with all of its neighbors one at a time which may prevent logarithmic height of the tree. 
%\quinten{This makes sense to me, but we haven't explained this issue prior to this so its sort of out of place.}
This property implies that there may exist clusters in the UFO tree with more than 2 children, which introduces some challenges with implementing updates (discussed in Section~\ref{sec:ufo_update}). 
%\laxman{We should not downplay the challenges introduced by this rule change so that our contribution is clear.}

% \laxman{The writing in the following section is choppy and should be polished; e.g., read the next two sentences.}
% \atharva{Addressed}
In each round of the tree contraction process represented by UFO trees, a maximal independent set of the allowed merges are performed.
Specifically, any high degree cluster must merge with all of its degree 1 neighbors.
For the remaining degree 1 and 2 clusters, they are merged along a maximal matching.
% Furthermore, UFO trees use a stronger notion of maximality in the partitions than that of topology trees.
% Like topology trees, UFO trees maintain that no two uncombined adjacent clusters can be combined in the above ways and not combine themselves and also add the stronger maximality condition that no degree 1 clusters adjacent to a degree $d \geq 3$ cluster may remain unclustered i.e. all degree $1$ clusters must contract into a degree $\geq 3$ cluster if adjacent to it. \laxman{The degree $d \geq 3$ cluster can only combine with degree $1$ clusters, right?} \atharva{Addressed.}
The stronger maximality invariant for high degree clusters is necessary to ensure that the number of clusters in each successive round decreases geometrically.

In particular, we use maximality to prove Theorem~\ref{thm:ufo_tree_properties}, which shows that UFO trees share the same balance property as topology trees.
We further prove Theorem~\ref{thm:ufo_tree_diameter}, showing that the improved maximality guarantee on UFO trees also gives a stronger bound on the height in terms of the diameter, $\diam$, of the input tree. Therefore, the height of a UFO tree has a stricter bound of $O(\min\{\log n,\diam\})$. All of the proofs are provided in \appref{Appendix~\ref{app:ufo_efficiency_proofs}}.

\begin{theorem}\label{thm:ufo_tree_properties}
    UFO trees have height $O(\log n)$, contain $O(n)$ total nodes, and use $O(n)$ space, where $n$ is the number of vertices in the input tree.
\end{theorem}

% \begin{lemma}\label{lem:ufo_const_frac}
%     At any level $l>0$ of a UFO tree, the number of clusters is at most $5/6$ times the number of clusters at level $l-1$.
% \end{lemma}
% \begin{proof}
%     The proof of this is almost identical to that of Lemma~\ref{lem:topology_const_frac}.
%     Imagine replacing every chain in the tree with an edge or vertex in the same way as Lemma~\ref{lem:topology_const_frac}. In the resulting tree of $N$ nodes (which contains no degree 2 nodes), the number of leaves is at least 2 more than the number of other nodes, so at least $N/2+1$ of the nodes are leaves. In a UFO tree, all of these leaves combine with their neighbor (unless the tree is just two nodes in which case the lemma is trivial), since the neighbor of each leaf is high degree. Therefore at least $N/4$ combinations occur.
%     The same logic as Lemma~\ref{lem:topology_const_frac} applies for reintroducing the chains. This proves that at least $N/6$ combinations occur.
% \end{proof}

% \begin{lemma}\label{lem:ufo_height}
%     The number of levels in a UFO tree for a tree with $n$ vertices is $O(\log n)$.
% \end{lemma}
% \begin{proof}
%     Since Lemma~\ref{lem:ufo_const_frac} proved that each level has at most $5/6$ the number of clusters of the previous level, the total number of levels is at most $\log_{6/5}n$.
% \end{proof}

\begin{theorem}\label{thm:ufo_tree_diameter}
    UFO trees have height $\leq \lceil \diam/2 \rceil$, where $\diam$ is the diameter of the input tree.
\end{theorem}

\begin{figure}
    \centering
    \includegraphics[width=\columnwidth]{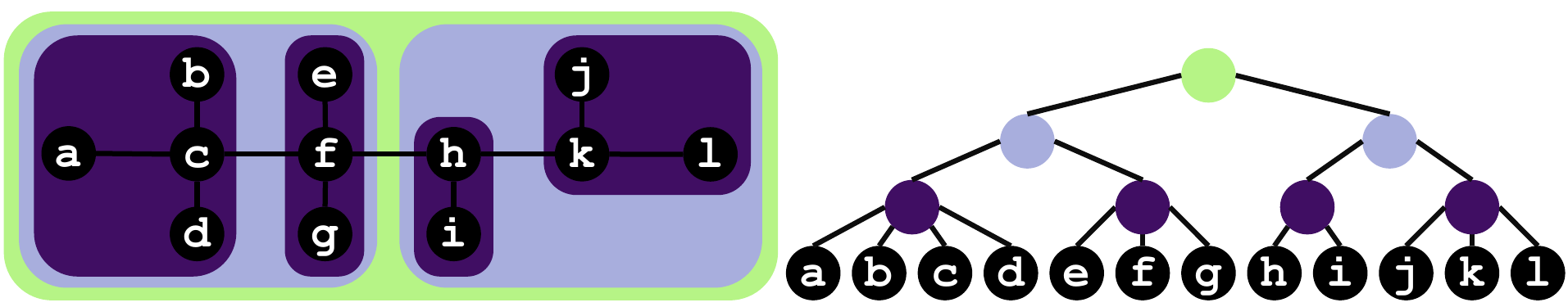}
    \caption{\small An example of an input tree (left) and a UFO tree for it (right). \revision{The high degree vertices $c$, $f$, and $k$ merge with all of their respective degree $1$ neighbors in the first round.}}
    \label{fig:ufo_example}
\end{figure}

\subsection{UFO Tree Updates}\label{sec:ufo_update}

% \laxman{The following ``goals'' could be more explicitly rewritten as desiderata for a dynamic trees algorithm that operates directly on the unbounded-degree tree.  Right now it is kind of written like that, but not as clearly as it could be.}

The general strategy of updates in UFO trees is the same as in topology trees: we delete the ancestors of clusters whose contents were changed and recluster the tree bottom-up.
However, doing this naively is not enough to maintain the same $O(\log n)$ update cost of topology trees.
One challenge arises when a high degree cluster is deleted, say with degree $\omega(\log n)$, since simply removing that cluster from the adjacency lists of its neighbors can be too costly.
Another challenge is that the tree is no longer binary: clusters may have $\omega(1)$ fanout (recall that the fanout of a cluster is its number of children).
We define a \defn{high fanout} cluster as a cluster with fanout $\geq 3$.
Deleting such a cluster may result in more than a constant number of root clusters at a level.
% Thus the algorithm should seek to avoid deleting clusters with many children. 
% A high fanout cluster can only be formed from the combination of a high degree cluster with 2 or more degree 1 clusters.
% Note that a high degree cluster is not the same as a high fan-out cluster, as the degree of a cluster is not the same as its number of children. A high fan-out cluster must have one high degree child and at least two degree 1 children.
% Additionally, the degree of a cluster is no greater than any of its children.
To prevent these issues, our update algorithm carefully avoids deleting high degree and high fanout clusters.
Figure~\ref{fig:ufo_update} shows an example of an update to a UFO tree.
Algorithms~\ref{alg:ufo_update_1}~and~\ref{alg:ufo_update_2} show the pseudo-code for UFO tree updates.
We first describe the \textsc{DeleteAncestors} procedure, which is used when performing an update in Algorithm~\ref{alg:ufo_update_2}.

\begin{figure*}
    \centering
    \includegraphics[width=\textwidth]{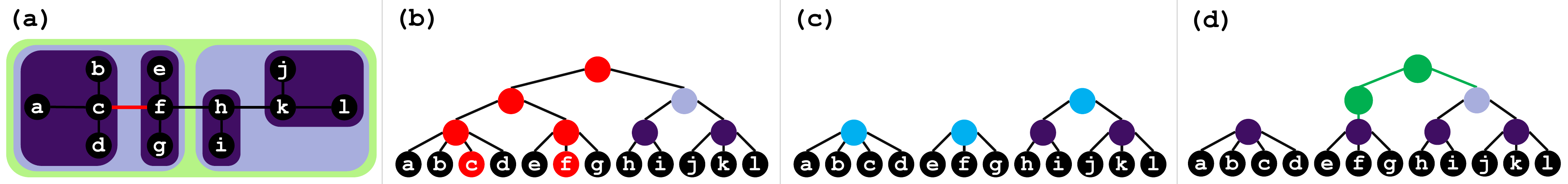}
    \caption{\small
    An example of an update to a UFO tree. Subfigure (a) shows the input tree, with the deleted edge $(c,f)$ in red.
    In (b), we show the original UFO tree with the ancestor clusters of $c$ and $f$ highlighted in red.
    In (c), we show the UFO tree after the initial ancestor removal step with the root clusters highlighted in blue.
    In this example, the parents of $c$ and $f$ were not deleted because they were high fanout, but all other ancestors were.
    In (d), we show the UFO tree after the reclustering process with the newly created clusters highlighted in green.
    }
    \label{fig:ufo_update}
\end{figure*}

% \begin{algorithm}[ht]
%     \small
%     \caption{\small $\fname{DeleteAncestors}(c,i)$}
%     \label{alg:ufo_update_1}
%     \begin{algorithmic}[1]
%         \State $\vname{prev} \gets c$, $\vname{curr} \gets c.\vname{parent}$, $\vname{next} \gets \vname{curr.parent}$
%         \While{$\vname{curr}$}
%             \If{$\vname{curr.degree} < 3$ and $\vname{curr.fanout} < 3$} \label{line:ufo_ra_check_start}
%                 \For{$x \in \vname{curr.neighbors}$}
%                     \State $x.\vname{neighbors}.\fname{delete}(\vname{curr})$
%                 \EndFor
%                 \For{$y \in \vname{curr.children}$}
%                     \State $y.\vname{parent} \gets \mathsf{null}$
%                     \State $\vname{rootclusters}[i].\mathsf{insert}(y)$
%                 \EndFor
%                 \If {$\vname{next}$}
%                     \State $\vname{next.children}.\fname{delete}(\vname{curr})$
%                 \EndIf
%                 \State $\fname{FreeCluster}(\vname{curr})$
%             \label{line:ufo_ra_check_end}
%             \ElsIf{$\vname{prev.degree} \leq 2$ and $\neg \vname{prev.deleted}$} \label{line:ufo_ra_disconnect_start}
%                 \State $\vname{prev.parent} \gets \mathsf{null}$
%             \EndIf \label{line:ufo_ra_disconnect_end}
%             \State $\vname{prev} \gets \vname{curr}$, $\vname{curr} \gets \vname{next}$, $\vname{next} \gets \vname{curr.parent}$
%             \State $i \gets i+1$
%         \EndWhile
%     \end{algorithmic}
% \end{algorithm}

\begin{algorithm}[ht]
    \small
    \caption{\small $\fname{DeleteAncestors}(c,i)$}\label{alg:ufo_update_1}
    $\vname{prev} \gets c$, $\vname{curr} \gets c.\vname{parent}$, $\vname{next} \gets \vname{curr.parent}$\;
    \While{$\vname{curr}$}{
        \uIf{$\vname{curr.degree} < 3$ and $\vname{curr.fanout} < 3$\label{line:ufo_ra_check_start}}{
            \For{$x \in \vname{curr.neighbors}$}{
                $x.\vname{neighbors}.\fname{delete}(\vname{curr})$
            }
            \For{$y \in \vname{curr.children}$}{
                $y.\vname{parent} \gets \mathsf{null}$\;
                $\vname{rootclusters}[i].\mathsf{insert}(y)$
            }
            \If {$\vname{next}$}{
                $\vname{next.children}.\fname{delete}(\vname{curr})$
            }
            $\fname{FreeCluster}(\vname{curr})$ \label{line:ufo_ra_check_end}\;
        }\ElseIf{$\vname{prev.degree} \leq 2$ and $\neg \vname{prev.deleted}$\label{line:ufo_ra_disconnect_start}}{
            $\vname{prev.parent} \gets \mathsf{null}$\label{line:ufo_ra_disconnect_end}
        }
        $\vname{prev} \gets \vname{curr}$, $\vname{curr} \gets \vname{next}$, $\vname{next} \gets \vname{curr.parent}$\;
        $i \gets i+1$\;
    }
\end{algorithm}

\myparagraph{$\fname{DeleteAncestors}$ Subroutine}
Whenever the algorithm attempts to delete a path of ancestors, it deletes only those clusters that are neither high degree nor high fanout (Algorithm~\ref{alg:ufo_update_1}, lines~\ref{line:ufo_ra_check_start}--\ref{line:ufo_ra_check_end}). The children of each deleted cluster become root clusters at their respective levels.
Additionally, if a cluster on the path is not deleted and its parent is also not deleted, the cluster is disconnected from its parent---i.e., it no longer participates in the merge forming the parent---only if its degree is $\leq 2$ (Algorithm~\ref{alg:ufo_update_1}, lines~\ref{line:ufo_ra_disconnect_start}--\ref{line:ufo_ra_disconnect_end}).
This is necessary because the contents of the degree 1 cluster have changed. If the cluster is high degree and its parent is not deleted, then it must be the ``center'' of some merge involving that high degree cluster, and it must still participate in the merge forming its parent since otherwise the remaining children of the parent no longer represent a valid type of merge.
We carefully prove that our update algorithm never produces any invalid merges in \appref{Appendix~\ref{app:ufo_proofs}}.

\begin{algorithm}[ht]
    \small
    \caption{\small $\fname{Update}(u,v,\vname{delete\_flag})$}\label{alg:ufo_update_2}
    $\fname{DeleteAncestors(u,0)}$, $\fname{DeleteAncestors(v,0)}$, $i \gets 0$ \label{line:ufo_update_ra_0} \;
    Insert or delete edge $(u,v)$ at all levels. \label{line:ufo_update_update}\;
    \While{$\vname{RootClusters}[i] \neq \emptyset$}{
        \For{$x \in \vname{RootClusters}[i]$ \label{line:ufo_update_high_deg_start}}{
            \If{$x.\vname{degree} \geq 3$}{
                $x.\vname{parent} \gets \fname{NewCluster()}$\;
                \For{$y \in x.\vname{neighbors}$}{
                    \If{$y.\vname{degree} = 1$}{
                        \If{$y.\vname{parent} \neq \mathsf{null}$}{
                            $\fname{DeleteAncestors}(y,i)$\;
                        }
                        $y.\vname{parent} \gets x.\vname{parent}$\;
                    }
                }
            }
        }\label{line:ufo_update_high_deg_end}
        \For{$x \in \vname{RootClusters}[i]$}{
            \lIf{$x.\vname{parent} \neq \mathsf{null}$} {\textbf{continue}}
            \uIf{$x.\vname{degree} = 2$ \label{line:ufo_update_deg2_start}}{ 
                \For{$y \in x.\vname{neighbors}$}{
                    \If{$y.\vname{degree} <= 2$ and $\neg y.\vname{merges}$}{
                        \uIf{$y.\vname{parent} \neq \mathsf{null}$}{
                            $\fname{DeleteAncestors}(y.\vname{parent},i+1)$\label{line:ufo_update_nonroot_1}\;
                            $x.\vname{parent} \gets y.\vname{parent}$\;
                        }\Else{
                            $x.\vname{parent} \gets \fname{NewCluster()}$\;
                            $y.\vname{parent} \gets x.\vname{parent}$\;
                        }
                    }
                }\label{line:ufo_update_deg2_end}
            }\ElseIf{$x.\vname{degree} = 1$\label{line:ufo_update_deg1_start}}{
                $y \gets x.\vname{neighbors}[0]$\;
                \uIf{$y.\vname{parent} \neq \mathsf{null}$ and $\neg y.\vname{merges}$}{
                    $\fname{DeleteAncestors}(y.\vname{parent},i+1)$ \label{line:ufo_update_nonroot_2}\;
                    $x.\vname{parent} \gets y.\vname{parent}$\;
                }\uElseIf{$y.\vname{parent} \neq \mathsf{null}$ and $y.\vname{degree} \geq 3$}{
                    $\fname{DeleteAncestors}(y.\vname{parent},i+1)$ \label{line:ufo_update_nonroot_3}\;
                    Delete edge $(y,x)$ at all levels $> i$.\;
                    $x.\vname{parent} \gets y.\vname{parent}$\;
                }\Else{
                    $x.\vname{parent} \gets \fname{NewCluster}()$\;
                    \If {$y.\vname{parent}=\mathsf{null}$}{$y.\vname{parent} \gets x.\vname{parent}$}
                }
            }\label{line:ufo_update_deg1_end}
            \If {$x$ could not merge with any neighbor}{
                $x.\vname{parent} \gets \fname{NewCluster}()$
            }
        }
        Populate the adjacency lists for new clusters.\;
        Add all new clusters to $RootClusters[i+1]$.\;
        $\vname{RootClusters}[i].\fname{clear}()$, $i \gets i+1$.
    }
\end{algorithm}

\myparagraph{\fname{Update} Algorithm}
The update algorithm begins by attempting to delete the ancestors of the endpoints of the update, $u$ and $v$ (Algorithm~\ref{alg:ufo_update_2}, line~\ref{line:ufo_update_ra_0}).
The updated edge is then inserted (or deleted) at all levels of the tree where it exists (Algorithm~\ref{alg:ufo_update_2}, line~\ref{line:ufo_update_update}).

Next, the root clusters are reclustered bottom-up, starting at level $0$.
First, the high degree root clusters are given a new parent and are merged with all of their degree 1 neighbors (Algorithm~\ref{alg:ufo_update_2}, lines~\ref{line:ufo_update_high_deg_start}--\ref{line:ufo_update_high_deg_end}). If any of those clusters are non-root clusters, the algorithm attempts to delete their ancestors to avoid memory leaks.

The remaining unmerged degree $2$ root clusters (Algorithm~\ref{alg:ufo_update_2}, lines~\ref{line:ufo_update_deg2_start}--\ref{line:ufo_update_deg2_end}) and degree $1$ root clusters (Algorithm~\ref{alg:ufo_update_2}, lines~\ref{line:ufo_update_deg1_start}--\ref{line:ufo_update_deg1_end}) then attempt to merge with a neighbor.
Degree 2 root clusters check both of their neighbors, looking for one that hasn't already merged with another cluster.
Degree 1 root clusters can either merge with a degree 2 neighbor that hasn't merged yet, or with any degree 1 or high degree neighbor.
If a root clusters merges with a non-root cluster, it is necessary to attempt to delete the ancestors of the non-root cluster's parent, since its contents have changed (Algorithm~\ref{alg:ufo_update_2}, lines~\ref{line:ufo_update_nonroot_1}, \ref{line:ufo_update_nonroot_2}, and \ref{line:ufo_update_nonroot_3}).

This update algorithm is efficient because all root clusters produced during an update have degree at most $4$. We prove this, along with bounds showing that the number of root clusters and the number of deleted clusters at each level are bounded by a constant. Combining this with the bounds on the height of UFO trees from Theorems~\ref{thm:ufo_tree_properties}~and~\ref{thm:ufo_tree_diameter}, results in Theorem~\ref{thm:ufo_tree_updates}. We provide the proof in \appref{Appendix~\ref{app:ufo_proofs}}.

\begin{theorem}\label{thm:ufo_tree_updates}
    UFO tree updates take $O(\min\{\log n,\diam\})$ time, where $n$ is the number of vertices in the input tree and $\diam$ is the diameter of the input tree.
\end{theorem}

\subsection{UFO Tree Queries}

We show how UFO trees can support (to the best of our knowledge) all well-known queries described over dynamic trees in the literature.
Our results are summarized in Theorem~\ref{thm:ufo_tree_queries}.
We give complete descriptions of how to support all of these queries in \appref{Appendix~\ref{app:ufo_queries}}.
\revision{All queries are read-only and any number of queries can be run in parallel with no synchronization.}
Most of these queries can be supported without changing the base cost per operation of UFO trees. 
Others (e.g., non-invertible subtree queries) require recomputing the result of a non-invertible function over all of the children of a high fanout cluster.
To efficiently support these types of queries, we store the sets of children of high fanout clusters using rank trees~\cite{wulff2013faster}, where the weight of each cluster is the number of vertices contained in it. 
%This allows for efficiently maintaining the value of any commutative and associative function over all clusters in the set, while supporting efficient insertions and deletions of clusters in logarithmic time.
%\kishen{Maybe this sentence can be dropped/moved to appendix?}
We use rank trees because their weight-biased property allows us to keep the overall height of the UFO tree and the cost of operations at $O(\log n)$ rather than $O(\log^2 n)$.

\begin{theorem}\label{thm:ufo_tree_queries}
    UFO trees can be modified to support path queries, subtree queries with invertible functions, and LCA queries with $O(\min\{\log n,\diam\})$ cost per operation.
    UFO trees can be modified to support any subtree queries, diameter queries, center queries, median queries, and nearest marked vertex queries with $O(\log n)$ cost per operation.
\end{theorem}

In \appref{Appendix~\ref{app:noninvertible_queries_lower_bound}} we show that there exist input trees with constant diameter where solving subtree minimum or maximum queries in $o(\log n)$ time per operation would break the comparison sorting lower bound. This lower bound helps explain why our data structure can not achieve $o(\log n)$ time per query for general non-invertible subtree queries.
We believe that similar lower-bound techniques could be used to prove that it is not possible to achieve $o(\log n)$ time per query for the other types of queries where UFO trees can not achieve an $O(\diam)$ cost bound.

%\kishen{Do we want to mention that some of these queries are hard to do in $o(\log n)$? Its not crucial, but nice to show.}

%% file: 05_parallel.tex
\section{Parallel Batch-Dynamic Trees}\label{sec:parallel}

%In this section, we describe parallel batch-dynamic topology trees and UFO trees.
We begin this section by presenting a new batch-dynamic update algorithm for topology trees, and then extend the ideas to UFO trees.
Both of our algorithms are work efficient and low-depth.
%Our algorithms take $O(\log n \log k)$ depth and $O(k \log (1+n/k))$ work~\cite{blelloch2020optimal}. 
In addition, we show that performing $k$ updates on a parallel batch-dynamic UFO tree takes $O(k\diam)$ work.
% The interface is defined as follows:
% \begin{itemize}
%     \item $\textbf{BatchUpdate}\{(x_1,u_1,v_1),\hdots,(x_k,u_k,v_k)\}$ takes a list of $k$ tuples representing edge updates. The $i$-th tuple contains a bit $x_i$ which represents whether the update is an insertion or deletion, and two vertices $u_i$ and $v_i$ which represent the edge being updated.
% \end{itemize}
% \atharva{We should make a note that while our current implementation does not support this it can be modified to do so.}
% \quinten{Let's do this in the implementation section.}

Our algorithms allow for mixed batches of insertions and deletions.
We assume that the input contains at most one update involving any edge, and that any possible ordering of the updates defines a valid sequence of sequential updates.

\subsection{Parallel Batch-Dynamic Topology Trees}

The basic framework for batch updates is the same as the sequential algorithm: first delete all the ancestors of the leaf clusters for the endpoints of each update, then recluster the tree bottom-up.
\revision{Unlike in the sequential setting, the number of root clusters and clusters to delete at each level can be much larger than $O(1)$ (we prove that it is $O(k)$ in \appref{Appendix~\ref{app:batch_update}}).
Hence an efficient batch-parallel algorithm should recluster and delete clusters in parallel.}
% First all the ancestor clusters of the leaves for each vertex with an incident update are deleted. Then the tree is re-clustered from the bottom level up, possibly deleting more paths of ancestors at each level.
% In the batch-dynamic setting, there are two additional important issues we must deal with:
% \begin{enumerate}[label=(\arabic*),topsep=0pt,itemsep=0pt,parsep=0pt,leftmargin=15pt]
%     \item Ancestor paths may be deleted starting from each level. Naively this could take $O(\log^2 n)$ depth across all levels.
%     \item Re-clustering the tree requires finding a maximal matching over the root clusters at each level.
% \end{enumerate}
Algorithm~\ref{alg:batch_update_topology} shows the pseudo-code for our batch-parallel update algorithm.

\begin{algorithm}[ht]
\small
\caption{\small $\fname{BatchUpdate}((u_1,v_1,\vname{del}_1),\ldots,(u_k,v_k,\vname{del}_k))$}\label{alg:batch_update_topology}
Update adjacency lists of updated leaf clusters.\label{line:upd_adj}\;
$R_0 \gets$ Leaf clusters for $u_1,\hdots,u_k,v_1,\hdots,v_k$.\label{line:root0}\;
$D_1 \gets R_0.\fname{MapToParents}()$ \label{line:del1}\;
\While{$|R_i| > 0~\text{or}~|D_{i+1}| > 0$\label{line:batch_loop}}{
    $D_{i+2} \gets D_{i+1}.\fname{MapToParents}()$ \label{line:next_del}\;
    $R_i \gets R_i~\cup D_{i+1}.\fname{MapToChildren}()$\label{line:new_roots}\;
    \parfor{$c \in D_{i+1}$\label{line:del_start}}{
        Delete $c$ from its neighbors' adjacency lists.\label{line:del_adj}\;
        Delete $c$ from its parent's child list.\label{line:del_child}
    }
    \parfor{$c \in D_{i+1}$}{
        $\fname{FreeCluster}(c)$
    }\label{line:del_end}
    $[M, NR] \gets \fname{ParallelMaximalMatching}(R_i)$\label{line:matching}\;
    \parfor{$(c_1,c_2) \in M$\label{line:merge_start}}{
        \lIf{$c_2 \in R_i$}{$c_2.\vname{parent} \gets \fname{NewCluster}()$}
        $c_1.\vname{parent} \gets c_2.\vname{parent}$\;
    }\label{line:merge_end}
    Populate the adjacency lists for parents of $R_i$. \label{line:new_adj}\;
    $D_{i+2} \gets D_{i+2}~\cup NR.\fname{MapToParents}()$\label{line:new_del}\;
    $R_{i+1} \gets R_i.\fname{MapToParents}()$ \label{line:next_roots}
}
\end{algorithm}

\myparagraph{Updating Adjacency Lists in Parallel}
The first step is to update the adjacency lists of the leaf clusters for endpoints of the updates (line~\ref{line:upd_adj}).
% Since clusters in the topology tree have degree $\leq 3$, updating the adjacency lists can be done in $O(\log k)$ depth and $O(k)$ work by using atomic compare-and-swap operations.
\revision{This is done by mapping each edge update to two directed edges and then using a parallel semisort~\cite{gu2015top} to group the directed edges by their endpoint.
Since the input to the topology tree has vertices with degree $\leq 3$, the number of updates incident to any vertex is at most $6$: deleting at most $3$ current incident edges and inserting at most $3$ new ones.}
A similar approach can be used to efficiently remove all deleted clusters from the adjacency lists of their neighbors (line~\ref{line:del_adj}), and to populate the adjacency lists for new clusters (line~\ref{line:new_adj}). 
The same approach handles updating the child lists of clusters (line~\ref{line:del_child}) since the fanout of clusters is  $\leq 2$ in topology trees.

\revision{\myparagraph{Low-Depth Ancestor Removal}}
During \revision{the sequential} reclustering step, whenever a root cluster merges with a non-root cluster, the ancestors of the non-root cluster's parent must be deleted.
Since this can happen at each of $O(\log n)$ levels, and the number of ancestors on a path is $O(\log n)$, proactively deleting the ancestors would result in $O(\log^2 n)$ depth overall.
Instead, we employ a lazy approach to deleting clusters.
The idea is to keep a set of clusters that should be deleted \revision{at the level just above the current level that is being reclustered.}
We use $R_i$ and $D_i$ to indicate the set of root clusters at level $i$ and the set of clusters to delete at level $i$ respectively.
At each level $i$, the algorithm maintains $R_i$, and $D_{i+1}$.
The set of updated leaf clusters will serve as $R_0$ (line~\ref{line:root0}).
Since all of the ancestors of the root clusters in $R_0$ should be deleted, $D_1$ is \revision{initialized as} the set of parents of clusters in $R_0$ (line~\ref{line:del1}).
Whenever a set of clusters \revision{is mapped to} its parents or its children \revision{with the \fname{MapToParents} or \fname{MapToChildren} function, we use a parallel map operation followed by a parallel remove duplicates operation~\cite{CLRS}.}

\revision{At each level $i$,} before reclustering the clusters in $R_i$, we store the parents of clusters in $D_{i+1}$ as $D_{i+2}$ \revision{(line~\ref{line:next_del})}.
Then the clusters in $D_{i+1}$ are all deleted, and their children are added to $R_{i}$ \revision{(line~\ref{line:new_roots})}.
Additionally, when a root cluster in $R_{i}$ merges with a non-root cluster, its grandparent (if any) must be added to $D_{i+2}$ \revision{(line~\ref{line:new_del}), which corresponds to removing the ancestors of that non-root cluster's parent in the sequential update algorithm. This is described in further detail in the following paragraphs.}

\revision{\myparagraph{Parallel Reclustering}}
\revision{The sequential topology tree reclustering algorithm uses a simple greedy algorithm to find a maximal matching of the root clusters and their non-root neighbors that don't already merge with another cluster.}
% When the topology tree is reclustered, we compute a maximal matching over \defn{chains} of degree 2 clusters and their neighbors.
For batch updates, the number of such clusters at each level can be much larger than constant (we prove that it is $O(k)$ in \appref{Appendix~\ref{app:batch_update}}).
% For degree 3 root clusters, they must simply combine with any degree 1 neighbor if it exists. For two neighboring degree 1 clusters they must combine.
% The difficulty lies in computing a maximal set of allowed combinations in the \defn{chains} of degree 2 clusters and their neighbors. The set of combinations computed should be precisely a maximal matching over the subgraph induced by these clusters.
% Sequentially, a maximal matching can be computed easily in linear time.
% In parallel, consider the equivalent problem of computing a maximal independent set (MIS) on the dual graph where each edge in the collection of chains is a vertex and each vertex is an edge. Let $k$ be the number of clusters in the collection of chains.
\revision{For degree $3$ clusters which can only merge with neighboring degree $1$ clusters, these can all be done naively in parallel since no other cluster is adjacent to a degree $1$ cluster.
The remainder of the clusters of interest form a collection of linked lists.}
\revision{To compute the maximal matching over these linked lists efficiently in parallel}, we apply parallel list ranking~\cite{blelloch2020optimal}, which takes $O(\log k)$ depth \whp{} and $O(k)$ worst-case work in the binary fork-join model.
\revision{Given the list ranking, each cluster at an even index can merge with its neighbor at the next index.}
%
% Additionally, Anderson~\etal~\cite{anderson2023parallel,anderson2024deterministic} described a deterministic algorithm to compute an MIS in a constant-degree graph in $O(\log^{(c)} k)$ depth and $O(k)$ work for any constant $c$ (e.g. $O(\log\log k)$ depth for $c=2$).

\revision{\myparagraph{Putting It All Together}}
\revision{The main loop of the algorithm deletes} level $i+1$ clusters and reclusters the level $i$ root clusters at every level until the tree is complete and all clusters that should be deleted have been (line~\ref{line:batch_loop}).
First, the parents of clusters in $D_{i+1}$ must be stored in $D_{i+2}$ before deleting the clusters in $D_{i+1}$ (line~\ref{line:next_del}).
Next, the clusters whose level $i+1$ parent will be deleted in this round are added to $R_i$ (line~\ref{line:new_roots}).
Finally, the clusters in $D_{i+1}$ are deleted from the adjacency lists of all their neighbors, and the clusters are deleted (lines~\ref{line:del_start}--\ref{line:del_end}).

Then the algorithm computes a maximal set of merges between root clusters and their non-root neighbors that don't already combine with another cluster, (line~\ref{line:matching}).
In the pseudo-code, the $\mathsf{ParallelMaximalMatching}$ function returns $M$ and $NR$. 
$M$ is the set of merge pairs. 
If a merge with a non-root cluster is found, it is the second cluster in the pair. 
$NR$ is the set of parents of non-root clusters that were merged with.
New parent clusters are created, and parent pointers are set appropriately for each merge, and adjacency lists at level $i+1$ are populated for the new clusters (lines~\ref{line:merge_start}--\ref{line:new_adj}).
For any combination in the computed matching that involves a non-root cluster, the algorithm must remove all of the ancestors of the existing parent cluster, thus the parents of clusters in $NR$ are added to $D_{i+2}$ (line~\ref{line:new_del}).
The new parent clusters and the parents of any non-root clusters involved in a combination compose $R_{i+1}$, the set of root clusters at the next level (line~\ref{line:next_roots}).
Once there are no root clusters or clusters to delete, the batch update is complete.

\myparagraph{Analysis}
We prove the correctness of Algorithm~\ref{alg:batch_update_topology} and analyze its cost in \appref{\ref{app:batch_update}}. Our result is summarized by Theorem~\ref{thm:batch_topology_cost}.

\begin{restatable}{theorem}{topologybatchupdate}
\label{thm:batch_topology_cost}
    Batch-parallel updates in topology trees take $O(k \log(1 + n/k))$ expected work and $O(\log n \log k)$ depth \whp{}, where $n$ is the number of vertices and $k$ is the batch size.
\end{restatable}

\subsection{Parallel Batch-Dynamic UFO Trees}
The algorithm for batch-parallel  updates in a UFO tree follows a similar basic progression as that of Algorithm~\ref{alg:batch_update_topology}.
Rather than re-explaining the algorithm from scratch, here we summarize the key differences and challenges for parallel batch-dynamic UFO trees.
\begin{enumerate}[label=(\arabic*),topsep=0pt,itemsep=0pt,parsep=0pt,leftmargin=15pt]
    \item \revision{The sizes of adjacency lists and child sets in a UFO tree are not necessarily $O(1)$.}
    \item \revision{The UFO tree update algorithm requires certain clusters along ancestor paths to not be deleted. Consequently, edge insertions or deletions performed at lower levels must be propagated to higher levels if both endpoints of the modified edge remain part of distinct clusters that were not deleted at those higher levels.}
    \item \revision{In the sequential UFO tree algorithm, the degree of any root cluster is bounded by a constant. This is no longer true in the batch setting.}
\end{enumerate}
\revision{Algorithm~\ref{alg:batch_update_ufo} shows the pseudo-code for our batch-parallel update algorithm in UFO trees, highlighting the differences from the topology tree algorithm.}

\begin{algorithm}[ht]
\small
\caption{\small $\fname{BatchUpdate}((u_1,v_1,\vname{del}_1),\ldots,(u_k,v_k,\vname{del}_k))$}\label{alg:batch_update_ufo}
Update adjacency lists of updated leaf clusters.\label{line:ufo_upd_adj}\;
$R_0 \gets$ Leaf clusters for $u_1,\hdots,u_k,v_1,\hdots,v_k$.\label{line:ufo_root0}\;
$D_1 \gets R_0.\fname{MapToParents}()$ \label{line:ufo_del1}\;
\colorbox{ufogreen}{$E^-_0 \gets \{(u,v,\vname{del}) \in U \mid \vname{del}=true\}$}\;
\colorbox{ufogreen}{$E^-_1 \gets E^-_0.\fname{MapToParents}()$}\;
\While{$|R_i| > 0~\text{or}~|D_{i+1}| > 0$\label{line:ufo_batch_loop}}{
    $D_{i+2} \gets D_{i+1}.\fname{MapToParents}()$ \label{line:ufo_next_del}\;
    $R_i \gets R_i~\cup D_{i+1}.\fname{MapToChildren}()$\label{line:ufo_new_roots}\;
    \colorbox{ufogreen}{Delete edges in $E^-_{i+1}$ from adjacency lists.}\label{line:ufo_delete_edges}\;
    \colorbox{ufogreen}{$E^-_{i+2} \gets E^-_{i+1}.\fname{MapToParents}()$}\;
    \parfor{$c \in D_{i+1}$\label{line:ufo_del_start}}{
        \If{\colorbox{ufogreen}{$c.\vname{degree} < 3$ and $c.\vname{fanout} < 3$}}{
            Delete $c$ from its neighbors' adjacency lists.\label{line:ufo_del_adj}\;
            Delete $c$ from its parent's child list.\label{line:ufo_del_child}
            \colorbox{ufogreen}{Add all edges incident to $c$ to $E^-_{i+1}$.}\label{line:ufo_update_del_edges}
        }
        \ElseIf{\colorbox{ufogreen}{$c.degree < 3$}}{
            \colorbox{ufogreen}{Disconnect $c$ from its parent and add $c$ to $R_{i+1}$.}
        }
    }
    \parfor{$c \in D_{i+1}$}{
        \If{\colorbox{ufogreen}{$c.\vname{degree} < 3$ and $c.\vname{fanout} < 3$}}{
            $\fname{FreeCluster}(c)$
        }
    }\label{line:ufo_del_end}
    \colorbox{ufogreen}{$[M, NR] \gets \fname{ParallelUFOMaximalMatching}(R_i)$}\label{line:ufo_matching}\;
    \parfor{$(c_1,c_2) \in M$\label{line:ufo_merge_start}}{
        \lIf{$c_2 \in R_i$}{$c_2.\vname{parent} \gets \fname{NewCluster}()$}
        $c_1.\vname{parent} \gets c_2.\vname{parent}$\;
    }\label{line:ufo_merge_end}
    Populate the adjacency lists for parents of $R_i$. \label{line:ufo_new_adj}\;
    $D_{i+2} \gets D_{i+2}~\cup NR.\fname{MapToParents}()$\label{line:ufo_new_del}\;
    $R_{i+1} \gets R_i.\fname{MapToParents}()$ \label{line:ufo_next_roots}
}
\end{algorithm}

\myparagraph{Challenge 1: \revision{Large} Adjacency Lists and Child Sets}
\revision{Our batch-parallel topology tree update algorithm relied on the fact that the sizes of adjacency lists and child sets were bounded by a constant to perform efficient parallel updates.}
Since there is no constant bound on the degree of clusters in UFO trees nor the fanout of clusters, we use a parallel hash table~\cite{gil1991towards} to represent the neighbors of each cluster and the children. This allows us to efficiently insert and delete batches of neighbors in low depth.
%Specifically, a parallel hash table~\cite{gil1991towards} supports updates in $O(k)$ work and $O(\log n)$ depth \whp{} for a batch of $k$ insertions or deletions, where $n$ is the number of elements in the set. This also supports a ``parallel for each element'' operation in $O(n)$ work and $O(\log n)$ depth.
%\quinten{Check that this is correct and in binary forking. Add appropriate cite.}
%
To batch update the adjacency lists at a level, we use parallel semisort~\cite{gu2015top} to group the updates by their endpoint, and then pass each group as a batch update to the neighbor set of that endpoint.
\revision{The same technique is used for updating child sets.}
%This can be done using a parallel semisort~\cite{gu2015top} in $O(k)$ expected work and $O(\log k)$ depth \whp{} per level.

\myparagraph{Challenge 2: \revision{Conditionally} Deleting Clusters}
\revision{The sequential UFO tree algorithm does not delete} clusters that are high degree or high fanout.
\revision{This is complicated by the fact that deleting or disconnecting some children of a high fanout cluster, can cause it to become low fanout and it should subsequently be deleted.}
With batch updates, a cluster that was previously very high degree can become low degree, or a very high fanout cluster can become low fanout.
Since we explicitly store \revision{and update} neighbor and child sets for each cluster, it is easy to determine the degree and fanout of a cluster.
\revision{Thus when checking whether to delete clusters in $D_{i+1}$ (lines~\ref{line:ufo_del_start}--\ref{line:ufo_del_end}), the algorithm can simply check the size of the neighbor set and the child set for each cluster.
Just like the sequential algorithm, any clusters that are not deleted but have degree $\leq 2$ are disconnected from their parent and added as root clusters at the next level.}

\revision{
% An additional complication in the sequential UFO tree update algorithm is that certain clusters along ancestor paths need to be not deleted. 
To handle the case of clusters becoming low degree due to edge deletion updates, we must remove these deleted edges from every level of the tree before determining whether to delete the clusters at that level.
% Since some clusters are not deleted, edge insertions or deletions performed at lower levels must be propagated to higher levels if both endpoints of the modified edge remain part of distinct clusters that were not deleted at those higher levels.
%
To do this in low depth in the batch-parallel setting, we employ a similar method as the low-depth ancestor removal in batch-parallel topology tree updates.
We maintain a set $E^-_i$ which represent the edges that need to be deleted at level $i$.
We maintain this set just for the level directly above the current reclustering level rather than propagating the updates up through all levels immediately.
The edges in $E^-_i$ must be deleted from level $i$ before the algorithm determines whether to delete level $i$ clusters to handle the case of high degree clusters becoming low degree (line~\ref{line:ufo_delete_edges}).
Additionally, whenever we delete a cluster, all of its incident edges are deleted, and thus added to $E^-_i$ (line~\ref{line:ufo_update_del_edges}).
}

\myparagraph{Challenge 3: \revision{Reclustering} High Degree Clusters}
For sequential UFO trees, we proved that any root cluster has degree $\leq 4$. This is not true in the batch algorithm: root clusters can have degree as high as $\Omega(k)$.
To recluster these \revision{high degree} root clusters, we iterate over their neighbors in parallel using the parallel hash table, and find all of their degree $1$ neighbors.
\revision{Just like for topology trees, these high degree root clusters can be handled independently in parallel, since the sets of clusters they must merge with are disjoint. The remaining clusters form a collection of linked lists for which a maximal matching is determined using list ranking.}
Although some root clusters can have degree $\Omega(k)$, we prove in \appref{Appendix~\ref{app:batch_update}} that the sum of degrees of root clusters at any level is $O(k)$. Thus, our algorithm is work efficient.

\revision{
\myparagraph{Analysis}
In \appref{Appendix~\ref{app:batch_update}} we analyze the correctness and the cost of the UFO tree batch-update algorithm in detail, resulting in Theorem~\ref{thm:batch_ufo_cost}.
}

\begin{restatable}{theorem}{ufobatchupdate}
\label{thm:batch_ufo_cost}
    Batch-parallel updates in a UFO tree take $O(\min \{k \log(1 + n/k), k\diam\})$ expected work and $O(\log n \log k)$ depth \whp{}, where $n$ is the number of vertices, $\diam$ is the input forest's diameter, and $k$ is the batch size.
\end{restatable}

%% file: 06_experiments.tex
\section{Experimental Evaluation}
All of the experiments presented in this paper were run on a 96-core machine (192 hyper-threads) with 4 $\times$ 2.1 GHz Intel Xeon(R) Platinum 8160 CPUs (each with 33MiB L3 cache) and 1.5TB of main memory.
Our implementations are all written in C++ and compiled with -O3 optimization.
Our parallel implementations use ParlayLib~\cite{parlaylib}, a C++ library for shared memory parallel programming. We use Microsoft's mimalloc~\cite{leijen2019mimalloc} in combination with the parallel allocator from ParlayLib for parallel memory allocation. 
% When reporting the average we refer to the \textbf{geomean} of the data we compare.
We provide details about our implementations in \appref{Appendix~\ref{app:experiments}}.

\myparagraph{Inputs}
To compare performance on a wide range of trees with different attributes, we test our implementations on several types of synthetic input trees, including path graphs, perfect binary trees, perfect k-ary trees, stars, and dandelions.
We also test on randomly generated degree 3, unbounded degree, and preferential attachment trees.

We also test our implementations on spanning trees of several real-world graphs, \revision{summarized in Table~\ref{tab:input_graphs}.}
% Due to space constraints, we describe our graph inputs in \appref{Appendix~\ref{app:experiments}}.
For each graph, we test both a {\em breadth-first spanning forest (BFS)} starting at a random vertex, and a {\em random incremental spanning forest (RIS)} generated by inserting the edges of the graph in a random order and only including edges that are not connected in the current spanning forest.
% For each tree input, we randomly shuffle the ids of the vertices to remove potential spatial locality benefits that some data structures may benefit from more than others.
% For each tree input with randomized structure and each real-world spanning tree, we report the average performance over multiple trials.

\subsection{Sequential Performance}
% Alternate conclusion: UFO tree and LCT scale with decreasing diameter while others don't
%We first benchmark the update and query efficiency of our sequential implementation of UFO trees. 
We compare UFO trees against our new implementations of topology and RC trees, both of which use ternarization only when required. 
We also compare against existing implementations of link-cut trees~\cite{tseng2019batch}, splay top trees~\cite{holm2023SplayTopTrees}, and Euler tour trees (implemented with treaps, splay trees, and skip lists)~\cite{tseng2019batch}.
We used our own sequential RC tree implementation because we found it to be significantly faster than running the recent parallel RC tree implementation~\cite{ikram2025parallel} with batch size $k=1$ on a single thread.

\begin{table}[ht]
   \revision{
   \centering
   \footnotesize
   \caption{\small The real-world graph datasets used in our experiments. }
   \begin{tabular}{l l l c c c}
       \toprule
       Name & Abbrv. & Type & $|V|$ & $|E|$ & Cite \\
       \midrule
       USA Roads & USA & Road & 23.95M & 28.85M & \cite{roadgraphs} \\
       ENWiki & ENW & Web & 4.21M & 91.94M & \cite{boldi2004webgraph} \\
       StackOverflow & SO & Temporal & 6.02M & 28.18M & \cite{paranjape2017motifs} \\
       Twitter & TWIT & Social & 41.65M & 1.20B & \cite{kwak2010what} \\
       \bottomrule
   \end{tabular}
   \label{tab:input_graphs}
   }
\end{table}

\begin{figure*}[ht]
    \centering
    \includegraphics[width=0.93\linewidth]{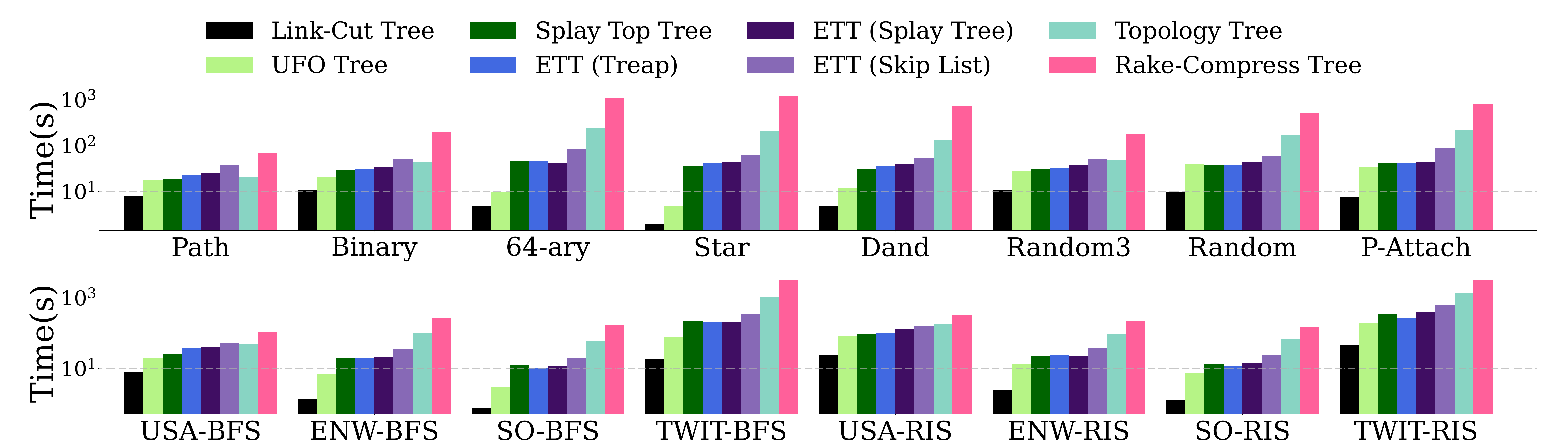}
    \caption{\small The results of our \textbf{sequential} dynamic tree data structure update speed experiments. The top row shows results on a variety of synthetic trees with $n=10^7$. The bottom row shows results on breadth-first search forests and random incremental spanning forests of our real-world graph datasets. We report the total time for inserting all edges and then deleting all edges, both in a random order.}
    \label{fig:seq_update_speed}
\end{figure*}

\myparagraph{Update Speed}
We study update speed by measuring the total time taken to build and delete an entire tree, i.e., for a tree with $10^7-1$ edges, to perform $10^7-1$ insert followed by $10^7-1$ deletions (both in a random order).
We run on all of our aforementioned synthetic trees with $n=10^7$, and BFS/RIS trees on real-world graphs. Figure~\ref{fig:seq_update_speed} shows the results.

We find that UFO trees perform better than all other data structures besides link cut tree on nearly every input.
We see that UFO trees are on average (geometric mean) only $3.23\times$ slower than link-cut trees (at most $5.7\times$ slower).
The next best data structure (splay top trees) can be as much as $18\times$ slower than link-cut trees ($6.11\times$ slower on average).
We also find that both UFO trees and link-cut trees are significantly faster than all other implementations on low-diameter inputs (e.g., star, 64-ary, and low-diameter real-world trees such as SO-BFS).
Compared to splay top trees, which support the same set of queries as UFO trees but are not obviously parallelizable, we find that UFO trees are $1.9\times$ faster on average.
RC trees, which are parallelizable and also support the same set of queries, always perform much slower than the other dynamic trees implementations, and are 24.78$\times$ slower than UFO trees on average.
In summary, we find that UFO trees perform on par with dynamic tree implementations that only support specialized queries i.e. link-cut trees and Euler tour trees, while supporting a much broader set of queries.

\begin{figure*}[ht]
    \centering
    \includegraphics[width=0.93\linewidth]{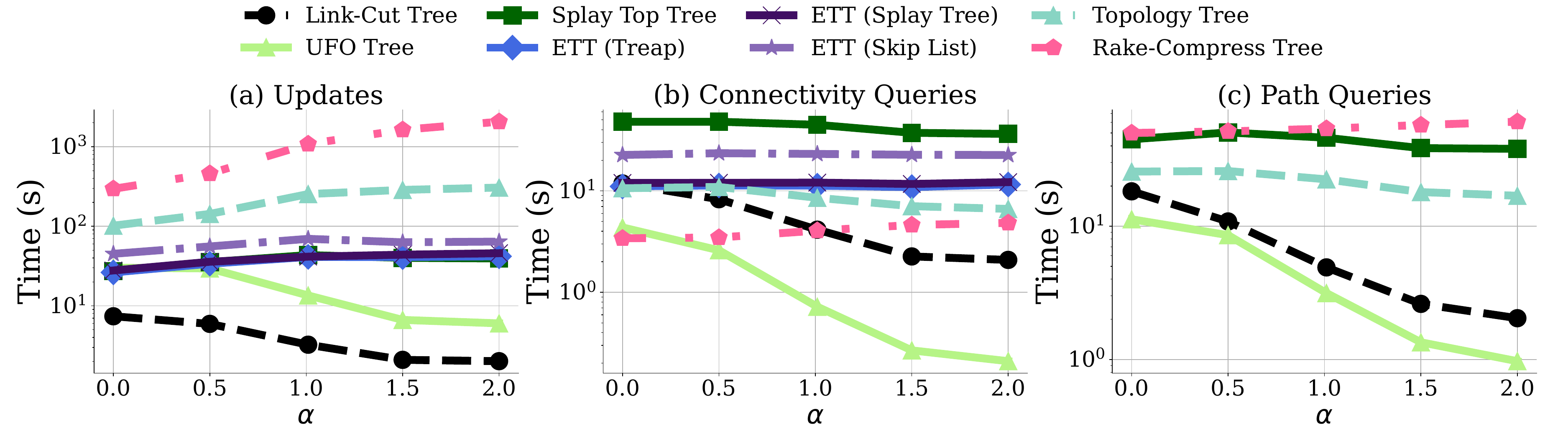}
    \caption{\small The results of our \textbf{diameter sweep} experiment. As $\alpha$ increases the diameter tends to decrease. Subplot (a) shows the results for total update time with $n=10^7$. Subplots (b) and (c) show the results for connectivity queries and path queries respectively. We report the total running time for $q=10^6$ queries on a full tree with $n=10^6$.}
    \label{fig:seqdiamsweep}
\end{figure*}

\myparagraph{Diameter Scaling for Updates and Queries}
To study whether the asymptotic diameter-based bounds we prove for link-cut and UFO trees in this paper yield improvements on low-diameter inputs in practice, we run an experiment that generates trees with varying (decreasing) diameter and study performance as a function of the diameter.
We generate trees by having node $i$ pick a target in $[0, i)$ based on the Zipf distribution with parameter $\alpha$, and then randomly permuting the node IDs.
As $\alpha$ increases the diameter tends to decrease.

The results for updates, connectivity queries, and path queries are shown in Figure~\ref{fig:seqdiamsweep}.
Our results show that both link-cut trees and UFO trees gracefully increase in speed with decreasing input diameter, for updates and two different types of queries.
Furthermore, all other implementations maintain relatively stable running times as the tree diameter decreases, or become slower.
We observe that updates in the topology tree and RC tree implementations become significantly slower on lower diameter inputs. 
This is due to ternarization overheads that become more costly on inputs with many high degree nodes (which is necessary for low diameter).
Interestingly, UFO trees obtain the best performance for both path and connectivity queries, even out-performing link-cut trees, likely due to the fact that link-cut trees require mutating the tree for queries, whereas UFO trees do not.

\begin{figure*}[ht]
    \centering
    \includegraphics[width=0.93\linewidth]{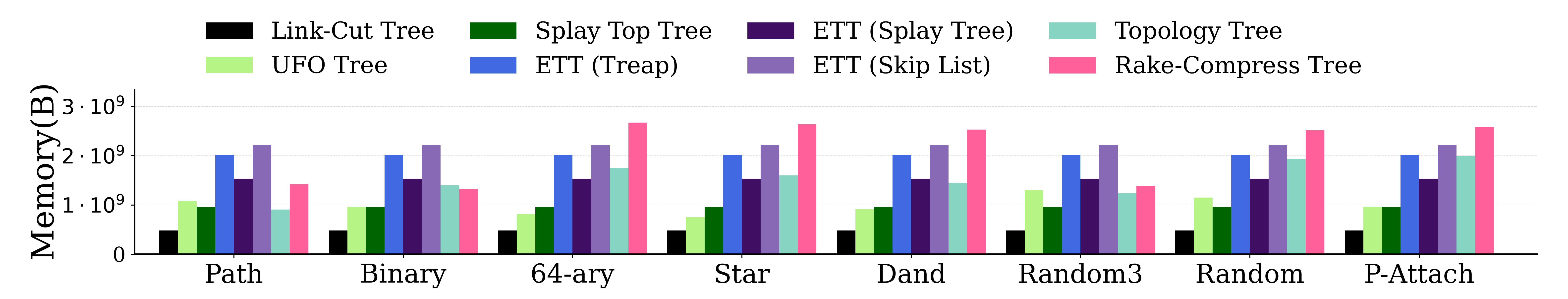}
    \caption{
    \revision{
    \small The results of our \textbf{memory usage} experiments for dynamic tree data structures. We measure memory usage on a variety of synthetic trees with $n=10^7$. We report the total bytes used for each data structure after inserting all of the edges.
    }
    }
    \label{fig:memory_usage}
\end{figure*}

\subsection{Memory Usage}
%tTo compare the memory usage of various dynamic tree data structures, 
To compare the memory usage of different data structures, we measure the amount of memory used after building a full $n$ vertex tree on $n-1$ edges.
We run this experiment using $n=10^7$ and report our results in \revision{Figure~\ref{fig:memory_usage}}.
%
%Each bar represents the total space used.
%
% In summary, our results show that the data structures that use ternarization use significantly more memory on high degree inputs, with topology and rake-compress trees using $1.91\times$ and $2.85\times$ more memory on average across such cases.
We found that tree contraction based data structures use less memory than ETTs when no ternarization is required. When ternarization is required, RC trees and topology trees use significantly more memory. UFO trees are the most space efficient tree contraction based data structure and also see a significant decrease in memory usage for low diameter inputs.
Link-cut trees consistently use the least memory in our experiments due to only using exactly $n$ tree nodes.

%Conclusion: link-cut trees should always win.  We can still show some diameter scaling showing that we go from n to n/2 based on how ``star-like'' the input is. 
%
%Pick four inputs (linked-list, star, binary tree, k-ary tree).  Show space for all data structures.  Here we would include both our RC tree and the parallel one.  The space we measure is to build the entire tree, insert all the edges (no deletions), then measure space.

% We are interested in measuring the peak memory usage of each implementation. We do this by measuring the memory usage after inserting all of the edges in the input tree.
% To provide a fair comparison we configure each data structure to support only connectivity queries, thus there is no extra memory used to store augmented values or weights for queries. \quinten{Again, we should actually do this.}
% Figure~\ref{fig:memory_usage} shows the results of the experiments for peak memory usage.

% \begin{figure}
%     \centering
%     \includegraphics[width=\linewidth]{figures/peak_space_100000.pdf}
%     \caption{Results of our experiments for memory usage.}
%     \label{fig:memory_usage}
% \end{figure}

\begin{figure*}[ht]
    \centering
    \includegraphics[width=0.95\linewidth]{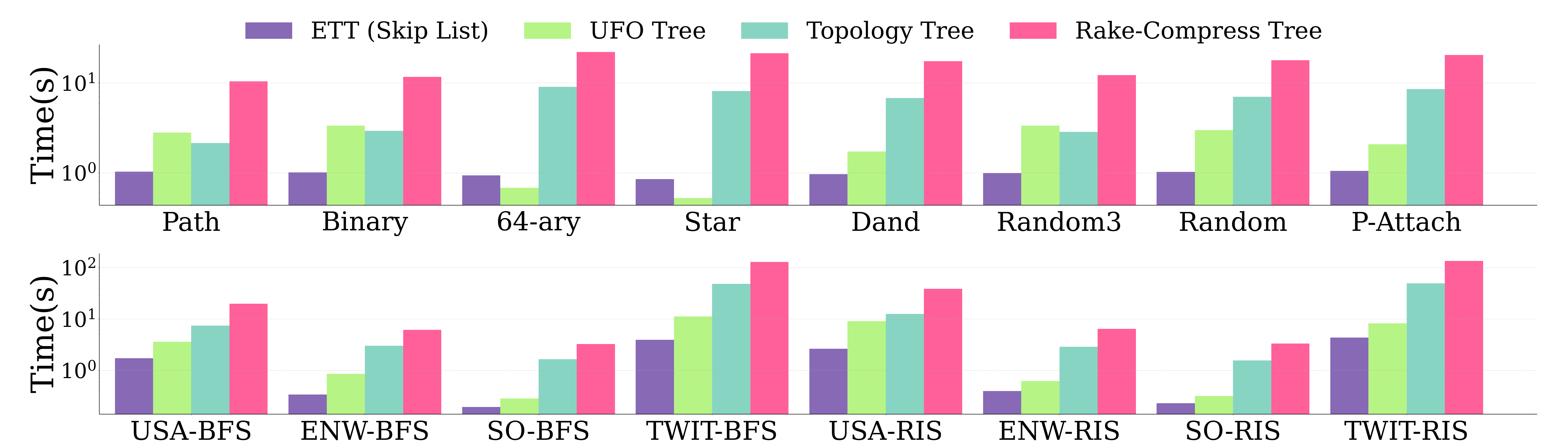}
    \caption{\small The results of our \textbf{parallel} dynamic tree data structure update speed experiments. All experiments use a fixed batch size of $k=10^6$. The top row shows results on a variety of synthetic trees with $n=10^7$. The bottom row shows results on breadth-first search forests and random incremental spanning forests of our real-world graph datasets. We report the total time for inserting all batches of edges and then deleting all batches edges, both in a random order.}
    \label{fig:par_update_speed}
\end{figure*}

\subsection{Parallel Batch-Dynamic Update Performance}
Next, we measure the performance of our new implementations of parallel batch-dynamic UFO trees and topology trees.
We compare against the state-of-the-art implementations of parallel Euler tour trees~\cite{tseng2019batch} and parallel rake-compress trees~\cite{ikram2025parallel}.
%To the best of our knowledge, these represent the state-of-the-art for parallel batch-dynamic tree data structures.
We use the batch-dynamic ternarization implementation from~\cite{ikram2025parallel} in both their implementation of RC trees and our new topology tree implementation.
\revision{We do not compare against link-cut trees or splay top trees in this section as to our knowledge no efficient batch-parallel algorithms or implementations of these data structures exist.}
We use the same synthetic and real-world inputs as the sequential update speed experiment, but we partition each input into update batches of size $k=10^6$.
\revision{We generate batches by randomly permuting the edges of the input tree, and partitioning the edges into groups of a fixed batch size. We first insert all of the edges in batches and then delete all of the edges in batches.}
Figure~\ref{fig:par_update_speed} shows the results.

We observe that in the parallel batch-dynamic setting, ETTs are the fastest implementation on average across our inputs.
The main reason for the impressive speed of ETTs is that the parallel algorithm is significantly simpler than algorithms based on tree contraction, and is very friendly for parallel and phase-concurrent settings. In particular, the update algorithm only runs a small constant number of phases of concurrent operations, and all nodes are pre-allocated hence updates require no additional memory allocation.

Despite being more complex, UFO trees are not far behind ETTs and are able to support a much larger set of queries.
UFO trees are only $1.95\times$ slower than ETTs on average, and can be up to $1.62\times$ faster on low diameter inputs.
We find that topology trees also perform well on some inputs, but suffer from ternarization overheads on inputs which require it.
Finally, RC trees do not perform as well on any of the inputs and also further suffer from ternarization overheads on inputs which require it.
RC trees are $8.89\times$ slower than UFO trees on average and upto $40\times$ slower in the worst case. 
We also show diameter scaling results for our parallel implementations in \appref{Appendix~\ref{app:results}}.

%Our points:
%- ETT is the fastest on average.
%- UFO tree is not far behind (Ax overhead on average), and better in some cases (up to Bx better on the low-diameter inputs).
%- Topology trees are also decent, but suffer on low-diameter inputs due to ternarization overheads.
%- RC trees are generally slow across the board, 
%Explain why ETTs are so much faster in the parallel setting.  Nodes are pre-allocated.  No memory allocations for batch updates, and the algorithms are phase-concurrent.  Given all this, it's actually cool that we're so close to it given the higher complexity of dynamic tree contraction, and the much broader set of queries we support.

%We run our experiments on all of the synthetic trees (using $n=10M$) and real world trees. In both cases our batch size is $k=?$. As in the sequential benchmark, each bar depicts the total time taken to build and delete an entire tree. Our results are shown in figure \ref{fig:par_update_speed}.

%\myparagraph{Diameter Scaling}
%same exp as in the sequential setting, but for batch-parallel updates.

\begin{figure}[ht]
    \centering
    \includegraphics[width=\linewidth]{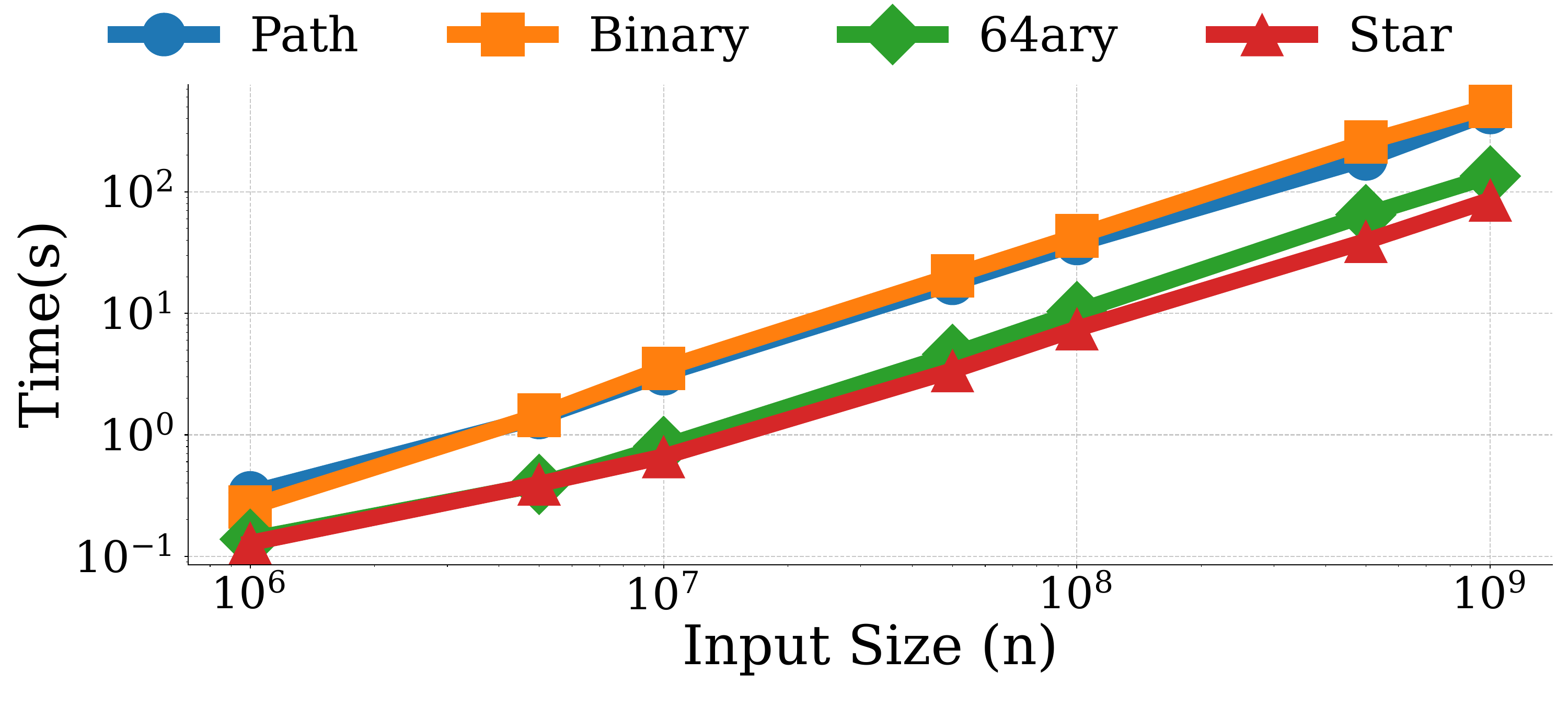}
    \caption{\small The total time for building (and destroying) UFO trees with a batch size of $k=10^6$ with varying $n$ on synthetic inputs.}
    \label{fig:scaling}
\end{figure}

\myparagraph{Scaling to Massive Inputs}
Lastly, we present the first experiments for dynamic trees on billion scale inputs and show how the performance of parallel UFO trees scales as we increase the input size.
We report the total time to insert $n-1$ edges and delete $n-1$ edges in the tree using a fixed batch size of $k=1M$ and a variable $n$.
We show the results in Figure~\ref{fig:scaling}.
%y-axis depicts the runtime in seconds taken, and the x-axis the value of $n$ on the log-log scale. These results are shown in figure \ref{fig:placeholder}.
% We see that the update performance is between 28-34 $\%$ better on average for the low-degree star and $35-40\%$ better on average for the 64-ary tree inputs at a billion scale.
All of these workloads finished within a few hundred seconds.
We note that during our experiments, only UFO trees were able to run on all input sizes on our machine.
Other implementations, like parallel ETTs segfaulted on the billion-scale input, or in the case of parallel RC trees ran out of memory on our machine.

%% file: 07_conclusion.tex
\revision{

\section{Conclusion}

We have presented UFO trees, a new parallel batch-dynamic 
tree that supports all known query functionality for dynamic
trees.
We show that our algorithms are highly work-efficient in both theory and in practice,
and obtain asymptotically better performance than most existing 
algorithms when the diameter of the input tree is shallow.
Our experimental results show that UFO trees can easily scale to billion-scale
inputs and are the fastest dynamic tree
implementation that supports a wide range of query functionality.
In future work, we plan to study tree compression
techniques to design even more space-efficient implementations of dynamic trees.

}

%% file: AA_prelims.tex
\section{Additional Related Work}

\subsection{Ternarization}\label{app:ternarization}

Since topology trees and rake-compress trees are defined only for constant-degree trees, \defn{ternarization} must be applied to the input tree before building a dynamic tree over it.
Ternarization is a standard technique that transforms a tree of arbitrary degree into a tree with maximum degree $\leq 3$.
An example is shown in Figure~\ref{fig:ternarization}.

When ternarization is applied to an input tree $T$, each vertex $v$ with degree $d(v) > 3$ is replaced by a path of $d(v)$ vertices, each of degree $\leq 3$ and adjacent to the vertices that the original vertex was adjacent to. This path is referred to as the {\it ternarized path} of $v$. 
For a tree with $n$ vertices and $n - 1$ edges, there are at most $n/d$ vertices with degree $d$. Since ternarization will convert each of these vertices to a vertex of degree $1$ adjacent to a path of length $d$ (with an additional $d$ vertices), at most $d \times \frac{n}{d} = n$ vertices will be added to the tree, resulting in a tree with at most $2n$ vertices and $2n - 1$ edges. This argument can be applied to each tree in a forest, leading to linear bounds on the size of the ternarized tree. 

In the case of topology trees, updates and queries are handled as follows: Adding an edge via a link is as simple as adding another vertex to the path created by ternarizing that vertex. Deleting an edge via a cut can be done as follows. Suppose in the the ternarized path of a vertex $u$, the vertices $u_1, u_2,$ and $u_3$ are connected to each other via 1 edge and connected to each other via fake edges and to actual neighbors of $u$ via real edges (it should be noted that $u$ could be connected to some vertex $v$ that is also ternarized, in which case $u_i$ is connected to some $v_j$ on the ternarized paths of both vertices). If we cut the real edge associated with $u_2$, we remove it from the path by disconnecting it from $u_1$ and $u_3$ and then connecting $u_1$ with $u_3$.

Handling queries depends on the specific query type. Connectivity queries, for instance, require no additional modifications. However, queries such as the sum of weights on a path or maximum weight queries require the edges on the ternarized path to be assigned a suitable $\perp$ element as a weight. These specifications must be provided by the user when constructing the topology tree.

We note that this mapping can also be efficiently maintained in the batch-dynamic setting~\cite{ikram2025parallel}.
Ternarization does not affect the asymptotic costs of the data structure, however, it can lead to high overheads in both space and time in practice.

\begin{figure}
    \centering
    \includegraphics[width=0.47\textwidth]{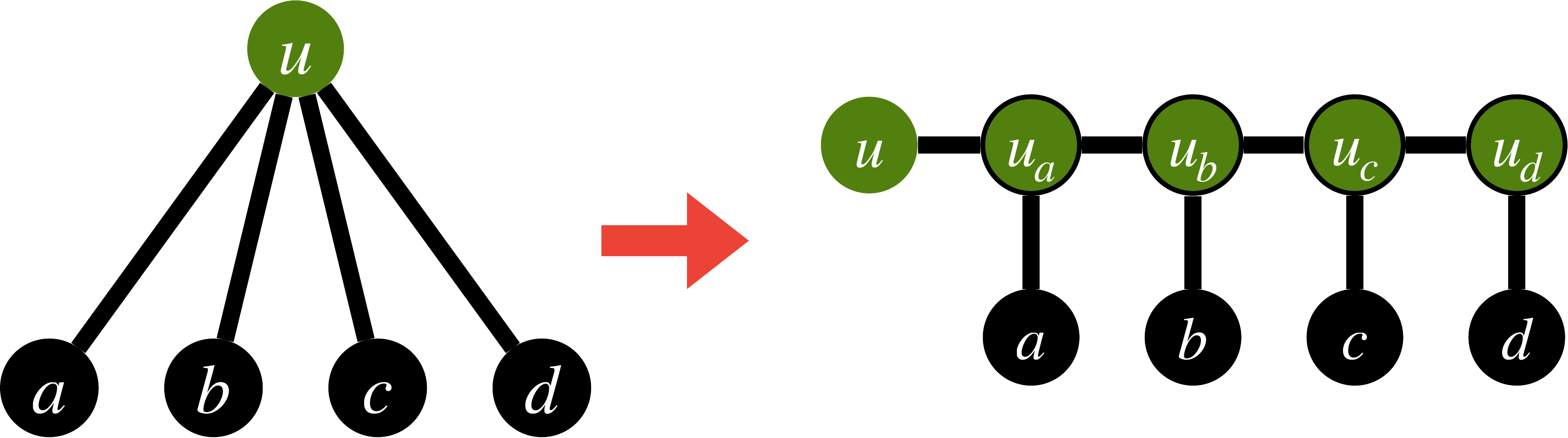}
    \caption{\small An illustration of some vertex in a tree being ternarized.}
    \label{fig:ternarization}
\end{figure}

%% file: AB_analysis.tex
\section{Detailed Analysis}

\subsection{New Analysis of Link-Cut Trees} \label{app:link_cut_tree}

Here we prove Theorem~\ref{thm:link_cut_diameter}, which states that the cost of link-cut tree operations is upper bounded by the diameter squared of the input tree.
We prove this for the original link-cut tree data structure~\cite{sleator1983data}, which uses splay trees and was previously proven to have $O(\log n)$ amortized cost per operation. Our proof should also apply to other link-cut tree variants, such as the variant with worst-case $O(\log n)$ cost per operation~\cite{sleator1983data}.
We assume the reader is familiar with the details of this data structure, and refer to the original paper otherwise.
Combining this new result with these prior results implies a cost bound of $O(\min\{\log n,\diam^2\})$ amortized or worst-case respectively for link-cut tree operations.

\begin{theorem}\label{thm:link_cut_diameter}
    Operations in link-cut trees take $O(\diam^2)$ worst-case time where $\diam$ is the diameter of the input tree.
\end{theorem}
\begin{proof}
    The cost of all operations in link-cut trees is proportional to the cost of the expose operations, plus the cost of a concatenate or split operation.
    Each concatenate or split operation takes worst-case $O(\diam)$ time, because there are at most $\diam$ elements in any path, so the splay tree operation takes worst-case $O(\diam)$ time.
    For the same reason, a splice operation takes worst-case $O(\diam)$ time.
    An expose operation may call splice, concatenate, or split at most $\diam$ times, since the length of the path being exposed is at most $\diam$.
    Thus expose takes $O(\diam^2)$ worst-case time.
\end{proof}

% These rules lead to the following properties of the topology tree.
% \begin{enumerate}
%     \item \label{tt_prop_1}
%     Each cluster has tree degree at most 3
%     \item \label{tt_prop_2}
%     Clusters of tree degree 3 have fanout 1
%     \item \label{tt_prop_3}
%     Clusters of tree degree 2 or less have fanout at most 2
%     \item \label{tt_prop_4}
%     No remaining adjacent clusters can be combined and still follow the above i.e. the merges are maximal.
% \end{enumerate}
%
% Note that property \ref{tt_prop_4} is critical for maintaining the balance of $\mathcal{T}$.

\subsection{Proof of Theorem~\ref{thm:topology_tree_properties}}\label{app:topology_efficiency_proof}

\begin{lemma}\label{lem:topology_const_frac}
    At any level $l>0$ of a topology tree, the number of clusters is at most $5/6$ times the number of clusters at level $l-1$.
\end{lemma}
\begin{proof}
Consider the tree formed by the $N$ level $l-1$ clusters. We will show that at least $N/6$ combinations must occur, thus the number of clusters at level $l$ is at most $5/6 \cdot N$.

To simplify the analysis, we construct a new tree by collapsing chains of degree 1 and degree 2 vertices.
% imagine a new tree formed by removing each connected subgraph in the tree of degree 1 and 2 vertices (a chain). 
Chains with two endpoints are replaced by an edge, and chains with one endpoint are replaced by a single leaf node. This results in a tree with only degree 3 nodes and leaves.
In such a tree, the number of leaves is at most 2 more than the number of degree 3 nodes. 
Thus, if this new tree has $N'$ nodes, then at least $N'/2+1$ of the nodes are leaves.

Now, assume each leaf corresponds to a single degree 1 vertex in the original tree. 
% Assuming $N>2$, 
Since each leaf's neighbor must be a degree 3 vertex, and assuming we are not in the degenerate case of a star graph with 4 nodes (in which the result trivially holds), at most two leaves can share a common neighbor.
% each leaf shares a neighbor with at most one other leaf.  
% In the worst case, every leaf shares a neighbor with another leaf and so only one of the two leaves in each pair will combine with the neighbor. 
This implies that at least half of the leaves can be combined with their neighbors. Given that there are at least $N'/2+1$ leaves, there are at least $N'/4$ combinations. 

Next, we reintroduce the collapsed degree 2 vertices (the original chains) to complete the argument.
% To complete the proof we must consider that we substituted each connected subgraph of degree 1 and 2 vertices (chain) with a single leaf or edge. Consider reintroducing each of these original chains.
%
First, consider the chains that were a single degree 2 vertex replaced by an edge. Note that these cannot participate in combinations.
Since each such vertex must have two degree 3 neighbors, the number of such vertices is at most one less than the number of degree 3 vertices. And since the number of degree 3 vertices is less than the number of degree 1 vertices, these degree 2 vertices constitute at most $1/3$ of the vertices in the new tree. Therefore, since at least $N'/4$ combinations occur in the other at least $2/3$ of the vertices, at least $N/6$ combinations occur in the new tree.

Next, consider longer chains that were replaced by an edge but contained two or more degree 2 vertices. Reintroducing these chains does not eliminate any combinations already counted. Furthermore, due to the maximality condition, at least $1/3$ of the vertices in each such chain will still participate in a combination.
Thus, the overall ratio of combinations to vertices remains at least $1/6$.

Now consider chains replaced by a single leaf. These consisted of zero or more degree 2 vertices followed by a degree 1 vertex.
For leaves that did not share a neighbor with another leaf, reintroducing the chain, in which the number of combinations is at least $1/3$ of the number of vertices, will not decrease the ratio of combinations to vertices below $1/6$.
For leaf pairs that shared a neighbor, reintroducing both chains may eliminate one previously counted combination. However, since in each new chain there is at least one combination, we can reassign one of the combinations in the shorter chain and let it account for the removed combination.
The longer chain still contains at least $1/3$ combinations relative to its size, which is at least $1/6$ of the total number of vertices in both chains.
Thus, the overall ratio of combinations to vertices remains at least $1/6$, which completes the proof.
\end{proof}

The previous lemma immediately implies bounds on the height and space of topology trees.

\begin{lemma}\label{lem:topology_height}
    The number of levels in a topology tree for a tree with $n$ vertices is $O(\log n)$.
\end{lemma}
\begin{proof}
   Since Lemma~\ref{lem:topology_const_frac} proved that each level has at most $5/6$ the number of clusters of the previous level, the total number of levels is at most $\log_{6/5}n$.
\end{proof}

% Since the number of clusters geometrically decreases from level to level, the number of clusters is $O(n)$, giving the following lemma.

\begin{lemma}\label{lem:topology_space}
    The space usage of a topology tree for an input tree with $n$ vertices is $O(n)$.
\end{lemma}

\begin{proof}
    From Lemma~\ref{lem:topology_const_frac}, we have that $(1/6)$ of the clusters contract at each level. Consequently, the total number of clusters is upper bounded by the sum of the geometric series $\sum_{i=0}^\infty n \cdot (5/6)^i \leq 6n$.
    Since each cluster uses up at most $O(1)$ space, the result follows.
\end{proof}

Lemmas~\ref{lem:topology_height} and \ref{lem:topology_space} together complete the proof of Theorem~\ref{thm:topology_tree_properties}.

\subsection{Proofs for the Update Efficiency of Sequential Topology Trees}\label{app:topology_update_proof}

\begin{restatable}{lemma}{topologyrootclusters}\label{lem:topology_root_clusters}
    During an edge insertion or deletion in a topology tree, there are $O(1)$ root clusters at any level.
\end{restatable}

\begin{restatable}{lemma}{topologydeleted}\label{lem:topology_deleted}
    During an edge insertion or deletion in a topology tree, there are $O(1)$ clusters that are deleted at any level.
\end{restatable}

We now provide proofs of the Lemmas ~\ref{lem:topology_root_clusters} and \ref{lem:topology_deleted}. 
To prove that updates take logarithmic time, we need to prove that during an update the number of root clusters and the number of deleted clusters at a particular level is constant.
Our proof for this is heavily inspired by and somewhat analogous to the analysis of updates in rake-compress trees given by Anderson~\cite{anderson2023parallel}. We restate the lemmas here for convenience.

The proof follows a level-by-level approach and shows that at each level at most a constant number of root clusters are added, while a constant fraction of them contract away. This is the key fact which implies that there are $O(1)$ root clusters at the next level.

To aid in presentation, we will start calling root clusters \defn{affected} clusters. 
We choose to use the same terminology used in \cite{anderson2023parallel}, as the notions of ``root clusters'' and ``affected clusters'' seem to serve the same purpose in topology trees and rake-compress trees respectively. 
We will also refer to the forest induced by the topology tree formed via the contractions that took place prior to the update as $T$, and the forest induced after the update as $T'$. To refer to the forest at a specific round $i$ of the contraction/update process, we will use the notation $T_i$ and $T_i'$.

We now define a notion of affection and spreading affection for topology trees.

\begin{definition}[Affected Cluster]\label{def_affection}
    A topology tree cluster in $T'$ is said to be affected if any of the following are true:
    \begin{enumerate}
        
        \item The cluster has at least 1 child affected cluster. % Affected cluster never gets unaffected and any cluster formed from merger of affected clusters is also affected.
        
        \item The cluster it merged with in the previous set of contractions does not exist. % Root clusters formed by ancestor deletion.

        \item The cluster did not contract in this round and is adjacent to an affected cluster that it can contract with. % The to preserve maximality step.

        \item It is one of the initial 2 clusters specified in the link/cut operation. % Is this the only time a simple adjacency list modification creates root clusters? This could be why RC Trees are slower and consider more affected vertices per level
        
    \end{enumerate}
\end{definition}

The motivation for this definition is to provide an upper bound on the number of root clusters during the update algorithm. We will show that the set of root clusters during an update is a subset of the affected clusters, therefore by bounding the number of affected clusters, we also have a bound for the number of root clusters.

In Definition~\ref{def_affection}, the first case accounts for both root clusters at level $i$ created by a combination of root clusters at level $i - 1$ and root clusters at level $i$ that are extensions of the root clusters at level $i - 1$ that did not merge with anything.
The second case accounts for root clusters created by the remove ancestors step. Note that for these clusters, the remove ancestors step deletes the cluster that they clustered with.
% So when the update algorithm reaches this round, either the cluster these clusters merged with is still the same {\color{red} (made up of the same descendants?)} and affected or is a different affected topology cluster. (proved in lemma \ref{lem:spreading})
%
The third case in the definition accounts for the step where a non-root cluster is forced to merge with a root cluster to preserve maximality.
Finally, the last case accounts for the initial $2$ clusters that are affected.

% The definitions of affection and of spreading affection need to be defined for how doing contractions adds affected clusters that are considered in the next round.

We will now define how affection ``spreads'' from an affected cluster to an unaffected neighbor in the topology tree. 
The idea behind this is that even though the definition of affected clusters is static (i.e., any cluster that satisfies the definition above becomes affected), we can relate the new affected clusters at level $i + 1$ (that were not affected due to rule (1) of Definition~\ref{def_affection}) to affected clusters at level $i + 1$ created via rule $(1)$ (these are the root clusters created as a result of the contractions of root clusters at level $i$). In some sense, we can now think of the ``new'' root clusters as being added as a consequence of the affected clusters at level $i$. 
Ultimately, this serves to show how many clusters become ``new'' root clusters between $2$ levels and that this number is upper bounded by a constant.

A similar idea was also introduced in the analogous proof for RC trees in \cite{anderson2023parallel}, but our definition has to account for the fact that unlike RC Trees, neighbors of a cluster in topology trees are not necessarily the same after contractions take place at a level. 
As such, we view round $i$ as all operations that take place starting with the contractions at level $i$  until the beginning of the contractions at level $i + 1$. 
Accordingly, the process of spreading takes place once the contractions at level $i$ have taken place and the clusters/vertices of the topology tree at level $i + 1$ have been determined (see Figure~\ref{fig:round_i}). 
We now define how affection spreads via the lemma below which illustrates the properties of an affected cluster spreading affection to unaffected ones.

For clarification, in the lemma below, we say that a cluster $n$ in $T'_i$ does not exist in $T_i$ if it is an affected cluster in $T_i'$. 

%\quinten{reword this stuff to make it sound less like things happen over time, since our definition of affected is a static property. Maybe after this lemma then introduce the notion of "spreading" over time.}

\begin{lemma}\label{lem:spreading}
If an unaffected cluster $c$ becomes affected at round $i$ (is affected at round $i + 1$) and $c$ did not become affected via a link/cut operation, it has an affected neighbor $n$ at round $i$ such that one of the following is true:
\begin{enumerate}
    \item $n$ did not exist in $T_i$ and the cluster that $c$ contracted with in $T_i$ was deleted. \label{rule:spreading1}

    \item $n$ and $c$ can merge together and $n$ did not contract in $T_i$. \label{rule:spreading2}
\end{enumerate}
\end{lemma}

\begin{proof}
    Consider a cluster $c$ that ``becomes affected'' during round $i$ (i.e., it is an affected/root cluster during round $i+1$). It clearly has at least one neighbor during round $i$; otherwise it would finalize and remain unaffected.

    Since $c$ does not have any affected descendants (we do not consider this case as this process does not add an affected cluster to level $i + 1$), and it becomes affected in round $i$, it must become affected via rule (2) or (3) of Definition~\ref{def_affection}.
    Rule (3) requires that the cluster in question has an adjacent affected cluster. Thus, if $c$ becomes affected via rule (3), then case (2) of spreading affection (this lemma) is true.

    Suppose for the sake of contradiction that $c$ becomes affected by rule (2), i.e., $c$ becomes affected as a result of the cluster it previously merged with not existing anymore, and additionally that it is not adjacent to an affected cluster.
    Since $c$ was not affected before round $i$, it must have been formed using the same contractions as in $T$ and, as such, must have the same neighbors and degree in both $T_i$ and $T_i'$.
    Since it has the same neighbors, and the neighbor it contracted with does not exist in $T_i'$, it must be adjacent to an affected neighbor that has replaced the cluster it contracted with in $T_i$. This neighbor must be affected since any cluster not in $T_i$ is affected, leading to a contradiction.
    Therefore either (1) or (2) in this lemma is true.
\end{proof}

Lemma~\ref{lem:spreading} shows how affected clusters at level $i$ can create affected clusters at level $i + 1$. We will now prove that only a constant number of affected clusters can be added during each round.
Let $A^i$ denote the set of affected clusters at round $i$.

%\quinten{The rest of this section after this point is straight lemmas. Can we add some sentences in between to improve readability or at least describe the flow of all of the lemmas here?}

%\begin{figure}
%    \centering
%    \includegraphics[width=0.4\textwidth]{figures/affection_rules.jpg}
%    \caption{Illustrations of the first $3$ cases of the defintion of Affected Clusters}
%    \label{fig:affection_rules}
%\end{figure}
%

\begin{figure}
    \centering
    \includegraphics[width=0.45\textwidth]{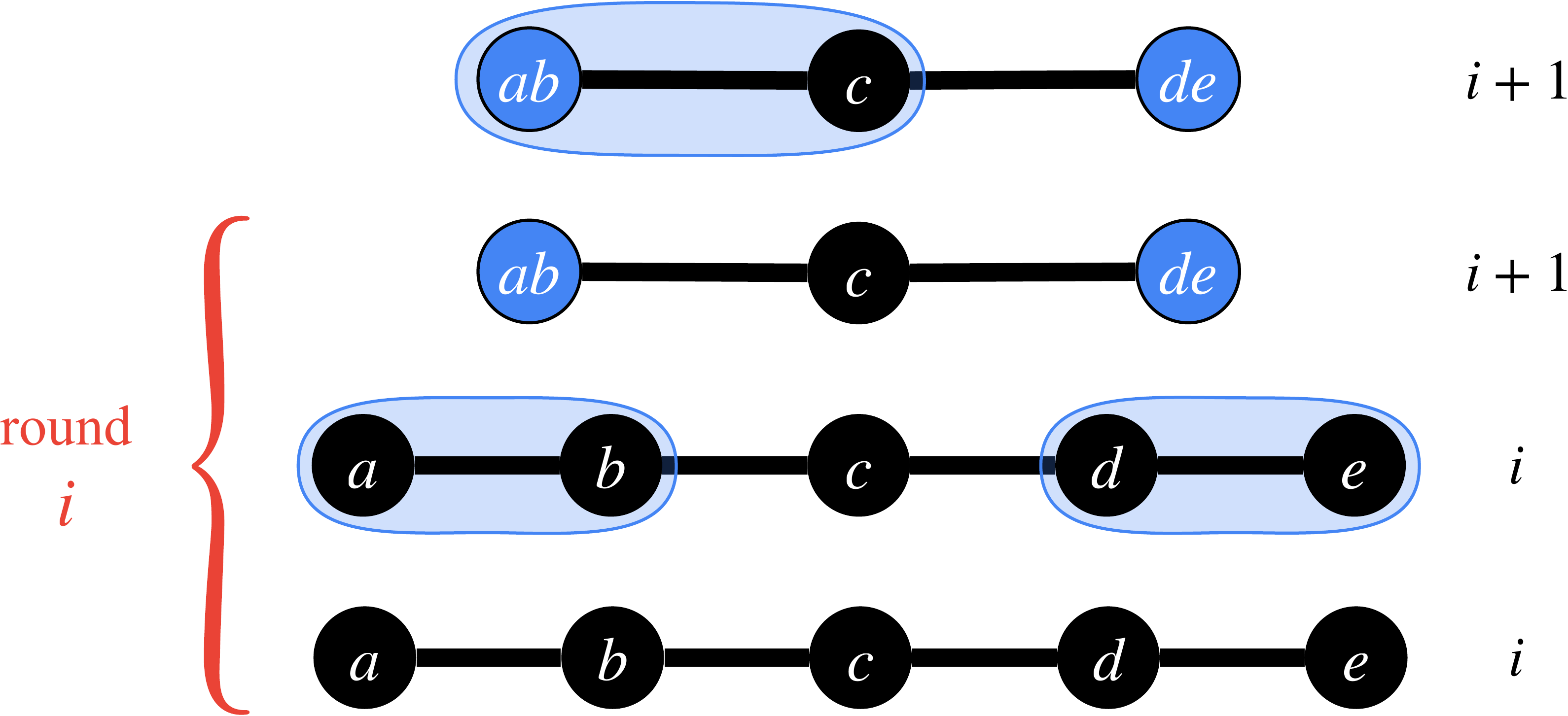}
    \caption{The operations that compose a round $i$ of contractions.}
    \label{fig:round_i}
\end{figure}

\begin{lemma}\label{lem:initialaffected}
At level $0$ of a topology tree, an update introduces at most $10$ affected clusters, i.e., $A^0 \le 10$.
\end{lemma}
\begin{proof}
The initial 2 clusters specified in the link/cut operation are affected via rule (4). In $T_0$, they both contracted with at most $1$ cluster each, which by rule (2) of Lemma~\ref{lem:spreading} becomes affected. Each of the clusters affected via rule (2) could then spread affection to at most two more clusters by rule (3) of Lemma~\ref{lem:spreading} while those affected by the original link/cut can spread affection to at most $1$ other cluster each ($1$ of their neighbors is affected via rule $4$ and the other via rule $2$). No further affection can be spread from the clusters added via rule (3) as they did not merge with any other cluster and are not adjacent to any others that they could have merged with in $T_i$. Thus $|A^0| \le 2 + 2 + 2 + 4$, i.e., $|A^0| \le 10$.
\end{proof}

Each of these $10$ affected clusters in $A^0$ will be considered as an \defn{origin cluster} of affection. 
Based on these, we will partition our affected forest into disjoint sub-components, each of which stores the affected clusters growing out of one of the origin clusters. Since there are at most $10$ such clusters, there are at most $10$ sub-components in any affected forest.
We will represent the set of affected clusters originating from the $j^{\text{th}}$ origin cluster as $A^i_j$.
%In the analysis we consider each sub-component independently, and we note that this may lead to a large amount of over-counting.

We now define the notion of \defn{frontier clusters}. These are clusters in an affected forest capable of spreading affection to unaffected clusters and are defined as follows:

% \quinten{Does this include clusters that do not exist anymore? This maybe should read "any affected cluster in the new tree..." - } {\atharva{Resolved}}
\begin{definition}[Frontier clusters]\label{def:frontiers}
    A frontier cluster is any affected cluster in $T_i'$ that is adjacent to an unaffected cluster. 
\end{definition}

Consider why these components must be disjoint: when each of these origin clusters spreads affection, it does so in the direction of unaffected clusters. As such, if $2$ origin clusters are connected, they effectively block each other from spreading affection in each other's sub-component. 

We now show that the number of affected clusters added in any sub-component is constant. 
Below, we first establish Lemma~\ref{lem:deg_3_affected}, which an intermediate result that is helpful to refer to in the proof of Lemma~\ref{lem:constant_frontiers}.
Lemma~\ref{lem:constant_frontiers} critically establishes that the number of frontier clusters is constant and thus the number of clusters that can become affected is constant in any sub-component (which we formalize in Lemma~\ref{lem:affectedadded_subc}). 

\begin{lemma}\label{lem:deg_3_affected}
Any degree $3$ affected cluster has at least $1$ affected neighbor.
\end{lemma}
\begin{proof}
    Since a degree $3$ cluster cannot be formed via a cut (due to our degree bound of $3$), it must be affected via either rule (4) of Definition~\ref{def_affection}, or was affected by an adjacent cluster (Lemma~\ref{lem:spreading}). 
    In either case, by the definition of affection, it has $\ge 1$ affected neighbor.
\end{proof}

\begin{lemma}\label{lem:constant_frontiers}
    Each component of the affected forest has at most $2$ frontier clusters at round $i$.
\end{lemma}
\begin{proof}
    We proceed by induction. At $A_j^0$ we see that this is trivially true as each sub-component only contains $1$ affected cluster.

    Consider a sub component at round $i$. Suppose the sub-component has only $1$ frontier cluster. If the frontier cluster is degree $1$ it can spread to at most $1$ unaffected neighbor and when it does it is no longer a frontier so the number of frontiers remains the same. If it is degree $2$ or $3$ it has only 2 unaffected neighbors (by Lemma~\ref{lem:deg_3_affected} for the degree $3$ case) and can spread to at most these $2$ unaffected neighbors. If it spreads to $1$ neighbor then we have added $1$ new frontier and the no. of frontiers is $\le 2$. If it spreads to both adjacent unaffected neighbors the original frontier cluster is no longer a frontier and the no. of frontier clusters is still $\le 2$.

    Now consider a sub component with $2$ frontier clusters, $c_1$ and $c_2$. Due to the fact that the partitioning scheme preserves the underlying connectedness of the affected subtrees, we note that both $c_1$ and $c_2$ are adjacent to at least $1$ affected cluster (the cluster joining it to its affected component subtree). 
    Suppose that a frontier cluster is degree $2$ (we cannot have a degree $1$ frontier cluster in this case as it has only 1 neighbor which is affected and as such it cannot be a frontier). This degree $2$ cluster has $1$ affected neighbor and spreads to at most $1$ new unaffected cluster (by Lemma~\ref{lem:spreading}). When it spreads to a neighboring cluster, all of its neighbors become affected and hence it is no longer an frontier cluster. If both $c_1$ and $c_2$ do this then the no. of frontier clusters is still at most $2$.

    Suppose now that we have a degree $3$ frontier cluster. 
    A clear issue with spreading affection via the rules described in \ref{lem:spreading}, is that if both $c_1$ and $c_2$ are degree $3$ frontier clusters then if they each spread affection to both neighbors then the total no. of frontiers increases by $2$. We show that this is not possible.

    Consider the first instance in which a degree $3$ cluster gets affected. 
    Either the degree $1$ cluster it clustered with in $T_i$ is now replaced with some affected cluster (rule (\ref{rule:spreading1}) of  Lemma~\ref{lem:spreading}), or it is spread affection to by an adjacent affected degree $1$ cluster (rule (\ref{rule:spreading2}) of lemma \ref{lem:spreading}).
    The latter case is not relevant to us as in it we only have $1$ frontier cluster. 
    Similarly, the case where the cluster it contracted with in $T_i$ is affected via a link/cut is also not something we have to consider, as we have shown previously that origin clusters spread affection in a manner disjoint to each other.
    Now suppose a degree $3$ cluster $u$ gets affected as the consequence of a frontier cluster (WLOG $c_1$) replacing some degree $1$ cluster $v$ that $u$ had clustered with in $T_i$ (see Figure~\ref{fig:deg3_frontier}).
    \begin{figure}
        \centering
        \includegraphics[width=0.9\linewidth]{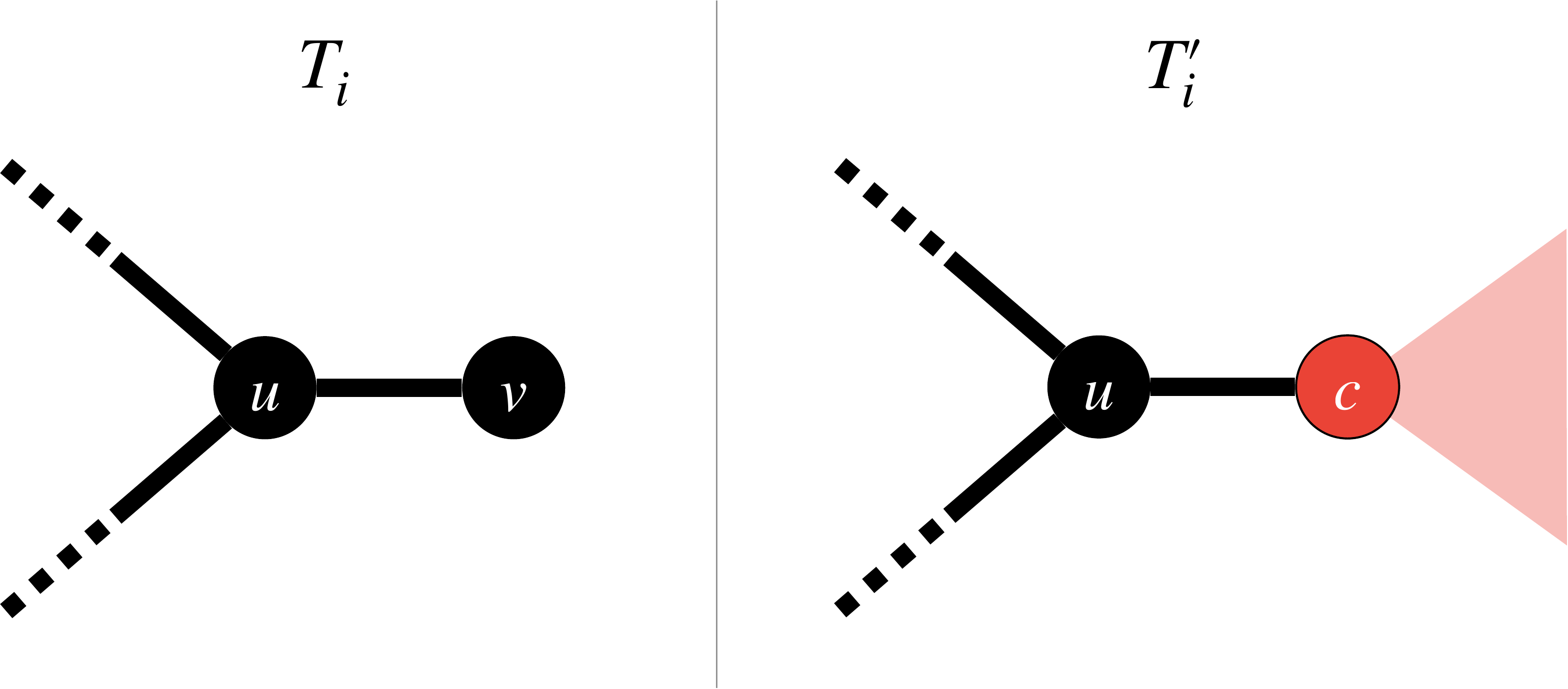}
        \caption{A depiction of affection spreading to a degree 3 cluster $u$ from frontier cluster $c$. If this happens, it must be the case that $c$ was the only frontier cluster.
       % \kishen{Is this supposed to be $c$ here?}
       }
        \label{fig:deg3_frontier}
    \end{figure}
     We know that the contraction operations preserve the connectedness of the underlying subtree and that all clusters in this sub-component that had a path to $u$ via either $v$ or an ancestor of $v$ in $T_i$ must have compressed into the degree $1$ cluster that $c_1$ has replaced. 
    As a consequence, all clusters that are in this sub-component and have a path to $u$ via $c_1$ in $T_i'$ must be affected. 
    We further explain this rationale below: Since all of these clusters contracted with each other and the origin cluster to form a degree $1$ cluster in $T_i$, if they all contracted with each other in the same manner again then they have contracted with at least $1$ affected cluster and as such $c_1$ is now a degree $1$ frontier which we do not consider. 
    If a different series of contractions have taken place then since each of these clusters contracted with each other, the clusters they contracted with will have inevitably been replaced and as such they are now affected.
    Thus, when the degree $3$ cluster becomes a frontier, it must be the only frontier cluster in its sub-component and thus we are back to our case with $1$ frontier cluster.

    This is not to say we cannot have a sub-component with $2$ frontier clusters where $1$ of them is a degree $3$ cluster.
    A degree $3$ cluster can be a frontier cluster and spread affection to $1$ of its neighbors in a round and only spread affection in that branch and remain as a frontier cluster. 
    However, as shown above, degree $2$ vertices cannot spread affection in a way that causes the no. of frontier clusters to be $> 2$ and we cannot add a degree $3$ affected cluster without causing all other vertices in the sub-component to be affected first. 
    Furthermore, when the degree $3$ cluster spreads affection to its other branch, it is no longer a frontier and thus the no. of frontiers remains $\le 2$.
\end{proof}

\begin{lemma}\label{lem:affectedadded_subc}
    For any subcomponent $A_j^i$, $|A_j^{i + 1}| - |A_j^i| \le 8$.
\end{lemma}
\begin{proof}
    By Lemma~\ref{lem:constant_frontiers}, we have at most $2$ frontier clusters. From Lemma~\ref{lem:deg_3_affected}, we know that each of these clusters can have at most 2 unaffected neighbors. 
    Affection spreads to at most these unaffected neighbors. 
    Since one of these was spread via rule (\ref{rule:spreading1}) and the other via rule (\ref{rule:spreading2}), one of them can possibly add an additional $2$ affected neighbors via rule (\ref{rule:spreading2})  of \ref{lem:spreading} (which cannot add any further neighbors as they did not merge with anything at this level). Thus the total no. of affected vertices added in a sub-component per round is at most $2 \times (2 + 2) = 8$
\end{proof}

\begin{lemma}\label{lem:component_affected}
    The total no. of affected clusters at any level $i$ and any sub-component $j$ is at most $48$ i.e for $i \ge 0, A^i_j \le 48$
\end{lemma}
\begin{proof}
    Since any affected sub-component  consists of only $2$ frontier clusters, the remainder of the affected sub-component can be thought of as its own connected subtree. 
    Consequently, since a maximal number of contractions must take place in this subtree, at least $1/6 - 2$ of the clusters must contract away. 
    Between any $2$ levels we add at most $8$ new root clusters. These 2 facts in combination give us the following inequality:
    $$A^{i + 1}_j \le 8 + \left(\frac{5}{6}\right)(A^i_j - 2) + 2.$$
    \noindent Thus, since $A^0_j = 1$,
    $$A^{i+1}_j \le \frac{50}{6} \cdot \sum_{r=0}^{\infty} \left(\frac{5}{6}\right)^r \le 50.$$
\end{proof}

We can now complete the proofs of the lemmas stated at the beginning of this analysis: \ref{lem:topology_root_clusters} and \ref{lem:topology_deleted}.

\topologyrootclusters*
\begin{proof}
    There are at most $10$ subcomponents and Lemma~\ref{lem:component_affected} proves that each subcomponent can have at most $A_j^i \le 50$. Therefore the total number of affected vertices per level is at most $10 \cdot 50 \le 500$.
    % \atharva{Why do they double this value in R.C. trees, it doesn't really make sense to do this for topology trees.}
\end{proof}

\topologydeleted*
\begin{proof}
    We utilize a simple level by level approach to prove this lemma. 
    Consider level $0$ in $T'$. A link/cut creates at most $10$ root clusters, each of which made a call to remove ancestors and deleted its respective parent in $T_1$. 
    As a consequence, of Lemma~\ref{lem:topology_const_frac}, $1/6$ of these must have previously contracted and as such at most $5/6$ of these cluster i.e. their parents in $T_1$ were deleted. 
    Consider now, what happens at $T_1$. By Lemma~\ref{lem:affectedadded_subc}, $8$ new clusters are added per subcomponent, of which at most $6$ clusters are added via rule (3) of Definition~\ref{def_affection} and as such call remove ancestors. 
    Since there are $10$ subcomponents, $\le 60$ clusters are added in this manner. 
    Thus the number of parents deleted at level $2$ is at most $5/6$ of the sum of number of clusters deleted at the previous level plus $60$.
    This pattern continues at each level up, thus the number of deletions per level is bounded by a constant:

    $$\text{\# Deletions} \le 60 \cdot \sum_{r=0}^{\infty} \left(\frac{5}{6}\right)^r\le 360.$$
\end{proof}

The proof of lemmas \ref{lem:topology_root_clusters} and \ref{lem:topology_deleted} now allow us to complete the proof of the total cost of the update algorithm for topology trees:

\begin{lemma}\label{lem:topology_update_cost}
    The cost of edge insertions and deletions in a topology tree with $n$ vertices is $O(\log n)$.
\end{lemma}
\begin{proof}
    At any level, each root cluster may examine its (at most 3) neighbors and create a new cluster, while some clusters may be deleted. By Lemma~\ref{lem:topology_root_clusters}, the number of root clusters per level is constant, and by Lemma~\ref{lem:topology_deleted}, the number of deleted clusters per level is also constant. Since there are $O(\log n)$ levels in the topology tree (Lemma~\ref{lem:topology_height}), the overall cost of an update is $O(\log n)$.
\end{proof}

To conclude this section, we provide a proof which shows that the topology tree obtained following an update is a valid topology tree.

\begin{lemma}\label{lem:topology_update_correct}
    The algorithm for edge insertion or deletion in a topology tree results in a valid topology tree following the update.
\end{lemma}

\begin{proof}
    We define a valid topology tree as one that satisfies the definition given in Section~\ref{sec:topology_definition}.
    
    First, note that after any remove ancestor call, root clusters can only combine with other clusters that were either not deleted or ones that they became adjacent to as a consequence of insertions into their adjacency list from a previous level.
    Thus, clusters can only be combined with other clusters that they are adjacent to at their level during reclustering.
    Furthermore, a cluster is inserted into the adjacency list of another cluster if and only if it was either adjacent to a child of that cluster or vice versa.
    Consequently, the degree of any cluster never grows to be $> 3$.
    Second, combinations are only permitted in the way listed in Section~\ref{sec:topology_definition}.
    Thus, we have that, at any level, combinations take place between adjacent clusters and only in such a way that satisfies the definition of topology trees.
    Lastly, to prove that the clustering at each level is maximal, we will show that every cluster obeys maximality, i.e., either it combines with another cluster or all its neighbors combine with a different cluster.
    Here, we use \emph{new clusters} to denote any cluster in the new topology tree that was a root cluster at some point during the update algorithm. Other clusters in the new tree are referred to as \emph{old clusters}.
    It can be observed that any new cluster obeys maximality since the update algorithm checks all of its neighbors and attempts to combine it with them.
    
    For old clusters only adjacent to other old clusters, they must obey maximality because all of their neighbors are exactly the same as the old topology tree (the tree before the update), and the parent of an old cluster never gets deleted.
    For an old cluster $X$ adjacent to one or more new clusters, either $X$ combined with another old cluster and still does, or it is not able to contract with any old cluster. 
    In the latter case, if one of the adjacent new cluster(s) does not combine with anything else and it can combine with $X$, then it does. 
    Therefore, $X$ obeys maximality.
\end{proof}

%\myparagraph{Discussion}\label{top_proof:discussion}
%Readers of the analogous proof in \cite{anderson2024deterministic} for R.C. trees will be quick to note the similarity in its structure with the one described in this section for topology trees. However, while the idea of how constant affected clusters is proved is similar in both cases, there are some core differences between both proofs due to how the data structures function. In order to help the reader better understand the difference between R.C trees and Topology trees, and to show the merits of the method described above as a proof technique we highlight these similarities and differences below. To further show its utility as proof technique we utilize a similar approach in proving constant root clusters for UFO trees.
%

\subsection{Proofs of Theorems~\ref{thm:ufo_tree_properties} and \ref{thm:ufo_tree_diameter}}\label{app:ufo_efficiency_proofs}
\begin{lemma}\label{lem:ufo_const_frac}
    At any level $l>0$ of a UFO tree, the number of clusters is at most $5/6$ times the number of clusters at level $l-1$.
\end{lemma}
\begin{proof}
    The proof of this is similar to that of Lemma~\ref{lem:topology_const_frac}.
    Imagine replacing every chain in the tree with an edge or vertex in the same way as Lemma~\ref{lem:topology_const_frac}. In the resulting tree of $N$ nodes (which contains no degree 2 nodes), the number of leaves is at least 2 more than the number of other nodes, so at least $N/2+1$ of the nodes are leaves. In a UFO tree, all of these leaves combine with their neighbor (unless the tree is just two nodes in which case the lemma is trivial), since the neighbor of each leaf is high degree. Therefore at least $N/4$ combinations occur.
    The same logic as Lemma~\ref{lem:topology_const_frac} applies for reintroducing the chains. This proves that at least $N/6$ combinations occur.
\end{proof}

The previous lemma immediately implies bounds on the height and space of UFO trees.

\begin{lemma}\label{lem:ufo_height}
    The number of levels in a UFO tree for a tree with $n$ vertices is $O(\log n)$.
\end{lemma}
\begin{proof}
    Since Lemma~\ref{lem:ufo_const_frac} proved that each level has at most $5/6$ the number of clusters of the previous level, the total number of levels is at most $\log_{6/5}n$.
\end{proof}

\begin{lemma}\label{lem:ufo_space}
    The space usage of a UFO tree for an input tree with $n$ vertices is $O(n)$.
\end{lemma}
\begin{proof}
    At each level of the UFO tree with $x$ clusters, the space is $O(x)$ since the induced subgraph is a tree, each edge is stored twice, and each cluster uses $O(1)$ additional space. Combining this with Lemma~\ref{lem:ufo_const_frac} gives us that the total space is bounded by $\sum_{i=0}^\infty O(n) \cdot (5/6)^i \leq 6\cdot O(n)$.
\end{proof}

Lemmas~\ref{lem:ufo_const_frac}, \ref{lem:ufo_height} and \ref{lem:ufo_space} complete the proof of Theorem~\ref{thm:ufo_tree_properties}.

% \laxman{In the argument below, where do we use any property of UFO trees? Doesn't the argument also apply to topology trees?}\kishen{In topology trees all leaves don't get raked, so this may not be guaranteed. But I think we can still get $O(d)$ bound in Topology trees: a leaf not raked in one round will be raked in the next, unless some compress happened. However, ternarization increases the diameter arbitrarily, so this is not very useful.}\quinten{Any ternarized (binary) tree has diameter $\Omega(\log n)$.}
\vspace{1em}
\begin{proof}[Proof of Theorem~\ref{thm:ufo_tree_diameter}]
   Consider any maximal path in the tree of length $d$ (i.e., no additional vertices can be added to extend the path). The endpoints of this path must be leaves. Focus on one of the endpoints, $x$, and let its neighbor on the path be $y$.
   If $y$ is of high degree or degree 1, then $x$ combines with $y$, reducing the length of the path by at least one.
   If $y$ is of degree 2, then either $x$ combines with $y$ or $y$ combines with its other neighbor $z$, which also lies on the path. In both cases, the length of the path decreases by at least one. The same argument applies symmetrically to the other endpoint of the path. 
   Thus, in each level of the UFO tree, the path shortens by at least two vertices. Therefore, the number of levels needed to fully contract the path is at most $\lceil d/2 \rceil$. Since the longest path in the input tree has length $\diam$, the bound follows. 
   % \atharva{Finish this}

%\atharva{
%\begin{enumerate}
%    \item Make clear why this does not disagree with the cell probe model lower bound
%    \item We should make clear why the fact that the longest path takes $\ge d/2$ time to contract implies that the UFO tree takes at most $d/2$ time to contract. Maybe the right way to phrase this argument is in terms of assuming only the leaves contract at every iteration.
%\end{enumerate}
%}
\end{proof}

\subsection{Sequential UFO Tree Proofs}\label{app:ufo_proofs}
%\quinten{Can we add a quick proof justifying the correctness of the UFO tree update algorithm? I think there might be some text for this commented out in the main paper file i.e. a version of lemma b.8 here}

As a result of the overlap in properties of UFO and Topology trees, many lemmas can be expressed with the same reasoning as those for topology trees with some minor differences. 
In this section, lemmas \ref{lem:ufo_root_clusters} -\ref{lem:ufo_deleted} establish the efficiency of the update algorithm and lemma \ref{lem:ufo_update_correct} establishes the correctness of the update algorithm.

For clarity, in the lemmas presented below we will refer to the UFO tree created before the update as $T$ and the one after the update as $T'$. Similarly, the tree induced at each round is expressed using $T_i$ and $T_i'$.

\begin{restatable}{lemma}{uforootclusters}\label{lem:ufo_root_clusters}
    During an edge insertion or deletion in a UFO tree there are $O(1)$ root clusters at any level.
\end{restatable}

Borrowing from our topology tree proof, we approach this proof using a level-by-level approach. 
%As such, we begin with defining a notion of affection for clusters in UFO trees:

% \begin{definition}
%     A UFO tree cluster is said to be affected if any of the following about it are true

%     \begin{enumerate}
%         \item The cluster has atleast $1$ child affected cluster. 

%         \item The cluster it combined with in the previous UFO tree is affected 

%         \item The cluster does not combine with anything at this level and is adjacent to a cluster that it can combine with

%         \item The cluster is part of the initial link/cut update operation.
%     \end{enumerate}
% \end{definition}

%This definition of affection is again directly based on the update algorithm and the "affected" clusters represent those clusters for which computations need to be redone (i.e. root clusters).
%Our definition of affection differs from our update algorithm in $1$ key respect i.e. high degree clusters can get affected, despite the fact that they do not become root clusters in the update algorithm. 
%This step is necessary to allow the low degree parents of these high degree clusters to become affected. 
 
%All of the rules $(1) - (4)$ account for the root clusters created in the manner akin to the topology tree update algorithm. 
%The additional steps taken for high fan-out and high degree clusters by the update algorithm are accounted for with the modifications to rule $(2)$. 
%Since high degree and high fan-out clusters are ignored by the update algorithm only degree $1$ clusters with low degree and low fan-out parents become root clusters. 

However, we no longer treat affection as a static property of clusters during an update but rather something that is spread by the initial affected clusters during a link/cut.
Effectively, we show how the existence of certain root clusters implies the existence of other root clusters at that level and that this considers all root clusters that can be created by the update algorithm.
The variation in our approach is due to the trickiness in defining a static definition for affection for UFO trees due to how the high degree and high fan-out cases are handled by our algorithm.

In the rest of the proof we slightly change the meaning of the terms {\bf low degree} and {\bf high degree} from our update algorithm - low degree now refers to a cluster with degree $\le 3$ and high degree a cluster with degree $\ge 4$.
\begin{definition}\label{lem:ufo_spreading}
    If a UFO tree cluster $c$ was unaffected at the beginning of round $i$, is affected at the beginning of round $i + 1$ (got affected during round $i$), and was not part of the initial link/cut operation, it must have an affected neighbor $n$ such that one of the following is true.
    
    \begin{enumerate}
        \item $n$ is low degree, did not exist in $T_i$, and the cluster $c$ combined with was deleted.
        \item $n$ is low degree and $c$ did not combine in $T_i$ but can now merge with $n$.
    \end{enumerate}
\end{definition}

Additionally, similar to topology trees, the clusters representing the endpoints of the update operation are affected, and any cluster with at least one affected child is affected.

Based on the lemma above, the consequent idea for the rest of the proof is that we can effectively use the previous affected clusters to figure out which of the unaffected clusters can become affected.
This is known as spreading affection.
This will serve as an upper bound on the no. of root clusters and allows us to bound this by a constant amount. 
Notice that high degree (degree $\ge 4$) clusters may become affected, but they cannot spread affection further to their unaffected neighbors. 
This follows from our algorithm as any such cluster would have either a high fan-out or degree $\ge 3$ parent which will not be deleted by the update algorithm.
%\quinten{Here it refers to $d\geq3$ right?} 
All other cases of spreading affection are analogous to topology trees.

Similar to our proof for Lemma ~\ref{lem:topology_root_clusters} we will now refer to the set of affected vertices at round $i$ as $A^i$.

\begin{lemma}\label{lem:ufo_init_affected}
    Following the initial link/cut operation, we can have at most $10$ affected vertices in round $0$ i.e. $A^0 \le 10$.
\end{lemma}
\begin{proof}
    This lemma follows from Lemma ~\ref{lem:initialaffected} for topology trees, as degree $\ge 4$ vertices cannot spread affection and in all other cases the spreading step is equivalent to that of topology trees.
\end{proof}

Like topology trees we again partition $A^i$ into disjoint sub-components with our $10$ origin clusters of affection. Going forward we will refer to the vertices that became affected as a result of affection being spread from the $j^{\text th}$ origin cluster as $A^i_j$.

We will now again work with frontier clusters (refer to definition ~\ref{def:frontiers} and the text that follows it) to show how existing affected clusters imply that other unaffected clusters at a level are affected. 

\begin{lemma}
    In any level $i$ sub-component of the affected forest $A^i_j$, the number of frontier clusters $\le 2$. 
\end{lemma}
\begin{proof}
    We note that at the base level this is already true, as each sub-component has only $1$ affected cluster and hence only $1$ frontier.

    At any subsequent level $i$. Suppose some sub-component $A^i_j$ has just $1$ frontier cluster. The cases of the frontier being degree $1, $ or $2$ can be handled in the same manner as topology trees. 
    The degree $3$ case is slightly more tricky for UFO trees. In the case where a degree $3$ frontier has $1$ affected neighbor, this can be handled in a manner similar to topology trees. 
    In the case where it has no affected neighbors (for example it was formed after a degree $4$ cluster had $1$ of its neighbors cut), then since it was previously high degree it could have only combined with adjacent degree $1$ neighbors and can thus spread affection only to them.
    As these degree $1$ clusters cannot be frontiers, regardless of how this degree $3$ frontier spreads affection, it will be the only frontier till it contracts away.
    Degree $\ge 4$ clusters cannot spread affection and any such frontier will remain so until it clusters into something that can spread affection.
    These cases are handled as shown previously in this lemma.

    Now consider the case where we have $2$ frontiers. 
    As the clustering operations of the UFO tree preserve the underlying connectedness of the tree, the $2$ frontiers $f_1$ and $f_2$ must have at least $1$ affected neighbor. 
    The degree $1$ and $2$ cases are easily seen to be handled in the same manner as topology trees.
    Since degree $3$ clusters still combine only with degree $1$ clusters, it is again the case that if a degree $3$ cluster spreads affection, it is the only frontier in its component and thus this case is also handled similar to Lemma ~\ref{lem:constant_frontiers}. Lastly, high degree clusters cannot spread affection and thus if either $f_1$ or $f_2$ is such a frontier, it will not add any frontier until it clusters into a low degree cluster. These cases are again shown to be handled above.
    From the case-work above, it is evident that each sub-component has $\le 2$ frontier clusters.
\end{proof}

\begin{lemma}
    For any round $i$ and sub-component $j$ the no. of affected clusters added between any $2$ rounds is $\le 8$ i.e. $|A^{i+1}_j| - |A^i_j| \le 8$.
\end{lemma}
\begin{proof} 
    We naturally want to consider only cases where each frontier has $2$ unaffected neighbors that it can spread affection to as these would add more affected clusters than if there was only one such unaffected neighbor. 
    Consider the cases where a cluster can spread affection to both its adjacent unaffected neighbors. 
    It is either a lone degree $2$ frontier or a degree $3$ frontier with an affected neighbor. 
    The only manner in which spreading can take place differently from topology trees is if the degree $3$ cluster spreads affection to both of its neighbors by rule $1$ of lemma ~\ref{lem:ufo_spreading}.
    However, this can add potentially only $2$ new affected clusters, which is less than the $4$ that could be added in the other methods. The rest of the lemma follows from ~\ref{lem:affectedadded_subc}.
\end{proof}

\begin{lemma}
    The total no. of affected clusters at any level $i$ and any sub-component $j$ is at most $50$ i.e. for $i \ge 0, A^i_j \le 50$.
\end{lemma}
\begin{proof}
    The mathematics used in proving this lemma is exactly the same as in Lemma~\ref{lem:component_affected}.
\end{proof}

As such, since each sub-component has at most $48$ root clusters and there are $10$ sub-components, the total number of root clusters is $\le 500$. 

%{\atharva{Maybe add a line about how having greater average degree high degree clusters results in greater speed?}}

\begin{restatable}{lemma}{ufodeleted}\label{lem:ufo_deleted}
    During an edge insertion or deletion in a UFO tree there are $O(1)$ clusters that are deleted at any level.
\end{restatable}

\begin{proof}
    We utilize a simple level by level approach to prove this lemma. 
    Consider level $0$ in $T'$. A link/cut creates at most $10$ root clusters, each of which made a call to remove ancestors and deleted its respective parent in $T_1$. 
    As a consequence, of Lemma~\ref{lem:ufo_const_frac}, $1/6$ of these must have previously contracted. Thus, in the worst case, no parents on the leaf-to-root path were high degree and $5/6$ of these clusters i.e. their parents in $T_1$ were deleted. 
    
    Consider now, what happens at $T_1$. 
    By Lemma~\ref{lem:affectedadded_subc}, $8$ new clusters are added per sub-component, of which at most $6$ clusters are added via rule $(3)$ of definition~\ref{def_affection} and as such call remove ancestors. 
    Since there are $10$ sub-components, $\le 60$ clusters are added in this manner. 
    Thus the number of parents deleted at level $2$ is at most $5/6$ of the sum of number of clusters deleted at the previous level plus $60$.
    This pattern continues at each level up, thus the number of deletions per level is bounded by a constant:

    $$\text{\# Deletions} \le 60 \cdot \sum_{r=0}^{\infty} \left(\frac{5}{6}\right)^r\le 360.$$
\end{proof}

Lastly, we show that the UFO tree that is created following an update is a valid UFO tree.

\begin{lemma}\label{lem:ufo_update_correct}
    The algorithm for edge insertion or deletion in a UFO tree results in a valid new UFO tree. That is all of the combinations in the new tree are allowed combinations, and the combinations are maximal at each level.
\end{lemma}
 \begin{proof}
     The components of the UFO tree after deleting some paths of ancestors contain only allowed combinations of clusters. Additionally, each new combination resulting from the re-clustering phase of the update algorithm can only be an allowed combination. Thus all combinations in the new UFO tree are allowed.

     The combinations not involving high degree clusters are definitely maximal by the same reasoning as topology trees.
     The combinations involving high degree clusters are also maximal since any degree 1 root cluster will definitely combine with a neighboring high degree cluster if it exists.
     For any non-root degree 1 cluster $X$ adjacent to a previously and still high degree cluster $Y$, $X$ and $Y$ must combine since the previous tree had a maximal set of combinations at each level.
     For any non-root degree 1 cluster $X$ adjacent to a cluster $Y$ that just became high degree, $Y$ must be exactly degree 3. The algorithm will examine all of the neighbors of $Y$ and combine the degree 1 neighbors with it, so $X$ will be found and made to combine with $Y$.
 \end{proof}

\subsection{Batch-Update Algorithm Proofs}\label{app:batch_update}

\myparagraph{Parallel Batch-Dynamic Topology Trees}
We now prove the correctness of the batch update algorithm for topology trees.
% Lemma~\ref{lem:batch_topology_correct} proves that the update algorithm maintains all the invariants of topology trees, impying the correctness of the batch update algorithm.

\begin{lemma}\label{lem:batch_topology_correct}
    The algorithm for batch updates in a topology tree results in a valid topology tree following the update.
\end{lemma}
\begin{proof}
    Note that lazily deleting clusters has no affect on the correctness of the algorithm, since the clusters examined at any level remain the same. The rest of the batch update algorithm mirrors the sequential case, with modifications to parallelize certain operations. Therefore, the correctness argument follows directly from the proof of Lemma~\ref{lem:topology_update_correct}.
\end{proof}

Next, we will prove Lemma~\ref{lem:batch_topology_root_clusters}, which we will in turn use to prove Theorem~\ref{thm:batch_topology_cost}, our main theorem analyzing the cost of batch-updates in topology trees.

\begin{lemma}\label{lem:batch_topology_root_clusters}
    During a batch update in a topology tree with a batch size of $k$, there are $O(k)$ root clusters at any level. Additionally there are $O(k)$ deleted clusters at any level.
\end{lemma}
\begin{proof}
    The sets of root clusters and deleted clusters produced by the batch algorithm is a subset of all the clusters that would be produced by running any update in the batch individually.
    From Lemmas~\ref{lem:topology_root_clusters} and~\ref{lem:topology_deleted}, we know each individual update may contribute at most a constant number of root clusters and deleted clusters at each level, thus in total there are at most $O(k)$ at each level.
\end{proof}

\topologybatchupdate*
\begin{proof}
    We will show that the work done at each level of the topology tree is $O(x)$, where $x$ is the number of root clusters and deleted clusters at that level.
    Using this fact and Lemma~\ref{lem:batch_topology_root_clusters}, the work is immediately bounded by $O(k \log n)$. We get the better bound of $O(k \log (1+n/k))$ because above some threshold level $l=O(\log(1+n/k))$, the total number of nodes is $O(k)$, and below that threshold the work at each level sums to $O(k \log (1+n/k))$.

    In Algorithm~\ref{alg:batch_update_topology}, mapping to parents can be done in $O(x)$ work trivially.
    For updating adjacency lists, since every cluster is degree at most 3, there are $O(x)$ edges to insert delete, which can be done in $O(x)$ time easily. All of the batch insertions into neighbor sets take $O(x)$ work in total.
    Similarly, deleting from child sets takes $O(x)$ work in total for the same reason.
    The parallel maximal matching is known to take $O(x)$ work for chains.
    Finally, freeing $O(x)$ constant sized clusters and allocating $O(x)$ clusters takes $O(x)$ work.

    In the binary fork-join model, all of these steps can be done in $O(\log x)$ depth. Thus the total depth is $O(\log n \log k)$.
\end{proof}

\myparagraph{Parallel Batch-Dynamic UFO Trees}
Similar to the proof of Lemma~\ref{lem:batch_topology_correct}, the correctness of the UFO tree batch-update algorithm follows from the same reasoning as the UFO tree sequential update algorithm (Lemma~\ref{lem:ufo_update_correct}). All clusters that are not deleted still represent valid merges, and the maximal matching step ensures that we recluster the tree to create maximal contractions at every level.

Now we will prove Lemma~\ref{lem:batch_ufo_root_clusters} which will allow us to prove Theorem~\ref{thm:batch_ufo_cost}, the main theorem analyzing the cost of batch-updates in UFO trees.

\begin{lemma} \label{lem:batch_ufo_root_clusters}
    During a batch update in a UFO tree with a batch size of $k$, there are $O(k)$ root clusters at any level, and the sum of their degrees is $O(k)$.
    Additionally, there are $O(k)$ deleted clusters at any level, and the sum of their degrees is $O(k)$.
\end{lemma}
\begin{proof}
    By a similar argument as the proof of Lemma~\ref{lem:batch_topology_root_clusters}, this time applying Lemmas~\ref{lem:ufo_root_clusters} and~\ref{lem:ufo_deleted} the number of root clusters and deleted clusters at any level in the UFO tree batch-update algorithm is $O(k)$.

    Root clusters are formed when their parent is deleted.
    Since the algorithm only deletes clusters with low fanout and low degree, any children of those clusters must be or have been low degree (for insertion batches they must have been low degree before the insertion of new edges, for deletion batches they must be low degree after the deletion of the edges).
    Root clusters can only become high degree when some the edges in the batch of edge insertions are inserted incident to them. This can only increase the total degree by $2k$. Since all clusters were or are constant degree, and the total is increased by at most $2k$, the total degree of root clusters is $O(k)$.

    For batches of insertions, clusters that will be deleted are never updated with the new edges. Thus since we only delete low degree clusters, these clusters all have constant degree, and their total degree is $O(k)$.
    For batches of deletions, the $k$ edges are deleted from all levels of the tree beforehand (taking $O(k \log (1+n/k))$ work). All the clusters that are deleted are then low degree, thus the total degree at each level is $O(k)$.
\end{proof}

\ufobatchupdate*
\begin{proof}
    We will once again prove that the work at each level is $O(x)$, where $x$ is the number of root clusters and deleted clusters at a level. By the same reasoning as the proof of Theorem~\ref{thm:batch_topology_cost},this gives us the work bound of $O(k \log (1+n/k))$.
    The work bound of $O(k\diam)$ also follows from this, and the fact that the height is bounded by $O(\diam)$.
    
    All of the steps that are shared with the topology tree batch-update algorithm take $O(x)$ work.
    Because of Lemma~\ref{lem:batch_ufo_root_clusters}, the semisort to group the edges by endpoint takes $O(x)$ work, and so does adding the edges to the neighbor set hash tables. Updating the child set hash tables takes $O(x)$ because there are $O(x)$ root clusters and deleted clusters.

    For the depth, the only step that is different than parallel topology trees, is inserting and deleting elements in adjacency lists and child sets.
    The depth of inserting or deleting a batch of elements into a parallel set is proportional to the size of the set plus the size of the batch to insert or delete.
    Specifically, the depth is $O(\log y)$ where $y$ is the current size of the set plus the number of elements to insert or delete.
    This means that if there are updates to a large neighbor or child set with much more than $k$ elements, the depth of that step is higher than $O(\log k)$, possibly up to $O(\log n)$. Naively, this would result in $O(\log^2 n)$ depth.

    To reduce the depth to $O(\log n \log k)$, we can defer some of the batch updates until the end of batch update (or run them in the background while the rest of the algorithm proceeds at the levels above).
    Note that the only time we insert elements in a large neighbor set is if the cluster is already high degree. Since we are inserting, the cluster will remain high degree and its parent won't be deleted, so it does not become a root cluster. It is fine to delay the neighbor set update in this case because it does not become a root cluster so we won't iterate through its neighbors.
    If we are deleting elements from a neighbor set whose current size is $\Omega(k)$, since there are at most $O(k)$ deletions, the cluster will not become low degree, its parent will not get deleted, hence it will not become a root cluster. For the same reason we can delay this update.
    If the neighbor set's size is $O(k)$, then the depth is $O(\log k)$ and we can afford to do the update proactively.
    For child sets, we can delay the deletions and insertions until the end of the entire recluster process, because it is possible to determine the children of a parent cluster by iterating through the neighbors of one of its children.
\end{proof}

%% file: AC_queries.tex
\section{Query Details}

\subsection{Topology Tree Queries} \label{app:topology_queries}

In the original paper, Frederickson described how to aggregate edge information in topology trees to find replacement edges in a minimum spanning tree~\cite{frederickson1985data}. Later, Frederickson showed how to modify topology trees to efficiently answer minimum edge weight queries along paths from any given vertex to the root, for rooted input trees~\cite{frederickson1997data}.
Here, we describe how to implement \defn{path queries} and \defn{subtree queries} in topology trees over unrooted input trees.
We also note that \defn{connectivity queries}, which ask whether a path exists between two vertices $x$ and $y$, can be answered in $O(\log n)$ time by finding and comparing the roots of the topology trees containing $x$ and $y$.
Additionally, we show that topology trees can support several other queries previously shown for rake-compress trees~\cite{acar2005experimental}.

\myparagraph{Path Queries}
A path query on an input tree $T=(V,E)$, where each edge in $E$ is assigned a value from a domain $\domain$, returns the result of applying an associative and commutative function $f : \domain^2 \rightarrow \domain$ to the values along the unique path between two connected vertices $u$ and $v$ in $V$. We assume that $f$ can be computed in $O(1)$ time.
Path queries may also be defined over values on vertices instead of edges; we focus on edge-based queries, but vertex-based queries can be implemented similarly.
Our algorithm is inspired by the approach of Acar et al. for path queries in rake-compress trees~\cite{acar2005experimental}, later clarified by Anderson~\cite{anderson2023parallel}.

We define a \defn{unary cluster} in a topology tree as a cluster with exactly one edge in $E$ with a single endpoint in the cluster. A \defn{binary cluster} has exactly two such edges.
The \defn{boundary vertices} of a cluster are defined as the vertices in the cluster incident to these crossing edges. Unary clusters always have one boundary vertex, while binary clusters may have one or two.
The \defn{cluster path} of a cluster is defined as the unique path between the boundary vertices of the cluster (see Figure~\ref{fig:cluster_path}).
Recall that a cluster in a topology tree represents a connected subgraph of the input tree $T$. Thus, the cluster path lies entirely within the cluster.
For unary or cardinality 1 clusters, the cluster path is empty, consisting of a single vertex and no edges. Only binary clusters may have a non-empty binary cluster.

\begin{figure}\centering
\includegraphics[width=0.25\textwidth]{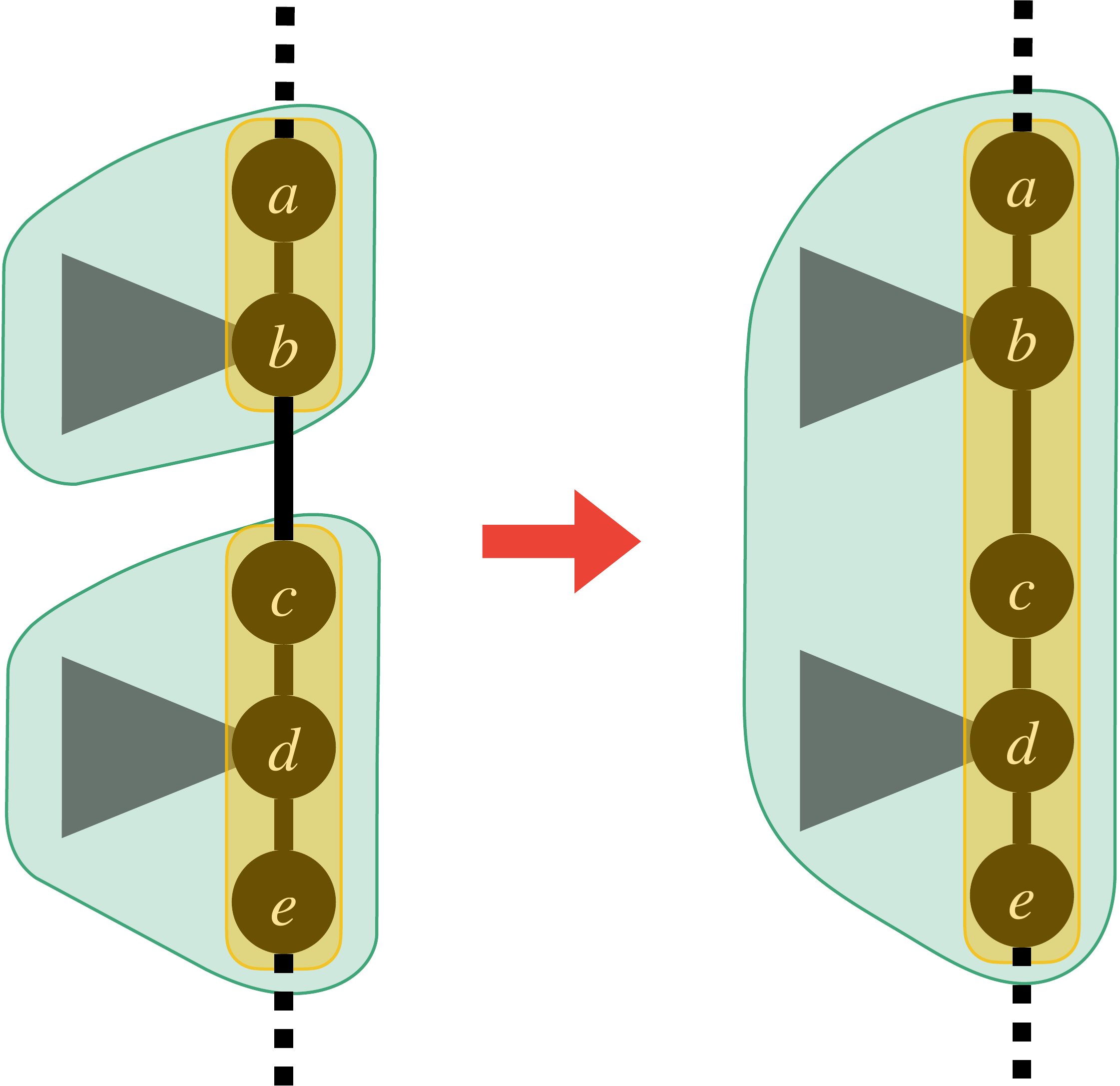}
\caption{\label{fig:cluster_path}
    An illustration of the cluster path of a binary cluster in a topology tree, and how the cluster path is constructed from the cluster paths of its two children.}
\end{figure}

To support efficient path queries, each binary cluster maintains the value of $f$ applied over all edges in its cluster path. If a cluster is not binary or if its cluster path is empty, it stores the identity element of $f$.
A binary cluster is always formed by merging two degree 2 clusters, $X$ and $Y$. Its cluster path is the union of the cluster paths of $X$ and $Y$, plus the edge between them (see Figure~\ref{fig:cluster_path}).
During the re-clustering step of an update, the value stored in the new binary cluster can be computed in $O(1)$ time by applying $f$ to the values stored in $X$, $Y$, and the edge connecting them.

To answer a path query between vertices $u$ and $v$, we consider the least common ancestor $P$ of leaf clusters $u$ and $v$ in the topology tree. Let $U$ and $V$ be the direct children of $P$ containing $u$ and $v$, respectively. Then there exists an edge $e = (b_u,b_v)$ between boundary vertices of $U$ and $V$: the edge that caused $U$ and $V$ to merge and form $P$).
The path between $u$ and $v$ in the original tree is the concatenation of the path from $u$ to $b_u$, the edge $e$, and the path from $v$ to $b_v$. So if we knew the values of the paths from $u$ and $v$ to the boundary vertices of $U$ and $V$, respectively, we can answer the path query.

We define the \defn{representative path(s)} of a cluster $X$ with respect to a vertex $x\in X$ as the unique path(s) from $x$ to the boundary vertex/vertices of $X$.
A unary cluster has a single representative path, while a binary cluster has two; even if it has only one boundary vertex, we treat it as having two identical representative paths.
Next we show how to inductively construct the representative path(s) of a cluster $X$ with respect to a vertex $x$.
Starting at the leaf cluster for $x$, we inductively compute the representative paths of each parent cluster up to cluster $X$. If a cluster $P$ has only one child $C$, then the representative path(s) are the same.
Otherwise, we consider the four inductive cases, as depicted in Figure~\ref{fig:path_induction}, involving the sibling cluster $S$ of $C$ (that it combined with to form $P$), and the edge $e$ connecting $C$ and $S$.

\begin{figure*}\centering
\includegraphics[width=\textwidth]{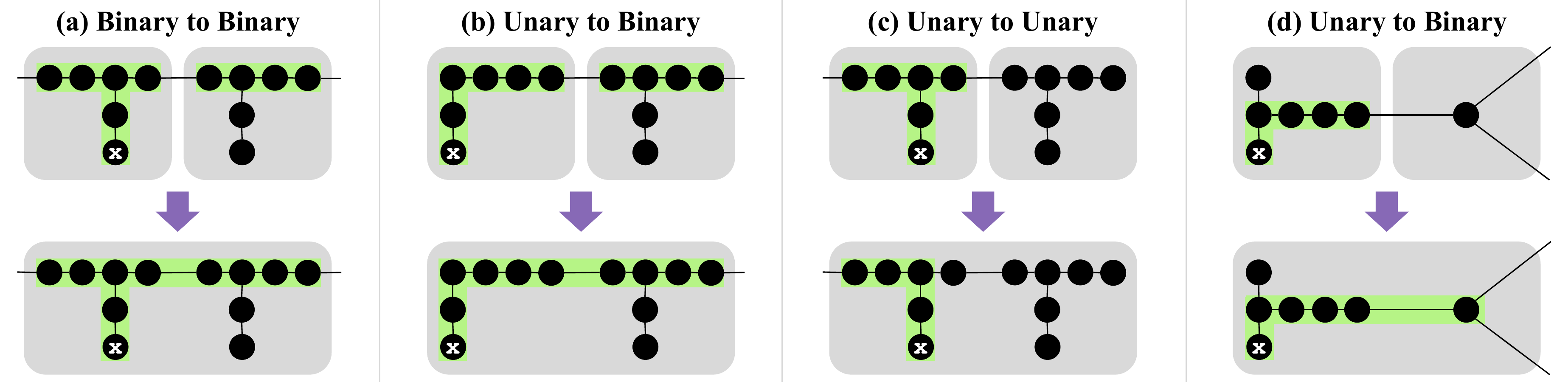}
\caption{\label{fig:path_induction}
    An illustration of the four inductive cases for constructing the solution to a path query in a topology tree. The four cases are: (a) a binary cluster with a binary parent, (b) a unary cluster with a unary parent, (c) a binary cluster with a unary parent, and (d) a unary cluster with a binary parent.}
\end{figure*}

\begin{enumerate}[label=(\alph*),topsep=0pt,itemsep=0pt,parsep=0pt,leftmargin=15pt]
\item The first case is if $C$ and $P$ are both unary clusters. This means $C$ combined with a binary cluster $S$. The representative path of $P$ is the union of the representative path of $C$, the edge $e$, and the cluster path of $S$.
\item The second case is if $C$ is unary and $P$ is binary. This is only possible if $S$ is a degree 3 cluster of cardinality 1. So the representative paths of $P$ are both the union of the representative path of $C$ and the edge $e$.
\item The third case is if $C$ is binary and $P$ is unary. In this case $S$ must have been a unary cluster. The representative path of $P$ is simply the representative path of $C$ that does not contain an endpoint of $e$ (or either path if $C$ has one boundary vertex).
\item The fourth case is if $C$ and $P$ are both binary clusters. This means that $S$ was also binary. One representative path of $P$ is the representative path of $C$ that does not contain an endpoint of $e$ (or either path if $C$ has one boundary vertex). The other representative path is the union of the representative path of $C$ containing an endpoint of $e$, the edge $e$, and the cluster path of $S$. 
Finally there is one additional case when starting at a degree 3 leaf cluster. Then its parent is a binary cluster with both representative paths empty.
\end{enumerate}

Using this inductive method and the stored values of $f$ applied over the cluster path of each binary cluster, we can compute the values of $f$ over the representative path(s) of $U$ with respect to $u$ and of $V$ with respect to $v$, thus allowing us to answer the path query. Since each inductive step can be computed in $O(1)$ time and the height of the tree is $O(\log n)$, we have the following lemma:

\begin{lemma}\label{lem:path_query}
    Path queries in a topology tree on an input tree with $n$ vertices take $O(\log n)$ time.
\end{lemma}

\myparagraph{Subtree Queries}
A subtree query on an input tree $T$, where each vertex is assigned some value in domain $\domain$, is defined as an operation that returns the value of some associative and commutative function $f : \domain^2 \rightarrow \domain$ applied over all the values of the vertices of the subtree rooted at $v$ with respect to a parent vertex $p$.
Alternatively, a subtree query my be defined in terms of the edges in a subtree. We focus on the former definition but note that the latter may be implemented in a similar manner.
Once again, our algorithm takes inspiration from that of Acar et al. for subtree queries in rake-compress trees~\cite{acar2005experimental}, which was later clarified by Anderson~\cite{anderson2023parallel}.

To efficiently support subtree queries, each cluster in the topology tree will maintain the value of $f$ applied over all of the vertices in the cluster.
Each leaf in the tree will contain the identity value of $f$. When re-clustering during updates, the value in any newly formed or modified clusters will be set to the result of applying $f$ on the values of the children of that cluster, which can be computed in constant time.

To answer a subtree query, we will trace the topology tree from the bottom up following the paths from $v$ and $p$ to the root. In fact, we will start at the LCA of $v$ and $p$.
Note that the LCA of $v$ and $p$ is the result of combining two clusters $V$ and $P$ along the edge $(v,p)$. The cluster $V$ represents a connected subgraph of the input tree including $v$ but not $p$. Therefore, each vertex in $V$ is in the subtree of $v$ rooted at $p$.
Our query will start with the value stored by $V$.

Let $X$ be the LCA cluster and $b_v$ be the boundary vertex of $X$ on the side of $v$ (if there is no such boundary vertex then the query is complete).
At this point we have the subtree rooted at $v$ minus the subtree rooted at $b_v$. We define this as the \defn{representative subtree} of cluster $X$ with respect to vertex $v$ and parent $p$.
As we traverse up the rest of the path to the root, we can inductively compute the representative subtree of the parent cluster. Consider a parent cluster $P$ that is the combination of $X$ and another cluster $C$. If $C$ and $X$ are combined through $b_v$, then the representative subtree of $P$ is the union of the representative subtree of $X$ and the set of vertices contained in $C$, and the value of $b_v$ is the other boundary vertex of $C$ if it exists. Else if $C$ and $X$ are not combined through $b_v$ the representative subtree of $P$ is just the representative subtree of $X$, and the value of $b_v$ is the same as before.
If $b_v$ no longer exists, then the query is complete.

Since each inductive step can be computed in $O(1)$ time, and the height of the tree is $O(\log n)$, we have the following lemma.

\begin{lemma}\label{lem:subtree_query}
    Subtree queries in a topology tree on an input tree with $n$ vertices take $O(\log n)$ time.
\end{lemma}

\myparagraph{Additional Queries}
Topology trees can also support diameter, lowest common ancestor, nearest marked vertex, and center/median queries. These algorithms are essentially the those of rake-compress trees~\cite{acar2005experimental}.
We briefly summarize these algorithms to show that they also apply to topology trees with little modification.

\defn{Diameter queries} return the length of the longest path in the input tree.
To support this each unary cluster maintains the diameter within the cluster and the length of the longest path in the cluster starting at the boundary vertex.
Each binary cluster maintains its diameter, the length of the longest path starting at both boundary vertices, and the length of its cluster path.
The answer to a diameter query can then be read from the root node of the topology tree.

\defn{Nearest marked vertex queries} take a query vertex $v$ and return the nearest vertex in a predetermined set of marked vertices. Here we describe how to answer the distance to the nearest marked vertex.
Each cluster maintains the distance from each boundary vertex to the nearest marked vertex in the cluster. Additionally binary clusters store the length of its cluster path.
To return the nearest marked vertex itself it is possible to augment each cluster with the nearest vertex itself in addition to the distance.
The nearest marked vertex query can then be answered by traversing up from the leaf for vertex $v$. For each cluster on the path to the root we compute the minimum distance from $v$ to a marked vertex in the cluster and the distance from $v$ to each boundary vertex.

\defn{Center queries} return the vertex which minimizes the maximum distance to any other vertex in the tree.
\defn{Median queries} return the vertex which minimizes the weighted (with respect to vertex weights) sum of the distances to all other vertices.
Here we describe how to support center queries, median queries can be supported quite similarly.
Each cluster maintains the maximum distance from each boundary vertex to any other vertex in the cluster. Additionally binary clusters store the length of its cluster path.

Then the query can be answered by traversing down from the root of the tree. For each cluster formed from the combination of two clusters $C_1$ and $C_2$ through the edge $(c_1,c_2)$, the algorithm traverses down into the cluster where the maximum distance from boundary vertex $c_1$ or $c_2$ to the rest of the tree is higher.
To get the maximum distance from $c_1$ to the rest of the tree, the maximum distance from the other boundary vertex of $C_1$ to the rest of the tree is maintained for each cluster during downward traversal. Then the maximum distance from $c_1$ to the rest of the tree is the maximum of the maximum distance from $c_1$ to vertices in $C_1$ and the sum of the length of the cluster path of $C_1$ and the maximum distance out of the other boundary vertex of $C_1$.

\defn{Lowest common ancestor (LCA) queries} take two query vertices $u$ and $v$ and a root $r$ and return the LCA of $u$ and $v$ with respect to $r$.
Note that $u$, $v$, and $r$ are completely interchangeable (e.g. the LCA of $u$ and $v$ with respect to $r$ is the same as the LCA of $u$ and $r$ with respect to $v$), and the LCA must be a degree 3 vertex.
The query is answered by traversing up from the leaves of $u$, $v$, and $r$. Without loss of generality, assume the paths of $u$ and $v$ intersect first. Let $X$ be the first cluster containing $u$ and $v$.
If $X$ is the combination of a degree 3 and degree 1 cluster, the LCA is the single vertex in the degree 3 cluster.
Else if $X$ is the combination of a degree 2 and a degree 1 cluster we traverse down the path of the degree 2 cluster. The LCA is the first degree 3 cluster that combines with a cluster down this path.
Else $X$ is the combination of two degree 2 clusters. Then we traverse down the path of the degree 2 cluster on the side of $r$ (this can be determined by traversing up until the path from $r$ intersects with the paths from $u$ and $v$). The LCA is the first degree 3 cluster that combines with a cluster down this path.
If $X$ is the combination of two degree 1 clusters then $u$, $v$, and $r$ are not connected.

\subsection{UFO Tree Queries}\label{app:ufo_queries}

\myparagraph{Path Queries}
Here we use the same definitions of path queries, unary clusters, binary clusters, boundary vertices, and cluster path as in Appendix~\ref{app:topology_queries}.
To account for high degree clusters which are neither unary nor binary, we define a \defn{superunary cluster} in a UFO tree as a cluster with one boundary vertex, but strictly greater than one edge in $E$ with exactly one endpoint in the cluster.
Superunary clusters may consist of a single degree $d \geq 2$ vertex, or they are the result of combining a high degree cluster with one or more neighboring degree 1 clusters.
We note that a cluster consisting of a single degree 2 vertex may be considered as a superunary or binary cluster, it does not matter in the analysis.
By definition, the cluster path of any superunary clusters is empty.

Each binary cluster will maintain the value of $f$ applied over all of the edges in its cluster path. The cluster path of unary and superunary clusters is empty so no value is stored.
These values are only recomputed when a combination between two binary clusters occurs during the re-clustering. The process is identical to topology trees (see Figure~\ref{fig:cluster_path}).

\begin{figure*}\centering
\includegraphics[width=0.7\textwidth]{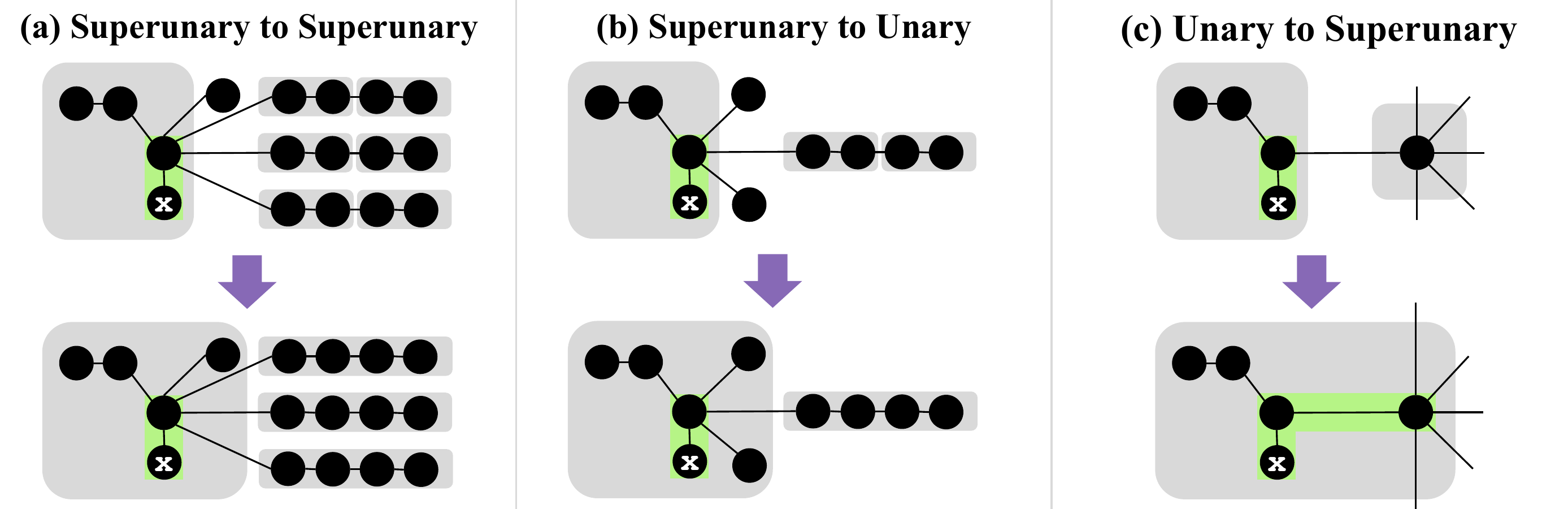}
\caption{\label{fig:ufo_induction}
    An illustration of the three additional inductive cases for constructing the solution to a path query in a UFO tree. The three cases are: (a) a superunary cluster with a superunary parent, (b) a superunary cluster with a unary parent, and (c) a unary cluster with a superunary parent.}
\end{figure*}

The representative path(s) of a cluster $X$ with respect to a vertex $x$ contained in the cluster is again defined as the unique path(s) from $x$ to the boundary vertex/vertices of $X$.
For unary and superunary clusters there is one representative path, and for binary clusters there are two.
To answer a path query we inductively compute the representative path(s) starting from the leaves representing the two endpoints of the path query until their LCA.
The four inductive cases from topology trees for unary and binary clusters are used again here (depicted in Figure~\ref{fig:path_induction}).
We introduce three additional inductive cases for superunary clusters (depicted in Figure~\ref{fig:ufo_induction}).

Assume we have the representative path(s) of a child cluster $C$ and we want to compute the representative path(s) of a parent cluster $P$. If $C$ is the only child of $P$, the representative path(s) are the same.
The three additional inductive cases for UFO trees are:
(a) The first case is if $C$ and $P$ are both superunary clusters. This means $C$ combined with one or more unary clusters. The representative path of $P$ is identical to that of $C$.
(b) The second case is if $C$ is superunary and $P$ is unary. This means that $C$ was degree $d$ and combined with exactly $d-1$ other unary clusters. Still the representative paths of $P$ is identical to that of $C$.
(c) The third case is if $C$ is unary and $P$ is superunary. In this case the cluster $S$ that $C$ combined with through edge $e$ to form $P$ must have been a superunary cluster. The boundary vertex of $P$ is that of $S$. Thus the representative path of $P$ is the union of the representative path of $C$ and the edge $e$.
Note that these three additional cases are sufficient since there is no way to form a superunary cluster from a binary cluster and vice versa via the rules of combination of UFO trees.

Using this inductive method and the stored values of $f$ applied over the cluster path of each binary cluster, we can compute the values of $f$ over the representative path(s) of $U$ with respect to $u$ and of $V$ with respect to $v$, thus allowing us to answer the path query. Since each inductive step can be computed in $O(1)$ time and the height of the tree is $O(\min\{\diam,\log n\})$, we have the following lemma:

\begin{lemma}\label{lem:ufo_path_query}
    Path queries in a UFO tree on an input tree with $n$ vertices and diameter $d$ take $O(\min\{\diam,\log n\})$ time.
\end{lemma}

\myparagraph{Subtree Queries With Invertible Functions}
Here we use the same definitions of subtree queries, boundary vertices, and representative subtrees as Appendix~\ref{app:topology_queries}.
To efficiently support subtree queries, each cluster in the UFO tree will maintain the value of $f$ applied over all of the vertices in the cluster.
Each leaf in the tree will contain the identity value of $f$. 

When re-clustering during updates, the value in any newly formed or modified clusters must be recomputed to be the result of applying $f$ on the values of all the children of that cluster.
In topology trees this can always be computed in constant time since there are at most two children.
For high fan-out clusters in a UFO tree, if the update adds a cluster to the combination, the new value is simply $f$ applied on the old value and the value of the added child cluster.
The challenge arises when a cluster is removed from the combination of a high fan-out cluster. In this case, if $f$ is an invertible function the new value can be computed in constant time as the inverse function of $f$ applied on the old value and the value of the removed child cluster.

To answer a subtree query, we will trace the UFO tree from the bottom up following the paths from $v$ and $p$ to the root. Starting at the child of the LCA of $v$ and $p$ which contains $v$ continue up and inductively compute the value of $f$ applied over the representative subtree of each cluster with respect to $v$ and $p$.

The inductive cases for unary and binary clusters are identical to those in topology trees.
The additional inductive case is if the parent of the current cluster $X$ is a superunary cluster $P$.
If $X$ is the high degree child of $P$, then $b_v$ would remain the same. Otherwise $b_v$ becomes the boundary vertex of the high degree child of $P$.
Whether or not $X$ is the high degree child of $P$ the representative subtree of $P$ is the union of the representative subtree of $X$ and the sets of vertices contained in all other children of $P$. In other words, it is $P \setminus (X \setminus R_X)$, where $R_X$ indicates the representative subtree of $X$.
Accordingly we can compute the value of $f$ applied to this set of vertices using two calls to the inverse function of $f$.

Since each inductive step can be computed in $O(1)$ time and the height of the tree is $O(\min\{\diam, \log n\})$, we have the following lemma:

\begin{lemma}\label{lem:ufo_subtree_query}
    Subtree queries in a UFO tree on an input tree with $n$ vertices and diameter $d$ take $O(\min\{\diam,\log n\})$ time.
\end{lemma}

\myparagraph{Lowest Common Ancestor Queries}
LCA queries may be answered similarly to topology trees.
We again use the observation that $u$, $v$, and $r$ are completely interchangeable. In a UFO tree the LCA must be a high degree vertex (not degree 1 or 2).
The query is answered by traversing up from the leaves of $u$, $v$, and $r$. Without loss of generality, assume the paths of $u$ and $v$ intersect first. Let $X$ be the first cluster containing $u$ and $v$.
If $X$ is the combination of high degree cluster and one or more degree 1 clusters, the LCA is the center vertex in the high degree cluster, that is the vertex found by traversing down the tree following the high degree child every time.
Else if $X$ is the combination of a degree 2 and a degree 1 cluster we traverse down the path of the degree 2 cluster. The LCA is the center vertex of the first high degree cluster that combines with a cluster down this path.
Else $X$ is the combination of two degree 2 clusters. Then we traverse down the path of the degree 2 cluster on the side of $r$ (this can be determined by traversing up until the path from $r$ intersects with the paths from $u$ and $v$). The LCA is the center vertex of the first high degree cluster that combines with a cluster down this path.
If $X$ is the combination of two degree 1 clusters then $u$, $v$, and $r$ are not connected.

\myparagraph{Additional Queries}
Some queries require maintaining the value of a non-invertible function over the children of a high fanout cluster. These queries include subtree queries with a non-invertible function, diameter queries, nearest marked vertex queries, center queries, and median queries.
Naively recomputing this \defn{aggregate value} when a single child of a high fanout cluster is deleted takes linear work, since the function is non-invertible.
Here we describe a technique that allows for efficiently maintaining aggregate values with logarithmic update time for deleting or inserting elements.
This technique makes the cost of UFO tree operations $O(\log n)$, losing the $O(\diam)$ bound.

The technique is to represent the set of children for a cluster as a binary search tree data structure. It is well known how to efficiently maintain binary search trees with augmented values.
Specifically, we use a rank tree~\cite{wulff2013faster} to store the set of children of a parent cluster $p$.
Each child cluster $c$'s weight in the rank tree is $n(c)$, the number of vertices in its cluster.
Rank trees guarantee that the depth of a child $c$ is at most $\log(n(p)/n(c))$. Additionally, $c$ can be inserted or deleted from the rank tree in $O(\log(n(p)/n(c))$ time.

Using this nested rank tree structure, the total depth of any leaf in the UFO tree is still $O(\log n)$, because of telescoping sum argument.
Thus the query complexity, which is proportional to the height of the tree, is $O(\log n)$.
Lemma~\ref{lem:rank_tree_ufo_update} proves that updates also cost $O(\log n)$ using this representation.

\begin{lemma} \label{lem:rank_tree_ufo_update}
    UFO trees using rank trees to represent the child sets of clusters can perform updates in $O(\log n)$ time, where $n$ is the number of vertices in the input tree.
\end{lemma}
\begin{proof}
    First note that since there are $O(1)$ root clusters at every level, we could map all changes during the update to $O(1)$ paths. We will show that each path can only do $O(\log n)$ total work for rank tree insertions or deletions.

    Consider the worst-case where the path has to do a rank tree insertion or deletion for every node on its path. The cost of all of these operations is at most $O(\log n)$ due to a telescoping sum argument. For example, let the path be $c_0, c_1, c_2,\hdots,c_{h-1},c_h$. The cost of the rank tree operations is:
    % $$O(\log(n(c_h)/n(c_{h-1}))+\hdots+\log(n(c_2)/n(c_1))+\log(n(c_1)/n(c_0))$$
    % $$= \log(n(c_h)/n(c_0))=\log(n/1).$$
    \begin{align*}
        &O\left(\log \frac{n(c_h)}{n(c_{h-1})} + \dots + \log \frac{n(c_2)}{n(c_1)} + \log \frac{n(c_1)}{n(c_0)}\right) \\
        &\quad = O\left(\log \left( \frac{n(c_h)}{n(c_{h-1})} \times \dots \times \frac{n(c_2)}{n(c_1)} \times \frac{n(c_1)}{n(c_0)} \right)\right) \\
        &\quad = O\left(\log \frac{n(c_h)}{n(c_0)}\right) = O(\log n).
    \end{align*}
\end{proof}

% The rank tree has all children as leaves. Each leaf is assigned a rank $\mathsf{rank}(c) \gets \lfloor \log n(c) \rfloor$. Initially, each leaf is seen as its own tree. The rank tree attaches two trees with the same rank together by adding a new root that is the parent of the two previous roots. This forms a new tree with rank one greater. This process repeats until all trees have distinct weights.

% \myparagraph{Maintaining Query Values With Batch Updates}
% For UFO tree queries that only store values in binary clusters or store no values (e.g. path queries and LCA queries), the analysis of batch updates does not change. This is because these clusters have at most two children thus recomputing the stored value can be done in constant work and depth.
% %
% For queries which store a value in unary and/or superunary clusters (e.g. subtree queries), the batch update may require several value updates to these clusters if they are high fan-out.
% To do this work efficiently and in low depth, a parallel reduction can be done over all of the value updates for the high fan-out cluster. This may take logarithmic depth, however since updating the values stored in clusters is not on the critical path for the update algorithm, it can be done separately in parallel for each such cluster and won't affect the overall depth of the batch update.

\subsection{Lower Bound for Subtree Queries with Non-Invertible functions}\label{app:noninvertible_queries_lower_bound}
We now show that supporting subtree minimum or maximum queries requires $\Omega(\log n)$ time per operation even on input trees of constant diameter. 

\begin{lemma}
    Any dynamic tree data structure that supports subtree minimum or maximum queries must take $\Omega(\log n)$ time per operation, amortized or worst case, even when restricted to trees of constant diameter. Otherwise, it would violate the $\Omega(n\log n)$ lower bound for comparison-based sorting.
\end{lemma}
\begin{proof}
    We reduce sorting of $n$ numbers to a sequence of operations on a dynamic tree representing a star graph. 
    
    Construct a star-shaped tree on $n+1$ vertices, i.e., with one center node and $n$ leaves, via a sequence of $n$ edge insertions where the edge weights correspond to the input numbers. Then, iteratively perform the following: query the subtree minimum (or maximum) at the center node, and delete the corresponding edge. This process essentially identifies and removes the minimum (or maximum) remaining element. After $n$ iterations, the sorted sequence of the numbers can be constructed. Note that, throughout, the tree remains a star with diameter $2$.
    
    The total cost is given by the cost of $n$ edge insertions, $n$ edge deletions, and $n$ subtree minimum (or maximum) queries. Thus, if each of these operations can be performed in $o(\log n)$ time (amortized or worst case), this would yield a sorting algorithm running in $o(n\log n)$ time, violating the comparison-based sorting lower bound.
\end{proof}

This lower bound explains why UFO trees can not achieve the stronger $O(\diam)$ cost bound for general non-invertible subtree queries. We believe that similar lower-bound techniques could be used to prove that it is not possible to achieve $o(\log n)$ time per query for other query types where UFO trees can not achieve an $O(\diam)$ cost bound.

%% file: AD_experiments.tex
\section{Experiment Setup and Additional Results}\label{app:experiments}

\subsection{Sequential Implementation Details}\label{app:seq_implementations}

\myparagraph{Existing Baselines}
There are many existing implementations of dynamic tree data structures. Here we list a few that we evaluate in our experiments. We use the sequential implementation of Euler tour trees (ETT) from Tseng~\etal~\cite{tseng2019batch} which uses a skip-list as the underlying sequence data structure.
We also use the implementation of splay top trees (STT), which we converted from C to a modern version of C++.
Both the ETT implementation and the STT implementation require the use of a map to find the leaf cluster corresponding to a given edge. We implement these using the flat hash map from Abseil~\cite{abseil}.

\myparagraph{New Topology Tree Implementation}
Each Topology cluster struct has a parent cluster pointer, and an array of 3 pointers to neighboring clusters.
Our implementation does not explicitly store children to avoid extra memory usage. This is fine because in all cases in our algorithms we can find the sibling of a cluster that shares the same parent rather having to directly look at the children of that parent.

\myparagraph{New UFO Tree Implementation}
Each UFO cluster struct has a parent cluster pointer, a fixed constant size array of pointers to neighboring clusters, and a pointer to an Abseil hash set for any additional neighbors.
Similar to topology Tree, our UFO tree implementation does not explicitly store children to avoid extra memory usage. This is fine because in all cases in our algorithms we can consider the siblings of a cluster that share a parent rather having to directly look at the children of that parent.

We keep a small fixed constant size array of neighbors so that when the struct is loaded into memory it includes a small amount of its neighbors. This prevents additional cache misses reading from the set of neighbors if we only need access to a small amount of the neighbors. This is particularly helpful because in any tree at least half of the vertices are degree 1 or degree 2.
We only store a pointer to the Abseil hash set so that clusters that can store all of their neighbors in the array do not have the space overhead of an empty hash set that is not used. Instead there is just an extra 8 byte pointer per cluster. If a cluster gains many neighbors the hash set is allocated in the heap and if it loses many neighbors it is freed. This optimization reduces memory usage significantly and also slightly speeds up updates and queries.

\myparagraph{RC-Tree Implementation}
We developed a sequential implementation of the deterministic and direct version of rake-compress trees which does not use the self-adjusting computation framework~\cite{anderson2024deterministic} and is not batch-dynamic ~\cite{ikram2025parallel}.
Our reasoning for developing our own implementation is that although the self-adjusting computation framework has much utility and can solve many problems, it uses a large amount of memory and tuning an implementation to a single problem is beneficial. 
Similarly, we found that our RC tree implementation is much faster than the batch dynamic version when running with 1 thread and batch size $=1$.
We further optimized our R.C. tree implementation to ensure a fair comparison with our performance-engineered Topology Tree and UFO tree implementations and to confirm that any speed-ups/slow-downs we saw were due to properties inherent to the data structure. 
To that end, we decided to simulate a set data structure for the set of affected vertices using a vector as the size of the set is provably lower than 312 items in the sequential case and even more so in practice. Additionally, we also attempted to minimize the amount of pointer jumps to allow for better locality and engineered our data structure accordingly. 
Some of these modifications include using fixed size arrays instead of vectors whenever helpful and minimizing lookups to the array of RC Tree clusters (preferring to store these references in a way that leads to better locality). 
We believe these changes also optimize the data structure's space usage as well and provide an accurate comparison with other dynamic trees.

\myparagraph{Ternarization Implementation}
To evaluate the impact of ternarization in practice sequentially, we developed our own implementation of an abstract dynamic ternarization interface. 
This maps any input tree to an underlying tree with degree $\le 3$, and maps the link and cut operations to the necessary links and cuts on the underlying tree.

\myparagraph{Link Cut Tree Implementation}
We use the implementation from Tseng~\etal~\cite{tseng2019batch} for Link Cut Tree (LCT). This implementation uses the amortized version of LCT implemented with splay trees.
There is one splay tree node per vertex, each containing one parent pointer, two child pointers, and a flip bit.
We modified this implementation to also support path maximum queries for integer weighted edges.
Note that vertex weighted path queries are trivial to implement, since the splay tree nodes represent a one to one correspondence to the vertices of the input tree, so we can simply augment each splay tree to store the aggregate of all nodes in its subtree. Then a path query just requires exposing the path of interest and reading the root aggregate value.

Sleator and Tarjan described how to support edge weighted path queries in the original paper~\cite{sleator1983data} for the purpose of implementing max flow algorithms. While being theoretically efficient, the implementation is somewhat complex and requires a lot of additional fields per tree node.
Our implementation only uses two extra integer fields per tree node. These store the weight of the "upward" and "downward" preferred edge incident to a vertex.
Each splay tree node then maintains an aggregate of both preferred edge weights in all nodes in its subtree.
The reason we need to store both the upward and downward edge weight (thus each preferred edge weight is stored twice) instead of just the upward weights is because when the root preferred path is reversed all of the upward weights become downward weights and vice versa.
These changes between upward and downward weights are handled with the same technique as used for flipping the left and right children of reversed splay tree nodes (keep a flip bit for each node and lazily push the flip to the children when traversing downward).

The upward weight field is also used to store the weight of the non-preferred edge out of the head of each path. This value is not included in the aggregates in the splay trees.
Fortunately paths with a non-preferred edge out of the head are never reversed (as only the root path is ever reversed), so we never have to worry about determining what the non-preferred edge out of the new head of a path is when it gets reversed (the previous tail).
When two paths are appended, we update the downward preferred edge weight of the tail of the upper path.
Similarly when we split a path we clear the downward weight in the upper node of the split edge.
Whenever a node gains or loses a preferred edge weight it is splayed to recompute the aggregates.

\subsection{Parallel Batch-Dynamic Trees Implementations}\label{app:par_implementations}

Our parallel implementations use ParlayLib, a C++ library for shared memory parallel programming. We use Microsoft's mimalloc in combination with the parallel allocator from ParlayLib for parallel memory allocation.

%For parallel UFO trees, we use the Parallel Augmented Maps (PAM) library~\cite{} to store sets of neighbors that support parallel operations.
%We use a parallel sample sort to group edges by their endpoint which is needed for batch updating adjacency lists in UFO trees.
%Rather than just using a known theoretically efficient parallel MIS algorithm, we designed and implemented a fast asynchronous MIS algorithm for chains inputs, which is used by both parallel topology trees and parallel UFO trees.

\myparagraph{Parallel Topology Trees Implementation}
For our parallel topology trees implementation we use a simple method for updating adjacency lists in parallel, which just iterates through the up to 3 neighbors and uses an atomic compare-and-swap operation on each index.

To test the topology tree on high degree inputs, we use the implementation of batch dynamic ternarization from ~\cite{ikram2025parallel}.
 
\myparagraph{Parallel UFO Trees Implementation}
Developing a performant implementation of parallel UFO trees faces many challenges.
The first and perhaps most obvious challenge, is maintaining tree neighbor sets in low depth. We use Parallel Augmented Maps (PAM) to represent the the sets of neighbors for each cluster in the UFO tree.

Since the algorithm receives a batch of updates in no particular order, one thing we must do is transform this batch into groups of updates for each vertex at level 0. Using the ``group\_by\_key'' function in ParlayLib is costly, so we developed our own faster method, ``integer\_group\_by\_key\_inplace''. As the name suggests, this takes advantage of the fact that the endpoints of our edges can be interpreted as integers. Also, instead of allocating new memory for each group, it performs the grouping in place. The method works by first creating a sequence of the edges in both directions, then integer sorting them by the first endpoint, and finally returning pointers to the beginning and end index in the sorted range where a group occurs. We found that using the parallel sample sort performed better than parallel integer sort.

In our theoretical batch-parallel algorithm we conveniently assume that each cluster maintains its set of neighbors, and its set of children. Our implementation does not store sets of children to avoid the extra space and update costs. This led to many challenges in developing a correct implementation, but the reward is much better performance. For example, how can you determine if a cluster is low fan-out (has less than 3 children), and thus should be deleted? Additionally, when you do decide to delete a level $i+1$ cluster, how do you know what its level $i$ children are that become root clusters?

Rather than just using a well-known MIS algorithm for general graphs, we designed and implemented a fast asynchronous maximal independent set algorithm for chains inputs. Our recluster step uses this algorithm to determine the merges of degree 2 clusters and degree 1 clusters neighboring degree 2 clusters. For the remaining clusters, we just have to always merge degree 1 clusters with their neighbor, or merge high degree clusters with all of their degree 1 neighbors.

\subsection{Additional Experimental Results}\label{app:results}

\myparagraph{Parallel Diameter Sweep Results}
Figure~\ref{fig:pardiamsweep} shows the results of our diameter sweep experiment for update speed in the parallel implementations.
Our results show that parallel UFO trees increase in speed with decreasing input diameter, for batch updates.
All other implementations maintain relatively stable running times as the tree diameter decreases.
When the diameter is low enough, parallel UFO trees are consistently faster than parallel ETTs, and are the fastest implementation.

Similar to the sequential setting, we observe that batch updates in the parallel topology tree and parallel RC tree implementations become slower on lower diameter inputs. 
This is due to ternarization overheads that become more costly on inputs with many high degree nodes (which is necessary for low diameter).

\begin{figure}[]
    \centering
    \includegraphics[width=\linewidth]{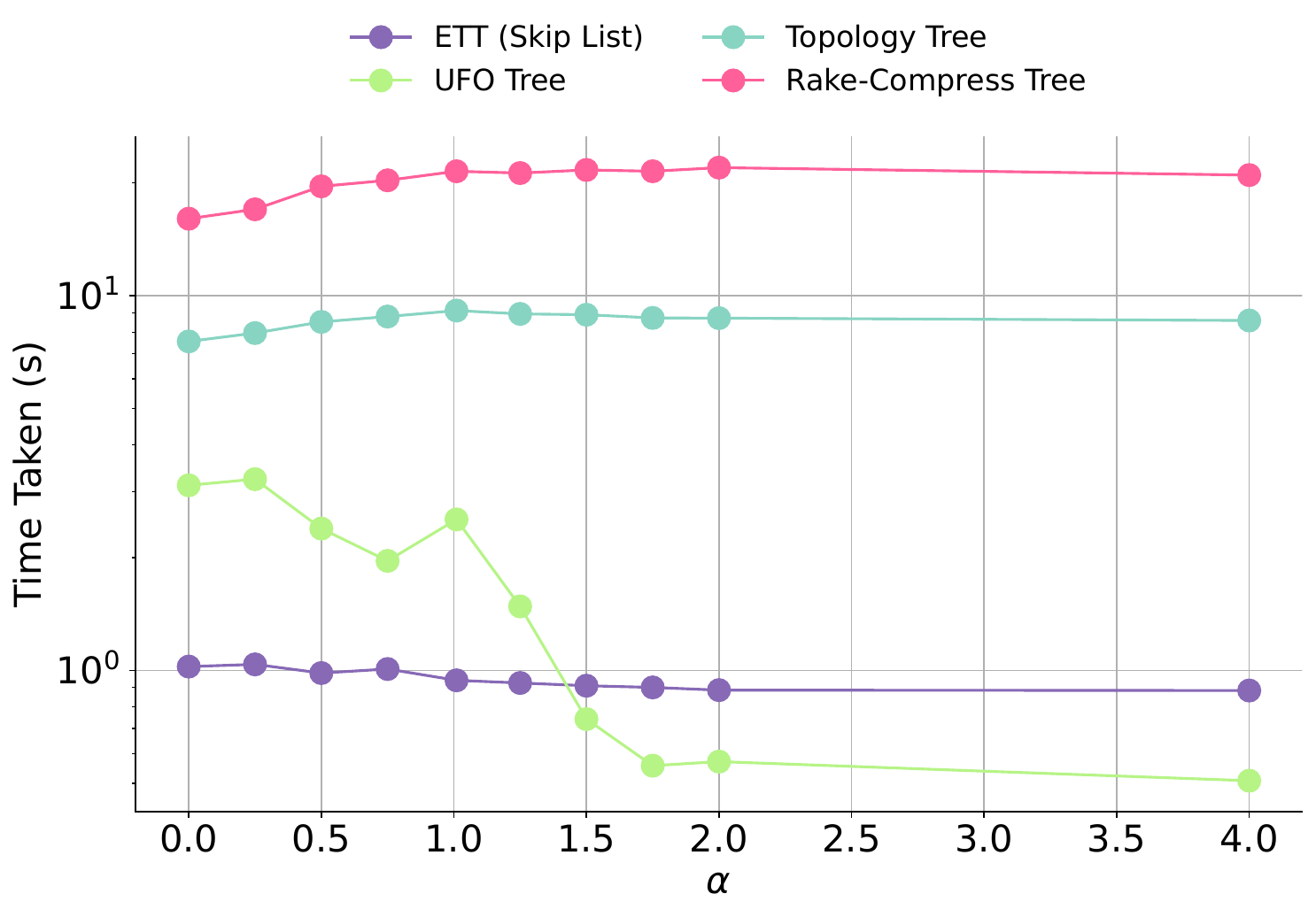}
    \caption{The results of running our diameter sweep experiment on all batch-dynamic trees with $n=10^7$.}
    \label{fig:pardiamsweep}
\end{figure}